\documentclass[runningheads]{llncs}

\usepackage{algorithm}
\usepackage{algpseudocode}
\usepackage{amsmath}
\usepackage{comment}
\usepackage{graphicx}
\usepackage{gensymb}
\usepackage{color}
\usepackage{hyperref}
\usepackage{placeins}
\usepackage{mathtools}
\usepackage{subcaption}

\newtheorem{observation}{Observation}

\begin{document}

\title{Centralised Connectivity-Preserving Transformations by Rotation: 3 Musketeers for all Orthogonal Convex Shapes}

\titlerunning{Transformations by Rotation: 3 Musketeers for Orthogonal Convex Shapes}

\author{Matthew Connor \and Othon Michail}

\authorrunning{M. Connor and O. Michail}

\institute{Department of Computer Science, University of Liverpool, United Kingdom\\
\email{M.Connor3@liverpool.ac.uk, Othon.Michail@liverpool.ac.uk}}

\maketitle

\begin{abstract}
We study a model of programmable matter systems consisting of $n$ devices lying on a 2-dimensional square grid, which are able to perform the minimal mechanical operation of rotating around each other. The goal is to transform an initial shape A into a target shape B. We are interested in characterising the class of shapes which can be transformed into each other in such a scenario, under the additional constraint of maintaining global connectivity at all times. This was one of the main problems left open by $[$Michail \emph{et al.}, JCSS'19$]$. Note that the considered question is about structural feasibility of transformations, which we exclusively deal with via centralised constructive proofs. Distributed solutions are left for future work and form an interesting research direction. Past work made some progress for the special class of \emph{nice} shapes. We here consider the class of \emph{orthogonal convex shapes}, where for any two nodes $u, v$ in a horizontal or vertical line on the grid, there is no empty cell between $u$ and $v$. We develop a generic centralised transformation and prove that, for any pair $A$, $B$ of colour-consistent orthogonal convex shapes, it can transform $A$ into $B$. In light of the existence of \emph{blocked} shapes in the considered class, we use a minimal 3-node \emph{seed} to trigger the transformation. The running time of our transformation is an optimal $O(n^2)$ sequential moves, where $n=|A|=|B|$. We leave as an open problem the existence of a universal connectivity-preserving transformation with a small seed. Our belief is that the techniques developed in this paper might prove useful to answer this.

\keywords{Programmable matter  \and Transformation \and Reconfigurable robotics \and Shape formation \and Centralised algorithms}
\end{abstract}

\section{Introduction}\label{sec1}

Programmable matter refers to matter which can change its physical properties in an algorithmic way. This means that the change is the result following the procedure of an underlying program. The implementation of this program can either be a system-level external centralised algorithm or an internal decentralised algorithm executed by the material itself. The model for such systems can be further refined to specify properties that are relevant to real-world applications, for example connectivity, colour \cite{CLSL11} and other physical properties.

As the development of these systems continues, it becomes increasingly necessary to develop theoretical models which are capable of describing and explaining the emergent properties, possibilities and limitations of such systems in an abstract and fundamental way. To this end, models have been developed for programmable matter. For example, algorithmic self-assembly \cite{D12,R06} focuses on programming molecules like DNA to grow in a controllable way, and the Abstract Tile Assembly Model \cite{RW00,W98}, the Kilobot model \cite{RCN14}, the Robot Pebbles system \cite{GKR10}, and the nubot model \cite{WCGD13}, have all been developed for this area. Network Constructors \cite{MS16} is an extension of population protocols \cite{AADFP06} that allows for network formation and reconfiguration.

The latter model is formally equivalent to a restricted version of chemical reaction networks, which ``are widely used to describe information processing occurring in natural cellular regulatory networks'' \cite{SCWB08,D13}. The CATOMS system \cite{TPB19,TPB20,FKK21} is a further implementation which constructs 3D shapes by first creating a ``scaffolding structure" as a basis for construction. Finally, there is extensive research into the amoebot model \cite{DDGR14,DDGP18,DGRS15,DGRS16}, where finite automata on a triangular lattice follow a distributed algorithm to achieve a desired goal, including a recent extension \cite{FPS22} to a circuit-based model.

Recent progress in this direction has been made in previous papers, for example \cite{MSS19}, covering questions related to a specific model of programmable matter where nodes exist in the form of a shape on a 2D grid and are capable of performing two specific movements: rotation around each other and sliding a node across two other nodes. The authors investigated the problem of transformations with rotations with the restriction that shapes must always remain connected (RotC-Transformability), and left universal RotC-Transformability as an open problem. They hinted at the possibility of universal transformation in an arxiv draft \cite{MSS17}. To the best of our knowledge, progress on this open question has only been made in \cite{CMP21}, where, by using a small seed, connectivity-preserving transformations by rotation were developed for a restricted class of shapes. In general, such transformations are highly desirable due to the large numbers of programmable matter systems which rely on the preservation of connectivity and the simplicity of movement, which is not only of theoretical interest but is also more likely to be applicable to real-world systems. Related progress was also made in \cite{AADD21}, which used a similar model but with a different type of movement. The authors allowed for a greater range of movement, for example ``leapfrog" and ``monkey" movements. They accomplished universal transformation in $O(n^2)$ movements using a ``bridging" procedure assisted by at most 5 seed-nodes, which they called \emph{musketeers}.

\section{Contribution}\label{sec2}

We investigate the RotC-Transformability problem, introduced in \cite{MSS19}, which asks to characterise which families of connected shapes can be transformed into each other via rotation movements without breaking connectivity. The model represents programmable matter on a 2D grid which is only capable of performing rotation movements, defined as the $90\degree$ rotation of a node $u$ around a neighbouring edge-adjacent node $v$, so long as the goal and intermediate cells are empty. As our focus is on the feasibility and complexity of transformations, our approach is naturally based on structural characterisations and \emph{centralised} procedures. Structural and algorithmic progress is expected to facilitate more applied future developments, such as distributed implementations.

We assume the existence of a \emph{seed}, a group of nodes in a shape $S$ which are placed in empty cells neighbouring a shape $A$ to create a new connected shape which is the unification of $S$ and $A$.
Seeds allow shapes which are blocked or incapable of meaningful movement to perform otherwise impossible transformations.
The use of seeds was established in \cite{MSS19}, leaving open the problem of universal RotC-Transformability. Another work \cite{CMP21} investigated this problem in the context of \emph{nice} shapes, first defined in \cite{AMP20} as a set of shapes containing any shape $S$ which has a central line $L$, where, for all nodes $u \in S$, either $u \in L$ or $u$ is connected to $L$ by a line of nodes perpendicular to $L$. Universal reconfiguration in the context of connectivity-preserving transformations using different types of movement has been demonstrated in \cite{AADD21}. That paper calls the seed nodes ``musketeers'' and their transformation requires the use of 5 such nodes. 

The present paper moves towards a solution which is based on connectivity-preservation and the tighter constraints of rotation-only movement of \cite{MSS19} while aiming to (i) widen the characterization of the class of transformable shapes and (ii) minimise the seed required to trigger those transformations. By achieving these objectives for orthogonal convex shapes, we make further progress towards the ultimate goal of an exact characterisation (possibly universal) for seed-assisted RotC-Transformability.

We study the transformation of shapes of size $n$ with \emph{orthogonal convexity} into other shapes of size $n$ with the same property, via the canonical shape of a \emph{diagonal line-with-leaves}. Orthogonal convexity is the property that for any two nodes $u, v$ in a horizontal or vertical line on the grid, there is no empty cell between $u$ and $v$.
A diagonal line-with-leaves is a group of components, the main being a series of 2-node columns where each column is offset such that the order of the nodes is equivalent to a line, and two optional components: two 1-node columns on either end of the shape and additional nodes above each column, making them into 3-node columns.

We show that transforming a orthogonal convex shape of $n$ nodes into a diagonal line-with-leaves is possible and can be achieved by $O(n^2)$ moves using a 3-node seed. This bound on the number of moves is optimal for the considered class, due to a matching lower bound from \cite{MSS19} on the distance between a line and a staircase, both of which are orthogonal convex shapes. A seed is necessary due to the existence of blocked orthogonal convex shapes, an example being a rhombus. As \cite{CMP21} shows, any seed with less than $3$ nodes is incapable of non-trivial transformation of a line of nodes. Since a line of nodes is orthogonal convex, the 3-node seed employed here is minimal.

The class of orthogonal convex shapes cannot easily be compared to the class of nice shapes. A diagonal line of nodes in the form of a staircase belongs to the former but not the latter. Any nice shape containing a gap between two of its columns is not a orthogonal convex shape. Finally, there are shapes like a square of nodes which belong to both classes. Nevertheless, the nice shapes that are not orthogonal convex have turned out to be much easier to handle than the orthogonal convex shapes that are not nice. We hope that the methods we had to develop in order to deal with the latter class of shapes, will bring us one step closer to an exact characterisation of connectivity-preserving transformations by rotation.

In Section \ref{sec3}, we formally define the programmable matter model used in this paper.
Section \ref{sec4} presents some basic properties of orthogonal convex shapes and of their elimination and generation sequences.
In Section \ref{sec5}, we provide our algorithm for the construction of the diagonal line-with-leaves which, through reversibility, can be used to construct other orthogonal convex shapes and give time bounds for it.
In Section \ref{sec6}, we conclude and give directions for potential future research.

\section{Model}\label{sec3}

We consider the case of programmable matter on a 2D grid, with each position (or cell) of the grid being uniquely referred to by its $(x, y)$ coordinates. Such a system consists of a set $V$ of $n$ nodes. Each node may be viewed as a spherical module fitting inside a cell of the grid. At any given time, each node occupies a cell, with the positioning of the nodes defining a shape, and no two nodes may occupy the same cell. It also defines an undirected \emph{neighbouring relation} $E\subset V\times V$, where $uv\in E$ iff $u$ and $v$ occupy \emph{horizontally} or \emph{vertically} adjacent cells of the grid. A shape is \emph{connected} if the graph induced by its neighbouring relation is a connected graph.

In general, shapes can \emph{transform} to other shapes via a sequence of one or more movements of individual nodes.
We consider only one type of movement: \emph{rotation}. In this movement, a single node moves relative to one or more neighbouring nodes. A single rotation movement of a node $u$ is a 90° rotation of $u$ around one of its neighbours. Let $(x, y)$ be the current position of $u$ and let its neighbour be $v$ occupying the cell $(x, y-1)$ (i.e., lying below $u$). Then $u$ can rotate $90\degree$ clockwise (counterclockwise) around $v$ iff the cells $(x+1, y)$ and $(x + 1, y - 1)$ ($(x-1, y)$ and $(x - 1, y - 1)$, respectively) are both empty. By rotating the whole system $90\degree$, $180\degree$, and $270\degree$, all possible rotation movements can be defined.

Let $A$ and $B$ be two connected shapes. We say that $A$ transforms to $B$ via a rotation $r$, denoted $A \stackrel{r}{\rightarrow} B$, if there is a node $u$ in $A$ such that if $u$ applies $r$, then the shape resulting after the rotation is $B$. We say that $A$ transforms in one step to $B$ (or that $B$ is reachable in one step from $A$), denoted $A \rightarrow B$, if $A \stackrel{r}{\rightarrow} B$ for some rotation $r$. We say that $A$ transforms to $B$ (or that $B$ is reachable from $A$) if there is a sequence of shapes $A = S_1, S_2, \ldots , S_t = B$, such that $S_i \rightarrow S_{i+1}$ for all $1 \leq i \leq t-1$. Rotation is a reversible movement, a fact that we use in our results. All shapes $S_1, S_2, \ldots, S_t$ must be \emph{edge connected}, meaning that the graph defined by the neighbouring relation $E$ of all nodes in any $S_i$, where $1 \leq i \leq t$, must be a connected graph.

At the start of the transformation, we will be assuming the existence of a \emph{seed}: a small connected shape $M$ placed on the perimeter of the given shape $S$ to trigger the transformation. This is essential because under rotation-only there are shapes $S$ that are $k$-blocked, meaning that at most $k\geq 0$ moves can be made before a configuration is repeated. When $k=0$, no move is possible from $S$, an example of $0$-blocked shape being the rhombus. (see, e.g., Fig. \ref{fig:blocked}).

For the sake of providing clarity to our transformations, we say that every cell in the 2D grid has a colour from $\{red, black\}$ in such a way that the cells form a black and red checkered colouring of the grid, similar to the colouring of a chessboard. This colouring is fixed so long as there is at least one node on the grid. This represents a property of the rotation movement, which is that any given node in a coloured cell can only enter cells of the same colour. We define $c(u)\in\{black,red\}$ as the colour of node $u$ for a given chessboard colouring of the grid. We represent this in our figures by colouring the nodes red or black. See Figure \ref{fig:Rot-abbrev} for an example and for special notation that we use to abbreviate certain rotations which we perform throughout the paper.

\begin{figure}
\centering
\includegraphics[width=0.6\textwidth]{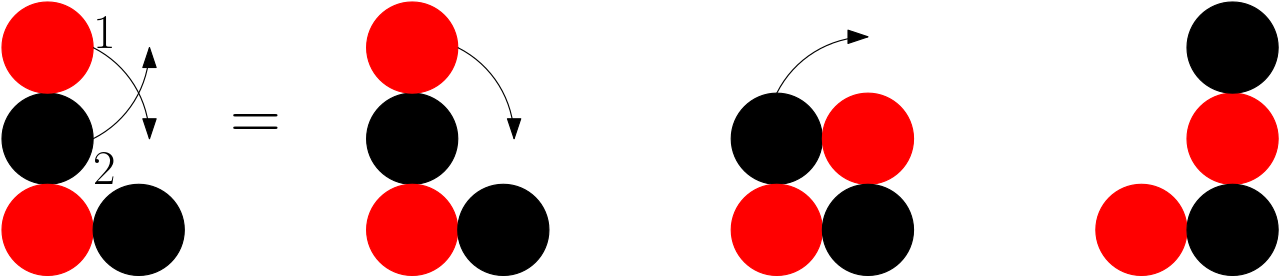}
\caption{The rotation on the left is an abbreviated version of the rotations on the right, used throughout the paper. The numbers represent the order of rotations. Reds appear grey in print, throughout the paper.}
\label{fig:Rot-abbrev}
\end{figure}

\FloatBarrier
\begin{figure}
\centering
\includegraphics[width=0.45\textwidth]{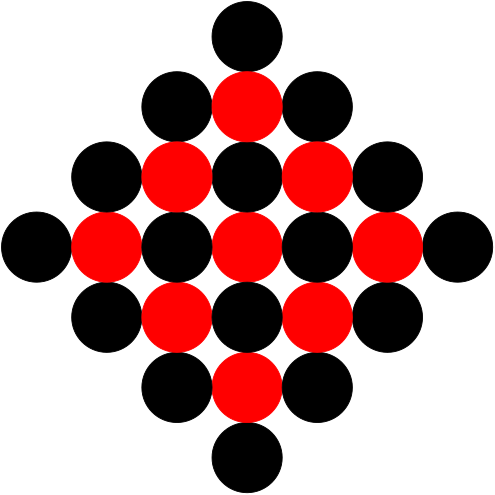}
\caption{An example of a black parity rhombus.}
\label{fig:blocked}
\end{figure}
\FloatBarrier

Any shape $S$ may be viewed as a coloured shape consisting of $b(S)$ blacks and $r(S)$ reds.
Two shapes A and B are \emph{colour-consistent} if $b(A) = b(B)$ and $r(A) = r(B)$.
For any shape $S$ of $n$ nodes, the \emph{parity} of $S$ is the colour of the majority of nodes in $S$. If there is no strict majority, we pick any as the parity colour. We use \emph{non-parity} to refer to the colour which is not the parity.

Depending on the context and purpose, the term \emph{node} will be used to refer both to the actual entity that may move between co-ordinates and to the co-ordinates of that entity at a given time.

\section{Preliminaries}\label{sec4}

\subsection{General Geometric Definitions}

We now define specific sections of the shape which we will refer to extensively throughout the paper.

\begin{definition} \label{def:perimeter}
Let $A$ be a connected shape. Mark each cell of the grid that is occupied by a node of $A$. A cell $(i, j)$ is part of a hole of $A$ if every infinite length single path starting from $(i, j)$ (moving only horizontally and vertically) necessarily goes through a black cell. Mark every cell that is part of a hole of $A$ as well, to obtain a compact shape of marked cells $A'$. Consider now polygons defined by unit-length line segments of the grid. Define the \emph{perimeter} of $A$ as the minimum-area such polygon that completely encloses $A'$ in its interior. The fact that the polygon must have an interior and an exterior follows directly from the Jordan curve theorem \cite{J93}.
\end{definition}

\begin{definition}\label{def:cell-perimeter}
Any cell of the grid that has contributed at least one of its line segments to the perimeter and has not been marked (i.e., is not occupied by a node of $A$) is the cell perimeter of shape $A$. See Figure \ref{fig:perimeter} for an example.
\end{definition}

\begin{definition} \label{def:surface}
The external surface of a connected shape $A$, is a shape $B$, not necessarily connected, consisting of all nodes $u \in A$ such that $u$ occupies a cell defining at least one of the line segments of $A$’s perimeter.
\end{definition}

\begin{definition} \label{def:extended-surface}
The extended external surface of a connected shape $A$, is defined by adding to $A$’s external
surface all nodes of $A$ whose cell shares a corner with $A$’s perimeter.
\end{definition}

\begin{figure}
\centering
\includegraphics[width=0.5\textwidth]{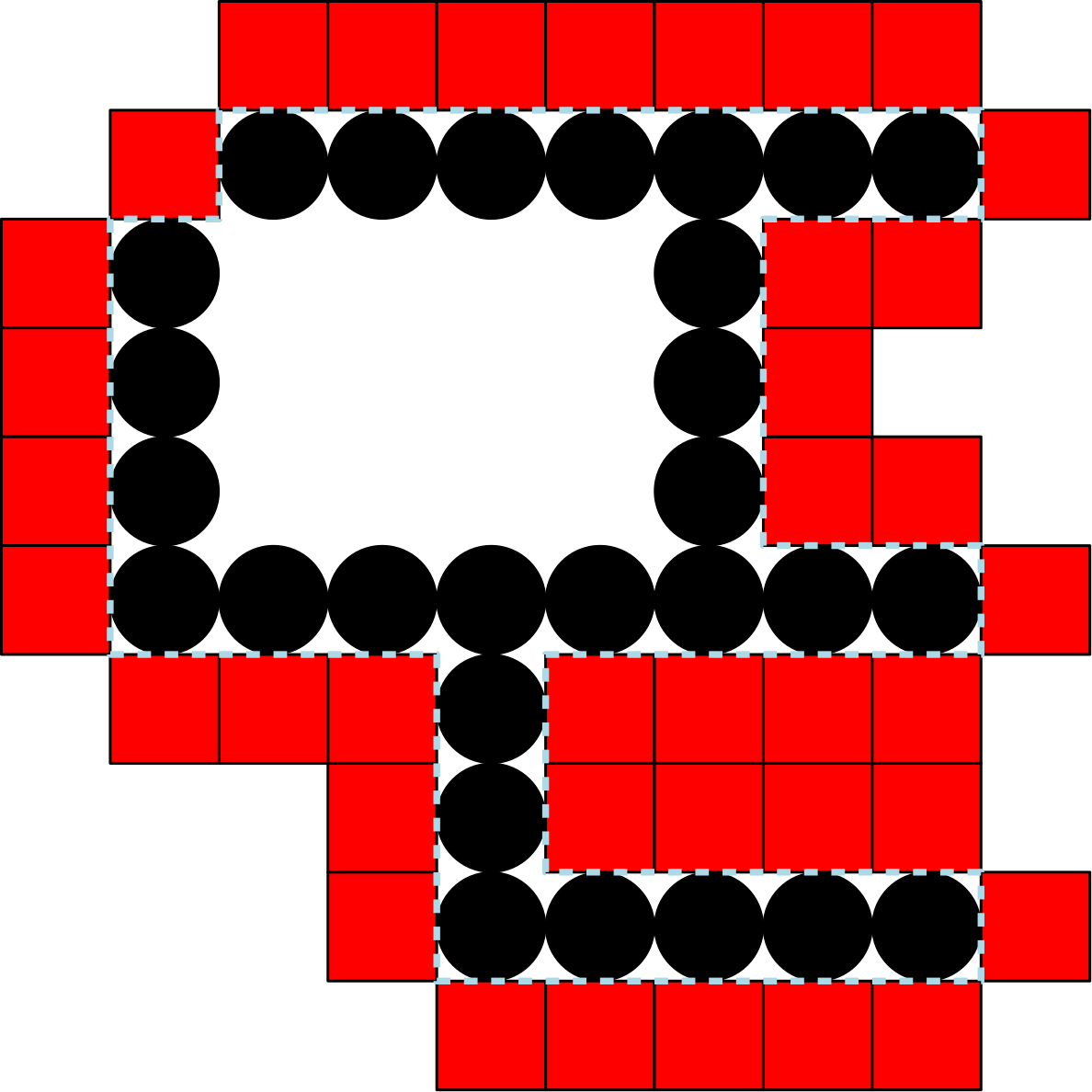}
\caption{The light-blue dashed line is the perimeter of the shape. The red squares are the cell perimeter, or the set of empty cells which contribute at least one side to the perimeter. All nodes which share a side with the cell perimeter are part of the exterior, and all cells enclosed by the exterior constitute the interior.}
\label{fig:perimeter}
\end{figure}

\subsection{Orthogonal Convex Shapes}

We now present the class of shapes considered in this paper together with some basic properties about them that will be useful later. 

\begin{definition} \label{def:orthogonal-convex}
A shape $S$ is said to belong to the family of \emph{orthogonal convex} shapes, if, for any pair of distinct nodes $(x_1,y_1),(x_2,y_2)\in S$, $x_1=x_2$ implies $(x_1,y)\in S$ for all $\min\{y_1,y_2\}<y<\max\{y_1,y_2\}$ while $y_1=y_2$ implies $(x,y_1)\in S$ for all $\min\{x_1,x_2\}<x<\max\{x_1,x_2\}$.
\end{definition}

\begin{observation}
Any discrete convex shape $S$ is also orthogonal convex.
\end{observation}

Observe now that the perimeter of any connected shape is a cycle drawn on the grid, i.e., a path where its end meets its beginning. The cycle is drawn by using consecutive grid-edges of unit length, each being characterized by a direction from $\{up, right, down, left\}$. For each pair of opposite directions, $(up,down)$ and $(left,right)$, the perimeter always uses an equal number of edges of each of the two directions in the pair and uses every direction at least once. For the purposes of the following proposition, let us denote $\{up, right, down, left\}$ by $\{d_1, d_2, d_3, d_4\}$, respectively. The perimeter of a shape can then be defined as a sequence of moves drawn from $\{d_1, d_2, d_3, d_4\}$, w.l.o.g. always starting with a $d_1$. Let also $N_i$ denote the number of times $d_i$ appears in a given perimeter.

\begin{proposition} \label{prop:regular-expression}
A shape $S$ is a connected orthogonal convex shape if and only if its perimeter satisfies both the following properties:
\begin{itemize}
    \item It is described by the regular expression 
    \begin{equation*}
    d_1(d_1\;|\;d_2)^{*}d_2(d_2\;|\;d_3)^{*}d_3(d_3\;|\;d_4)^{*}d_4(d_4\;|\;d_1)^{*}
    \end{equation*}
    under the additional constraint that $N_1=N_3$ and $N_2=N_4$.
    \item Its interior has no empty cell.
\end{itemize}
\end{proposition}

\begin{proof}
We begin by considering the forward direction, starting from a connected orthogonal convex shape $S$. For the first property, the $N_i$ equalities hold for the perimeter of any shape, thus, also for the perimeter of $S$. In the regular expression, the only property that is different from the regular expressions of more general perimeters is that, for all $i\in\{1,2,3,4\}$, $d_{i-2}$, where the index is modulo 4, does not appear between the first and the last appearance of $d_i$.

Assume that it does, for some $i$.

Then $d_{i-2}$ must have appeared immediately after a $d_{i-1}$ or a $d_{i+1}$, because a $d_{i-2}$ can never immediately follow a $d_i$. If it is after a $d_{i-1}$, then this forms the expression $d_i(d_{i-1}\;|\;d_i)^{*}d_{i-1}d_{i-2}$,  
which always has $d_id_{i-1}^{+}d_{i-2}$ as a sub-expression. But for any sub-path of the perimeter defined by the latter expression, the nodes attached to its first and last edges would then contradict Definition \ref{def:orthogonal-convex}, as the horizontal or vertical line joining them goes through at least one unoccupied cell, i.e., one of the cells external to the $d_{i-1}^{+}$ part of the sub-path. The $d_{i+1}$ case follows by observing that, in this case, the sub-expression satisfied by the perimeter would be $d_{i-2}d_{i+1}^{+}d_i$, which would again violate orthogonal convexity of $S$.

The second property, follows immediately by observing that if $(x,y)$ is an empty cell within the perimeter's interior, then the horizontal line that goes through $(x,y)$ must intersect the perimeter at two distinct points, one to the left of $(x,y)$ and one to its right. Thus, these three points would contradict the conditions of Definition \ref{def:orthogonal-convex}.

For the other direction, let $S$ be a shape satisfying both properties. For the sake of contradiction, assume that $S$ is not orthogonal convex, which means that there is a line, w.l.o.g horizontal and of the form $(x_l,y),(x_l+1,y),\ldots,(x_r,y)$, where $(x_l,y)$ and $(x_r,y)$ are occupied by nodes of $S$ while $(x_l+1,y),\ldots,(x_r-1,y)$ are not. Observe first that any gap in the interior would violate the second property, thus $(x_l+1,y),\ldots,(x_r-1,y)$ must be cells in the exterior of the perimeter of $S$ and $(x_l,y)$, $(x_r,y)$ nodes on the perimeter.
There are two possible ways to achieve this: $d_3d_2^+d_1$ and $d_1d_4^+d_3$.
These combinations are impossible to create with the regular expression, thus contradicting that $S$ satisfies the properties. Similarly for vertical gaps.
It follows that any shape fulfilling the two properties must belong to the family of connected orthogonal convex shapes.
\qed
\end{proof}

\begin{figure}
\centering
\includegraphics[width=0.5\textwidth]{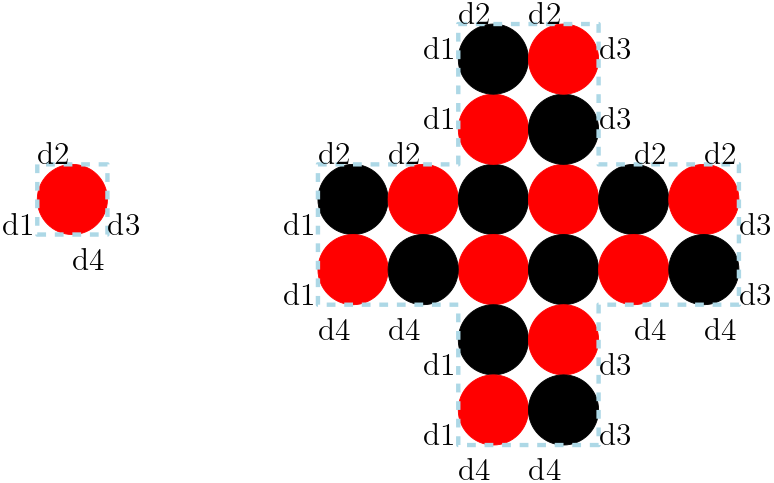}
\caption{An example of two orthogonal convex shapes, with the directions of the perimeter labelled.}
\label{fig:prop-one}
\end{figure}

Let $c_x$ denote the column of a given shape $S$ at the $x$ coordinate, i.e., the set of all nodes of $S$ at $x$. Let $y_{max}(x)$ ($y_{min}(x)$) be the largest (smallest) $y$ value in the $(x, y)$ coordinates of the cells which nodes of a column $c_x$ occupy.

\begin{proposition} \label{prop:columns-rows}
For any connected orthogonal convex shape $S$, all the following are true:
    \begin{itemize}
    \item Every column $c_x$ of $S$ consists of the consecutive nodes\\ $(x,y_{min}(x)),(x,y_{min}(x)+1),\ldots,(x,y_{max}(x))$. 
    \item There are no three columns $c_{x_1}$, $c_{x_2}$, and $c_{x_3}$ of $S$, $x_1<x_2<x_3$, for which both $y_{max}(x_1) > y_{max}(x_2)$ and $y_{max}(x_3) > y_{max}(x_2)$ hold.
    \item There are no three columns $c_{x'_1}$, $c_{x'_2}$, and $c_{x'_3}$ of $S$, $x'_1<x'_2<x'_3$, for which both $y_{min}(x'_1) < y_{min}(x'_2)$ and $y_{min}(x'_3) < y_{min}(x'_2)$ hold.
    \end{itemize}
    All the above hold for rows too in an analogous way.
\end{proposition}

\begin{proof}
For the first property, observe that any discontinuity would violate vertical convexity of column $c_x$ of $S$, thus, vertical convexity of $S$. Next, assume that the second property does not hold, that is, that there are columns $c_{x_1}$, $c_{x_2}$, and $c_{x_3}$ of $S$, $x_1<x_2<x_3$, for which both $y_{max}(x_1) > y_{max}(x_2)$ and $y_{max}(x_3) > y_{max}(x_2)$ hold true. Let w.l.o.g. $y_{max}(x_3)\leq y_{max}(x_1)$. Then the horizontal line joining $(x_3,y_{max}(x_3))$ and $(x_1,y_{max}(x_3))$ passes through an empty cell above $(x_2,y_{max}(x_2))$, thus contradicting orthogonal convexity of $S$. A symmetrical argument holds for the third property. The proof for rows is identical, by rotating the whole system 90\textdegree.
\qed
\end{proof}

\begin{lemma}\label{lem:maximum-difference}
For all $n\geq 3$, the maximum colour-difference of a connected horizo\-ntal-vertical convex shape of size $n$ is $n-2\lfloor n/3 \rfloor$.
\end{lemma}

\begin{proof}
We shall perform a column-based analysis of the maximum colour-differe\-nce of a shape $S$, assuming w.l.o.g. that the majority colour is black. By Proposition \ref{prop:columns-rows}, every column is a consecutive sequence of nodes. This implies that every even-length column has an equal number of blacks and reds, thus does not contribute to the colour-difference of $S$. It also implies that every odd-length column can contribute at most 1 ($-1$) to the colour-difference and a contribution of 1 ($-1$, resp.) happens iff that column starts and ends with a black (red, resp.), including the case of single-node columns. As a consequence, for a shape to maximise its colour-difference it must be maximising the number of black-parity odd-length columns while minimising the number of red-parity odd-length columns.

Consider any internal (i.e., which is not the leftmost or the rightmost) black-parity column $c_x$ of $S$ of length 1. Due to connectivity of $S$, the single black node $(x,y)$ forming $c_x$ must have the red neighbours $(x-1,y)$ and $(x+1,y)$. Note now that $c_{x-1}$ and $c_{x+1}$ cannot both have nodes above $y$ nor both below $y$, as any of these would violate horizontal convexity of $S$. If only one of these two columns has additional nodes, then the contribution to the colour-difference by these 3 columns is 1 by using 5 nodes. If both columns have additional nodes, then let w.l.o.g. $c_{x-1}$ have nodes above $y$ and $c_{x+1}$ below $y$. Then, again, the best contribution to the colour-difference is 1 by using 5 nodes, obtained by adding one black to each column. Adding more nodes to any of these cases cannot improve the $1/5$ ratio.

Next, over all columns of odd length at least 3, the maximum contribution is obtained by the length-3 column $(black,red,black)$, which contributes to the colour-difference 1 per 3 nodes.

Consequently, given $n\geq 3$ nodes, a shape maximising the colour-difference is the one consisting of $\lfloor n/3 \rfloor$ columns of length 3 and $n-3\lfloor n/3 \rfloor\leq 2$ terminal single-node columns, for a maximum colour-difference of $\lfloor n/3 \rfloor+n-3\lfloor n/3 \rfloor =n-2\lfloor n/3 \rfloor$, as required. This shape, which we call the \emph{diagonal line-with-leaves}, is depicted in Figure \ref{fig:diagonal-line-with-leaves}.
\qed
\end{proof}

\FloatBarrier

\begin{figure}
\centering
\includegraphics[width=0.4\textwidth]{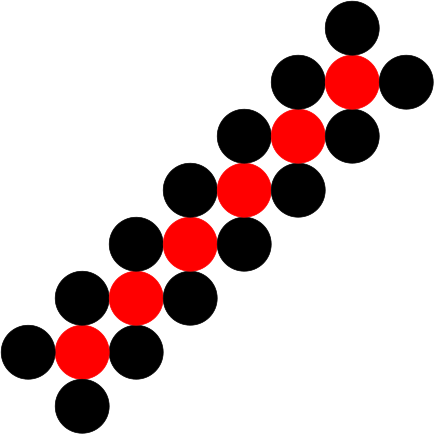}
\caption{The diagonal line-with-leaves shape on $n$ nodes, consisting of $k$ red nodes and $2k$ black nodes arranged in $\lfloor n/3 \rfloor$ columns of length 3, plus at most $2$ terminal single-node columns at the left and right ends.}
\label{fig:diagonal-line-with-leaves}
\end{figure}

\FloatBarrier

A \emph{staircase} is a shape of the form $(x,y), (x+1,y), (x+1,y+1), (x+2,y+1), \ldots$ or $(x,y), (x,y+1), (x+1,y+1), (x+1,y+2), \ldots$. An \emph{extended staircase} is a staircase $Stairs=\{(x_l,y_d), (x_l,y_d+1), (x_l+1,y_d+1), (x_l+1,y_d+2), \ldots\}$ with a bicolour pair at $(x_l-1,y_d)$, $(x_l-1,y_d+1)$ or at $(x_l-1,y_d-1)$, $(x_l-1,y_d)$. Additionally, there are three optional \emph{node-repositories}, $BRep$, $RRep$ and a single-black repository. $BRep=\{(x_l,y_d+2),(x_l+1,y_d+3),(x_l+2,y_d+4),\ldots\}$, $RRep=\{(x_l,y_d-1),(x_l+1,y_d),(x_l+2,y_d+1),\ldots\}$ and the single-black repository at $(x_l-2,y_d)$.

\subsection{Elimination and Generation Sequences}

For convenience, we define $E(S) = \{R_1, R_p, C_1, C_q\}$ as the set containing the first and last rows and columns of a given shape $S$ (omitting $S$ when clear from context), called \emph{terminal} rows/columns, and $adjacent\colon E \to E'$, where $E'=\{R_2,R_{p - 1},C_2,C_{q - 1}\}$, as a function mapping $R_1$ to $R_2$, $R_p$ to $R_{p - 1}$, $C_1$ to $C_2$ and $C_q$ to $C_{q - 1}$.

Recall that, by Proposition \ref{prop:columns-rows}, every row/column of a horizontal/vertical convex shape is a line.

Let $S$ be a connected orthogonal convex shape. A \emph{shape elimination sequence} $\sigma=(u_1, u_2, \ldots, u_n)$ of $S$ is a permutation of the nodes of $S$ satisfying the following properties. Let $S_t=S_{t-1}\setminus \{u_t\}$, where $1\leq t\leq n$ and $S_0=S$. Observe that $S_n$ is always the empty shape. The first property is that, for all $1\leq t\leq n-1$, $S_t$ must be a connected orthogonal convex shape. Moreover, for all $1\leq t\leq n$, $u_t$ must be a node on the external surface of $S_{t-1}$. Essentially, $\sigma$ defines a sequence $S=S_0[u_1]S_1[u_2]S_2[u_3]\ldots S_{n-1}[u_n]S_n=\emptyset$, where, for all $1\leq t\leq n$, a connected orthogonal convex shape $S_t$ is obtained by removing the node $u_t$ from the external surface of the shape $S_{t-1}$.

A \emph{row elimination sequence} $\sigma$ of $S$ is an elimination sequence of $S$ which consists of $p$ sub-sequences $\sigma=\sigma_1\sigma_2\ldots\sigma_{p}$, each sub-sequence $\sigma_i$, $1\leq i\leq p$, satisfying the following properties. Sub-sequence $\sigma_i$ consist of the $k=|R_i|$ nodes of row $R_i$, where $u_1,u_2,\ldots,u_{k}$ is the line formed by row $R_i$. Additionally, $\sigma_i$ is of the form $\sigma_i=\sigma_i^1\sigma_i^2$, where (i) $\sigma_i^1=(u_1,\ldots,u_{k})$ or $\sigma_i^1=(u_{k},\ldots,u_{1})$ and $\sigma_i^2$ is empty or (ii) there is a $u_j\in R_i$, for $2\leq j < k$, such that $\sigma_i^1=(u_1,\ldots,u_{j})$ and $\sigma_i^2=(u_{k},\ldots,u_{j+1})$ or (iii) there is a $u_j\in R_i$, for $1\leq j < k-1$, such that $\sigma_i^1=(u_{k},\ldots,u_{j+2})$ and $\sigma_i^2=(u_{1},\ldots,u_{j+1})$. We shall call any such sub-sequence $\sigma_i$ an elimination sequence of row $R_i$. A \emph{column elimination sequence} of $S$ can be obtained by rotating the whole system by 90\textdegree.

Given a connected orthogonal convex shape $S$ of $n$ nodes, a \emph{shape generation sequence} $\sigma=(u_1, u_2, \ldots, u_n)$ of $S$ is a permutation of the nodes of $S$ satisfying the following properties. Let $S_t=S_{t-1} \cup \{u_t\}$, where $1\leq t\leq n$ and $S_0 = \emptyset$. Observe that $S_n = S$. Any shape generation sequence also satisfies the following properties, which it shares with the shape elimination sequence. The first property is that, for all $1\leq t\leq n-1$, $S_t$ must be a connected orthogonal convex shape. Moreover, for all $1\leq t\leq n$, $u_t$ must be placed in the cell perimeter of $S_{t-1}$. Essentially, $\sigma$ defines a sequence $\emptyset=S_0[u_1]S_1[u_2]S_2[u_3]\ldots S_{n-1}[u_n]S_n=S$, where, for all $1\leq t\leq n$, a connected orthogonal convex shape $S_t$ is obtained by adding the node $u_t$ to the cell perimeter of $S_{t-1}$.

Let $S$ be an extended staircase of $n$ nodes. An \emph{extended staircase generation sequence} $\sigma=(u_1, u_2, \ldots, u_n)$ of $S$ is a generation sequence of $S$ which consists of $q$ sub-sequences $\sigma=\sigma_1\sigma_2\ldots\sigma_{q}$, where each $\sigma_i$ contains the nodes of the column $C_i$ of $S$, ordered such that they do not violate the properties of a shape generation sequence. A \emph{diagonal line-with-leaves generation sequence} is an \emph{extended staircase generation sequence} where the repository of the constructed extended staircase is $\emptyset$.

\begin{lemma}\label{lem:remove-multiple-nodes}
Every connected orthogonal convex shape $S$ has a row (and column) elimination sequence $\sigma$.
\end{lemma}

\begin{proof}
Let $R_1$ be the bottom-most row of $S$, $u_1,u_2,\ldots,u_{k}$ being the line formed by row $R_1$. It is sufficient to prove that there is an elimination sequence $\sigma_1$ of $R_1$, as this can then be applied repeatedly to each subsequent bottom-most row $R_i$, $2\leq i\leq p$, until $S$ becomes empty, $\sigma$ then being obtained by $\sigma=\sigma_1,\sigma_2,\ldots,\sigma_{p}$. 

If there is a single node $u_j$ in $R_1$ which is adjacent to a node $v$ in $R_2$, then, if $2\leq j\leq k-1$, $\sigma_1=(u_1,\ldots,u_{j},u_{k},\ldots,u_{j+1})$ is an elimination sequence of $R_1$ and, if $j=1$ or $j=k$ the same holds for $\sigma_1=(u_{k},\ldots,u_{1})$ and $\sigma_1=(u_1,\ldots,u_{k})$, respectively. This holds because, in all these cases, only removing $u_{j+1}$ before the last step in the sequence could disconnect the shape, thus, connectivity is preserved. Moreover, orthogonal convexity is not violated by any removal as this would contradict either the assumption that $R_1$ is bottom-most or the fact that nodes are only removed from the current endpoints of $R_1$. 

Finally, observe that if multiple nodes in $R_1$ are adjacent to distinct nodes in $R_2$, then these must necessarily be consecutive, otherwise orthogonal convexity would be violated in $R_2$. Setting any of those nodes of $R_1$ as the $u_{j+1}$ of the previous case, will again give elimination sequences of $R_1$.
\qed
\end{proof}

\begin{lemma}\label{lem:hvconvex-parity-bound}
For any connected orthogonal convex shape $S$ of $n$ nodes, given a row elimination sequence $\sigma$ of $S$ and a diagonal line-with-leaves generation sequence $\sigma'$ of a fixed parity which is colour-order preserving w.r.t $\sigma$, the maximum imbalance of any prefix of size $m \leq n$ of $\sigma'$ is at most $O(2m/3)$.
\end{lemma}

\begin{proof}
Assume w.l.o.g that the parity of $\sigma$ is black. If $S$ is a diagonal line-with-leaves with the red parity, where each column in $S$ has $3$ nodes, then every prefix of $\sigma'$ of $m$ nodes will have an imbalance of $2$ red nodes for every black node for all $m$ nodes, leading to the maximum imbalance of $O(2m/3)$.
\qed
\end{proof}

\begin{lemma}\label{lem:dll-parity-bound}
For any diagonal line-with-leaves generation sequence $\sigma$ generating a shape $S$ with a fixed parity column by column, for any sub-sequence $\sigma'$ which is a prefix of $\sigma$, the number of non-parity nodes in $\sigma'$ cannot exceed the number of parity nodes by more than 2.
\end{lemma}

\begin{proof}
Assume that there is such a $\sigma'$, constructing a diagonal line-with-leaves $S'$ of $q$ columns $C_1, C_2, \ldots, C_{q}$. It must be the case that, in the process of constructing $S$, the shape generation sequence generates the shape constructed by $\sigma'$. We assume w.l.o.g. that the parity of $\sigma$ (and by extension $\sigma'$) is black. Therefore, each column $C_i$ in $S'$ constructed by $\sigma'$ must have at least one black node neighbouring every red node to preserve connectivity. Therefore, $\sigma'$ has two possible locations to store additional red nodes without increasing the number of black nodes: by placing one red node in $C_1$ and by placing another in $C_{q}$. Placing any more red nodes violates the structure of a black parity diagonal line-with-leaves by making the lowest node in any $C_i$ the non-parity colour, and is therefore impossible.
\qed
\end{proof}

For the next proof, we ignore the trivial shape of a node surrounded by four other nodes.

\begin{lemma}\label{lem:res-subset-bound}
Let $S$ be a connected orthogonal convex shape. Then there is a row (column) elimination sequence of $S$ which has no single-coloured 3-sub-sequence.
\end{lemma}

\begin{proof}
Assume that every row (column) elimination sequence $\sigma$ has such a single-coloured 3-sub-sequence $\sigma^\prime=(u_i, u_{i + 1},$ $u_{i + 2})$. Assume there is a row $R$ of $S$ such that $u_i, u_{i + 1}, u_{i + 2}\in R$. Recall that a row elimination sequence for a given row $R$ is of the form $\sigma^1\sigma^2$ resulting from the partitioning of $R$ into two consecutive lines, where at most one can be empty. It follows that $\sigma^\prime$ cannot be a sub-sequence of $\sigma^1$ or $\sigma^2$ because each is an alternating sequence of colours. So, $\sigma^\prime$ must be spanning the switching point from $\sigma^1$ to $\sigma^2$, sharing a 2-sub-sequence with either the suffix of $\sigma^1$ or the prefix of $\sigma^2$. But that 2-sub-sequence cannot be single-coloured because each of $\sigma^1$ and $\sigma^2$ is an alternating sequence of colours.

Next, we consider the situation where $\sigma^\prime$ spans multiple rows. Note that if $S$ is a series of one node rows, then $\sigma^\prime$ cannot contain nodes belonging to different rows of $S$ because any row elimination sequence must switch colour to move between rows. If there are two rows $R_1$ and $R_2$, then if $R_2$ is even then we can select the colour by selecting between $\sigma^1$ and $\sigma^2$. If $R_2$ is odd, then both sequences can start with the same colour, but because each alternates there cannot be a 3-sub-sequence unless one is immediately followed by the other. This is only possible if $R_2$ is a $3$ node line and there is a third line $R_3$ with one node. Because we ignore the trivial shape, there must be an $R_4$, and by rotating the row elimination sequence we can get a $\sigma$ without $\sigma^\prime$.
\qed
\end{proof}

Given an extended staircase $S$, an \emph{empty slot} is a cell in the cell perimeter of $S$ which can be occupied by a node $u$ such that $T = S \cup \{u\}$ is an extended staircase.\\

We now present an algorithm (see Algorithm \ref{alg:extended-staircase-gen-alg}), which given a row elimination sequence $\sigma$ returns an extended staircase generation sequence $\sigma^\prime$. The algorithm, which is stated in more general terms and works for a larger set of bi-coloured sequences, first constructs a prefix of 4-5 nodes (4 plus an optional repository for a black node) and then extends it by placing nodes on the staircase. If this is not possible, then it places nodes in the repository corresponding to the colour of the node.

The following is an informal description of Algorithm \ref{alg:extended-staircase-gen-alg}. By assumption, it expects the first two nodes of the input sequence to form a bicolour pair, the third node to be black, and no single-colour 3-sub-sequence to ever arrive. The algorithm positions the pair vertically and the black to its right at $(x_l,y_d)$. If the fourth node is black, it goes to an optional single-black repository to the left of the pair and the fifth node must then be red. Otherwise the fourth node is red. In both cases, a red will be placed over the $(x_l,y_d)$ black. Thus, the prefix of the shape constructed by the algorithm always consists of two vertical pairs and the possibility of a black stored at the single-black repository to their left. If the next node is a red it will be stored in the first red repository position at $(x_l,y_{d-1})$. If not, it is a black. In both cases the next black will start a new column to the right and the algorithm has finished the construction of the prefix having reached its invariant configuration. The invariant satisfies the following properties. New columns always start with the placement of a black. The red repository position of that column below the black and the black repository of the previous column are unoccupied at that point. Any nodes that alternate colours keep growing the staircase part of the shape, preserving the above invariant conditions. If two consecutive nodes of the same colour ever arrive, the second of these nodes will be stored to the first available position of the repository corresponding to its colour. This keeps growing a staircase extended with an upper black and a lower red repository. Both repositories are diagonal lines of consecutive nodes attached to the staircase, starting from its bottom left and having no gaps. The current length of the staircase is an upper bound on the length of the red repository and on the length of the black repository plus 1. 

The following assumptions are made by Algorithm \ref{alg:extended-staircase-gen-alg}. The third node is always a black node. This is a necessary technical assumption that we shall later ensure is always satisfied by our transformations. Variables $N_B$, $N_R$ are assumed to be always set to the current \#nodes in the black, red repository, respectively. The single-black repository at $(x_l-2,y_d)$, not counted in $N_B$, stores the fourth node if both the third and the fourth node of the sequence are black.

\begin{algorithm}
\caption{ExtendedStaircase($\sigma$)}\label{alg:extended-staircase-gen-alg}
\begin{algorithmic}
\Require row elimination sequence $\sigma=(u_1,u_2,\ldots,u_n)$
\Ensure extended staircase generation sequence $\sigma^\prime=(u'_1,u'_2,\ldots,u'_n)$ which is colour-order preserving w.r.t. $\sigma$\\

\State $N_B$, $N_R$: current \#nodes in the black and red repository, respectively\\

\If{$c(u_1)=red$ and $c(u_2)=black$} \Comment{1st and 2nd are always a bicolour pair}
    \State $u'_1=(x_l-1,y_d)$, $u'_2=(x_l-1,y_d+1)$
\Else
    \State $u'_1=(x_l-1,y_d-1)$, $u'_2=(x_l-1,y_d)$
\EndIf

\State $u'_3=(x_l,y_d)$ \Comment{Assumption that 3rd is always black}\\
\If{$c(u_4)=black$}
    \State To be stored at the single-black repository
    \State $u'_5=(x_l,y_d+1)$ \Comment{5th must be red}
    \State $i=6$
\Else
    \State $u'_4=(x_l,y_d+1)$
    \State $i=5$
\EndIf\\

\For{all remaining $i\leq n$}
\State If first of new column, then $u'_i=(x_r+1,y_u)$ \Comment{this is always black}
    \If{$c(u_i)\neq c(u_i-1)$}
        \If{$c(u_i)=black$}
            \State $u'_i=(x_r+1,y_u)$
        \Else
            \State $u'_i=(x_r,y_u+1)$
        \EndIf
    \Else
        \If{$c(u_i)=black$}
            \State $u'_i=(x_l+N_B,y_d+N_B+2)$
        \Else
            \State $u'_i=(x_l+N_R,y_d+N_R-1)$
        \EndIf
    \EndIf
\EndFor
\end{algorithmic}
\end{algorithm}

\begin{lemma}\label{lem:approx-dll-gen-bound}
Let $\sigma$ be a bicoloured sequence of nodes that fulfills all the following conditions:
\begin{itemize}
    \item The set of the first two nodes in $\sigma$ is not single-coloured.
    \item The third node of $\sigma$ is black.
    \item $\sigma$ does not contain a single-coloured 3-sub-sequence.
\end{itemize}
Then there is an extended staircase generation sequence $\sigma' = (u'_1, u'_2, \ldots, u'_n)$ which is colour-order preserving with respect to $\sigma$.
\end{lemma}

\begin{proof}
The sequence $\sigma'$ is the one obtained by applying Algorithm \ref{alg:extended-staircase-gen-alg} to $\sigma$. The algorithm begins by placing the first 4 or 5 nodes of $\sigma'$, depending on whether $u_3$ is red or black respectively. The result is a shape-prefix with 4 nodes, possibly with an extra black in the repository, with $2$ empty black slots and $2$ empty red slots neighbouring the nodes in $(x_l, y_d)$ and $(x_l, y_d+1)$.
We now begin to follow the loop of Algorithm \ref{alg:extended-staircase-gen-alg}. When we extend the staircase by one node, this creates a new column with two empty slots for the opposite colour, one in the new column and another in the repository of a third column. When we add a node of that colour to the column, we create two new empty slots for the first colour in the same manner. As a result, the number of empty slots in the repositories only rises as the staircase extends. Therefore, the $3$ node restriction of the second condition for $\sigma$ is the minimum necessary for the worst case where we have only $2$ empty slots, and the $\sigma'$ derived from such a $\sigma$ by the construction algorithm generates an extended staircase as required.
\qed
\end{proof}

\emph{ExtendedStaircase} is an algorithm which creates an extended staircase generation sequence from a row elimination sequence of a connected orthogonal convex shape.

\begin{lemma}\label{lem:approx-dll-prefix-convexity}
For an extended staircase generation sequence $\sigma$ generated by ExtendedStaircase, every shape generated by a prefix of $\sigma$ is orthogonal convex.
\end{lemma}

\begin{proof}
Observe that an extended staircase consists of $4$ diagonal lines of nodes: the two diagonals of $Stairs$, and the two nodes which connect to and extend them, $BRep$ and $RRep$. The construction of $Stairs$ never has a gap between nodes as the lines of the algorithm which add nodes to it require the colour of the nodes to alternate and the algorithm alternates between creating a new column and adding another node to it. The diagonal lines $BRep$ and $RRep$ grow node by node from the first column of $Stairs$ to the last. Their sizes are therefore upper bounded by the size of $Stairs$, and there can be no vertical or horizontal gap.
\qed
\end{proof}

\begin{lemma}\label{lem:convex-to-approx}
For any connected orthogonal convex shape $S$ of $n$ nodes, given a row elimination sequence $\sigma=(u_1,u_2,\ldots,u_n)$ of $S$ where the set of the first two nodes in $\sigma$ is not single-coloured and $u_3$ is black, there is an extended staircase generation sequence $\sigma'=(u'_1,u'_2,\ldots,u'_n)$ which is colour-order preserving w.r.t $\sigma$ and such that, for all $1 \leq i \leq |\sigma|$, $D_i = \{u'_1, u'_2, \ldots, u'_i\}$ is a connected orthogonal convex shape.
\end{lemma}

\begin{proof}
By Lemma \ref{lem:res-subset-bound}, $\sigma$ will not have a single-coloured 3-sub-sequence. Therefore, by our assumption about $\sigma$ and Lemma \ref{lem:approx-dll-gen-bound} we have a $\sigma'$. We can then place the nodes of $\sigma'$ as in Algorithm \ref{alg:extended-staircase-gen-alg}. By Lemma \ref{lem:approx-dll-prefix-convexity}, all prefixes $\sigma'_i$ of $\sigma'$ construct an orthogonal convex shape (excluding the black repository), and therefore all $D_i$ are connected orthogonal convex shapes.
\qed
\end{proof}

\begin{observation}
For any connected orthogonal convex shape $S$ of $n$ nodes, if the set of the first two nodes in the row elimination sequence $\sigma=(u_1,u_2,u_3,\ldots,u_n)$ is single-coloured, $u_3$ is black and there is an empty cell $c$ of the opposite colour in the cell perimeter of $S$ such that if $c$ is occupied by $v$ then $S \cup \{v\}$ is an orthogonal convex shape, then $S \cup \{v\} \setminus u_1$ has a row elimination sequence $\sigma'$ where the set of the first two nodes in $\sigma'$ is not single-coloured.
\end{observation}

The \emph{anchor node} of the shape $S$ of $p$ rows $R_1, R_2, \ldots, R_p$ is the rightmost node in the row $R_p$, counting rows from bottom to top.
\emph{ExtendedStaircase} is an algorithm which creates an extended staircase generation sequence from a row elimination sequence of a connected horizontal-vertical convex shape.

\begin{lemma}\label{lem:shape-staircase-properties}
Let $S$ be a connected orthogonal convex shape of $n$ nodes divided into $p$ rows $R_1, R_2, \ldots, R_p$, and $\sigma=(u_1,u_2,\ldots,u_n)$ a row elimination sequence from $R_1$ to $R_p$ of $S$. If the bottom node of the first two nodes placed by ExtendedStaircase is fixed to $(x_c, y_c+1)$, where $(x_c, y_c)$ are the co-ordinates of the anchor node of $S$, the shape $T_i = ExtendedStaircase(\sigma_i)$, where $\sigma_i=(u_1,u_2,\ldots,u_i)$, $1 \leq i \leq n$, fulfills the following properties:

\begin{itemize}
    \item $S \cup T_i$ is a connected shape.
    \item $S \cap T = \emptyset$.
    \item excluding the single-black repository, $R_p \cup T_i$ is an orthogonal convex shape.
\end{itemize}
\end{lemma}

\begin{proof}
Let $u_1, u_2, \ldots, u_i$ be the nodes in the sequence $\sigma_i$. If the first node is black, Algorithm \ref{alg:extended-staircase-gen-alg} places a node in $(x_l-1, y_d-1)$, otherwise it places it in $(x_l-1, y_d)$. By Lemma \ref{lem:approx-dll-prefix-convexity}, all $\sigma_i$ generate an orthogonal convex shape, so $T_i$ cannot be a disconnected shape. Therefore, the shape $S \cup T_i$ is connected. In addition, the co-ordinates $(x_l-1, y_d-1)$ and $(x_l-1, y_d)$ represent the two potential bottom-left corners of the shape $T$. Therefore, there can be no overlap (i.e. placement of nodes in occupied cells) as the existence of a node of $S$ in the space $T$ is constructed in would contradict the definition of an anchor node. In addition, the cell $(x_l-2, y_d)$ (the single-black repository) is always empty as a node in that cell would have the $y$ co-ordinate $y_d$, which is above the anchor node at $y_d-2$ or $y_d-1$, violating the definition of the anchor node. Finally, since the nodes $u_1$ and $u_2$ construct a column, and every node $u_3, \ldots, u_n$ (excluding the single-black repository) is necessarily to the right of this column, there cannot be a violation of orthogonal convexity with the row $R_p$.
\qed
\end{proof}

\begin{lemma}\label{lem:extend-to-dll}
For any extended staircase $W \cup T$ of $n$ nodes, where $W$ is the $Stairs$, $T \subseteq \{BRep \cup RRep\}$ and $k = |T|$, given a shape elimination sequence $\sigma=(u_1,u_2,\ldots,u_k)$ of $T$, there is a diagonal line-with-leaves generation sequence $\sigma'=(u'_1,u'_2,\ldots,u'_k)$ which is colour-order preserving w.r.t $\sigma$ and such that, for all $1 \leq i \leq |\sigma|$, $D_i = W \cup \{u'_1, u'_2, \ldots, u'_i\}$ is a connected orthogonal convex shape.
\end{lemma}

\begin{proof}
We use a shape elimination sequence $\sigma$ of $T$ which alternates between taking nodes from the black repository $BRep$ and the red repository $RRep$. It does this until only one repository remains. We can then use $\sigma'$ to place the nodes of $D$ as in the for loop of Algorithm \ref{alg:extended-staircase-gen-alg}, effectively extending $W$. If we maximise the size of $\sigma$ then the resulting $D_k$, a $Stairs$ with only one repository, is equivalent to a diagonal line-with-leaves. By Lemma \ref{lem:approx-dll-prefix-convexity}, all prefixes $D_i$ generated by $\sigma'$ are connected orthogonal convex shapes (excluding the black repository).
\qed
\end{proof}

\section{The Transformation}\label{sec5}

In this section, we present the transformation of \emph{orthogonal convex} shapes, via an algorithm (Algorithm \ref{alg:high-level-alg}) for constructing a diagonal line-with-leaves from any orthogonal convex shape $S$. For the first step of the algorithm, we generate a 6-robot from the seed and the shape, which we then use to transport nodes. By using a row elimination sequence of $S$ and an extended staircase generation sequence, we convert the initial shape $S$ into an extended staircase. We then use appropriate elimination and generation sequences focused on the repositories of the extended staircase, to convert the latter into a diagonal line-with-leaves. Given any two colour-consistent orthogonal convex shapes $A$ and $B$ and their diagonal line-with-leaves $D$, our algorithm can be used to transform both $A$ into $D$ and $B$ into $D$ and, thus, $A$ into $B$, by reversing the latter transformation.

Our transformations rely on the use of a \emph{k-robot}, a shape with $k$ nodes which is responsible for transporting nodes. The $k$-robot \emph{extracts} a node $u$ if it is positioned such that $u$ rotates around a node of the robot and the result is a $k+1$-robot where $u$ is the load of the robot. The $k+1$-robot \emph{places} its load in the cell $c$ if it is positioned such that the load rotates into $c$ and the result is a $k$-robot.

\begin{algorithm}
\caption{HVConvexToDLL($S, M$)}\label{alg:high-level-alg}
\begin{algorithmic}
\Require shape $S \cup M$, where $S$ is a connected orthogonal convex shape of $n$ nodes and $M$ is a 3-node seed on the cell perimeter of $S$, row elimination sequence $\sigma=(u_1,u_2,\ldots,u_n)$ of $S$, extended staircase generation sequence of $W \cup T = \sigma^\prime=(u'_1,u'_2,\ldots,u'_n)$ which is colour-order preserving w.r.t. $\sigma$, shape elimination sequence $\sigma=(u_1,u_2,\ldots,u_{|T|})$ of $T$, shape generation sequence of $X = \sigma^\prime=(u'_1,u'_2,\ldots,u'_{|T|})$ which is colour-order preserving w.r.t. $\sigma$
\Ensure shape $G = W \cup X \cup M$, where $G$ is a diagonal line-with-leaves and $M$ is a connected 3-node shape on the cell perimeter of $S$.

\State $R \gets$ GenerateRobot($S, M$)
\State $\sigma \gets$ rowEliminationSequence($S$)
\State $\sigma' \gets$ ExtendedStaircase($\sigma$)
\State $W \cup T \gets$ HVConvexToExtStaircase($S, R, \sigma, \sigma'$)
\State $\sigma \gets$ repsEliminationSequence($W \cup T$)
\State $\sigma' \gets$ stairExtensionSequence($W \cup T$)
\State $G \gets$ ExtStaircaseToDLL($W \cup T, R, \sigma, \sigma'$)
\State TerminateRobot($G, R$)
\end{algorithmic}
\end{algorithm}

\pagebreak

\begin{algorithm}
\caption{HVConvexToExtStaircase($S, R, \sigma, \sigma'$)}\label{alg:extended-staircase-build-alg}
\begin{algorithmic}
\Require shape $S \cup R$, where $S$ is a connected orthogonal convex shape of $n$ nodes and $R$ is a 6-node robot on the cell perimeter of $S$, row elimination sequence $\sigma=(u_1,u_2,\ldots,u_n)$ of $S$, extended staircase generation sequence $\sigma^\prime=(u'_1,u'_2,\ldots,u'_n)$ which is colour-order preserving w.r.t. $\sigma$
\Ensure shape $T \cup R$, where $T$ is the extended staircase generated by $\sigma'$

\For{all $1 \leq i \leq n$}
    \State $source \gets \sigma_i$
    \State $dest \gets \sigma^\prime_i$
    \While{R cannot extract source}
        \If{R can climb}
            \State $Climb(R)$
        \Else
            \State $Slide(R)$
        \EndIf
    \EndWhile
    \State $Extract(R, source)$
    \While{R cannot place its load in dest}
        \If{R can climb}
            \State $Climb(R)$
        \Else
            \State $Slide(R)$
        \EndIf
    \EndWhile
    \State $Place(R, dest)$
\EndFor
\end{algorithmic}
\end{algorithm}

\pagebreak

\begin{algorithm}
\caption{ExtStaircaseToDLL($W, R, \sigma, \sigma'$)}\label{alg:dll-build-alg}
\begin{algorithmic}
\Require extended staircase $W=Stairs\cup \{BRep \cup RRep\}$ and a 6-robot $R$ on its cell perimeter, shape elimination sequence $\sigma=(u_1,u_2,\ldots,u_{|T|})$ of $T\subseteq \{BRep \cup RRep\}$, shape generation sequence $\sigma^\prime=(u'_1,u'_2,\ldots,u'_{|T|})$ which is colour-order preserving w.r.t. $\sigma$
\Ensure shape $Stairs^\prime \cup R'$, where $Stairs^\prime\setminus Stairs$ is an extension of $Stairs$ generated by $\sigma'$ and $R'$ is a 6-robot which is colour-consistent with $R$.

\For{all $1 \leq i \leq |T|$}
    \State $source \gets u_i$
    \State $dest \gets u^\prime_i$
    \While{R not at $source$}
        \If{R can climb}
            \State $ClimbTowards(R,source)$
        \Else
            \State $SlideTowards(R,source)$
        \EndIf
    \EndWhile
    \State $Extract(R, source)$
    \While{R not at $dest$}
        \If{R can climb}
            \State $ClimbTowards(R,dest)$
        \Else
            \State $SlideTowards(R,dest)$
        \EndIf
    \EndWhile
    \State $Place(R, dest)$
\EndFor
\end{algorithmic}
\end{algorithm}

\pagebreak

\subsection{Robot Traversal Capabilities}

\subsubsection{6-Robot Movement}

We first show that for all $S$ in the family of orthogonal convex shapes, a connected 6-robot is capable of traversing the perimeter of $S$. We prove this by first providing a series of scenarios which we call \emph{corners}, where we show that the 6-robot is capable of making progress past the obstacle that the corner represents. We then use Proposition \ref{prop:regular-expression} to show that the perimeter of any $S$ is necessarily made up of a sequence of such corners, and therefore the 6-robot is capable of traversing it.

\FloatBarrier

\begin{figure}
\centering
\begin{subfigure}{1\textwidth}
\centering
\includegraphics[width=0.75\linewidth]{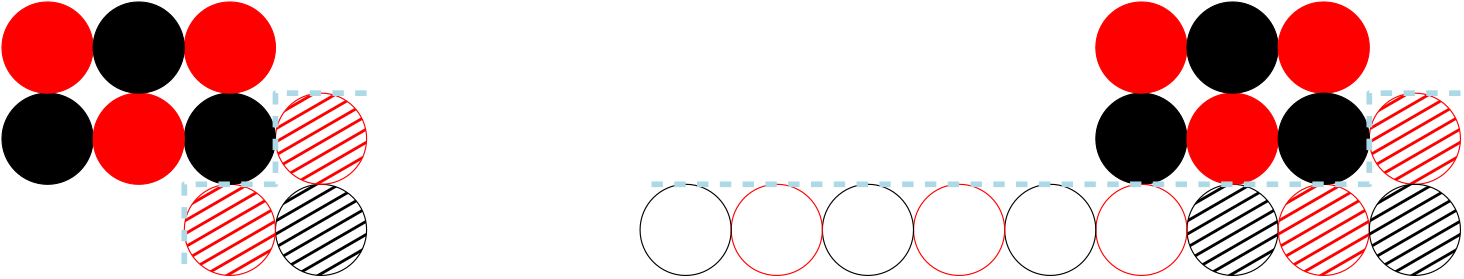}
\caption{The height 1 cases, with widths 1 and 2+.}
\vspace{0.4cm}
\end{subfigure}
\begin{subfigure}{.5\textwidth}
\vspace{1cm}
\centering
\includegraphics[width=0.75\linewidth]{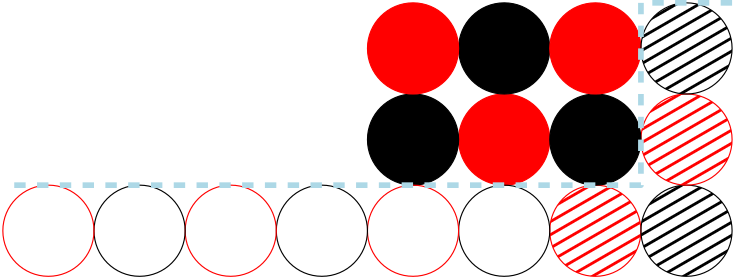}
\caption{The height 2 case.}
\end{subfigure}%
\begin{subfigure}{.5\textwidth}
\centering
\includegraphics[width=0.75\linewidth]{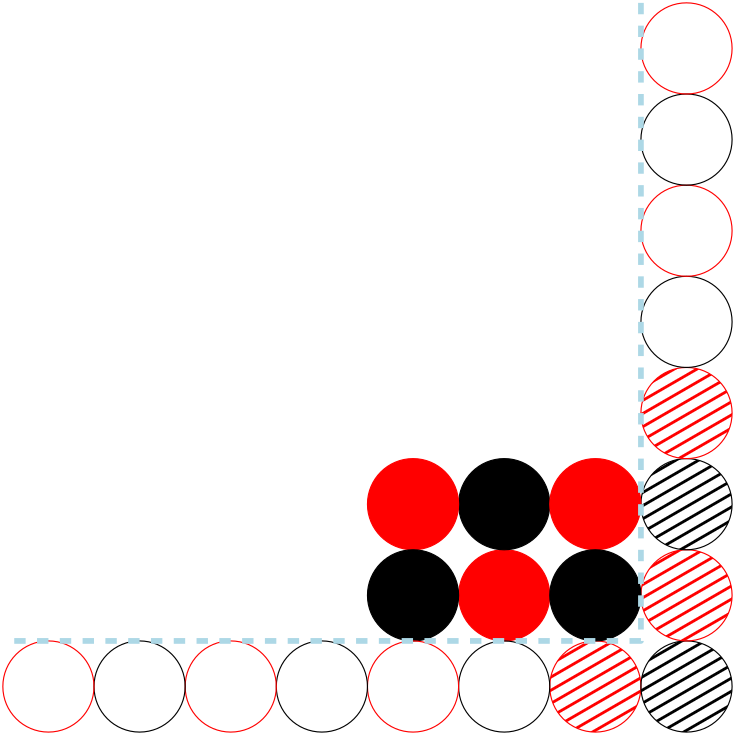}
\caption{The height 3+ case.}
\end{subfigure}
\begin{subfigure}{.5\textwidth}
\vspace{1cm}
\centering
\includegraphics[width=0.3\linewidth]{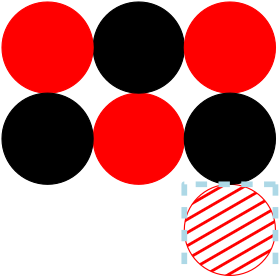}
\caption{An extreme version of the partial quadrant, when a quadrant consists of a single line.}
\end{subfigure}
\caption{The five cases considered in the proof, four corner cases and one edge case where a section of the perimeter does not correspond to a corner case due to its structure. Striped circles represent the nodes on the exterior of the shape. Hollow circles represent potential space for additional nodes for corner scenarios which are not in this set (due e.g. to having longer horizontal/vertical lines).}
\label{fig:scenarios}
\end{figure}

\FloatBarrier

We define progress as the movement of the 6-robot upwards and to the right of its starting position, i.e. any change in the position of the shape such that the shape is in the same $a \times b$ formation but the co-ordinates of all of the nodes have increased. This is equivalent to being able to traverse the relevant section of a perimeter. Our goal is to show that attaining the maximum progress (i.e. the movement which maximises the increase in the co-ordinates while preserving the formation) for each corner is possible. Since we can construct a series of corners where every corner follows from the point of maximum progress of the previous corner, it follows that for such a series we can make progress indefinitely. By rotating the robot and the quadrant as necessary, we can make the same argument for progress in any direction, which is equivalent to being able to traverse a perimeter indefinitely, provided we also show that the perimeter is necessarily made up of such corners.

We begin by considering the $up$-$right$ quadrant, that is any cells which neighbour the section of the perimeter defined by the regular expression $d_1(d_1\;|\;d_2)^{*}d_2$ $(d_2\;|\;d_3)^{*}d_3$, where $d_1, d_2$ and $d_3$ are $up$, $right$ and $down$ respectively, as our base case.

Let $\mathcal{C}$ be a set of orthogonal convex shapes, where each shape is a corner scenario for the $up-right$ quadrant, depicted in Figure \ref{fig:climb-scenarios}. Given a corner-shape scenario $C\in \mathcal{C}$ consisting of a horizontal line $(x_l,y_d), (x_l+1,y_d),\ldots,(x_r,y_d)$ and a vertical line $(x_r,y_d), (x_r,y_d+1),\ldots,(x_r,y_u)$, as depicted in Figure \ref{fig:climb-scenarios}, we define its \emph{width} $w(C)=|x_r-x_l|$, i.e., equal to the length of its horizontal line, and its \emph{height} $h(C)=|y_u-y_d|$, i.e., equal to the length of its vertical line, excluding in both cases the corner node $(x_r,y_d)$.

\FloatBarrier

\begin{figure}
\centering
\includegraphics[width=0.75\linewidth]{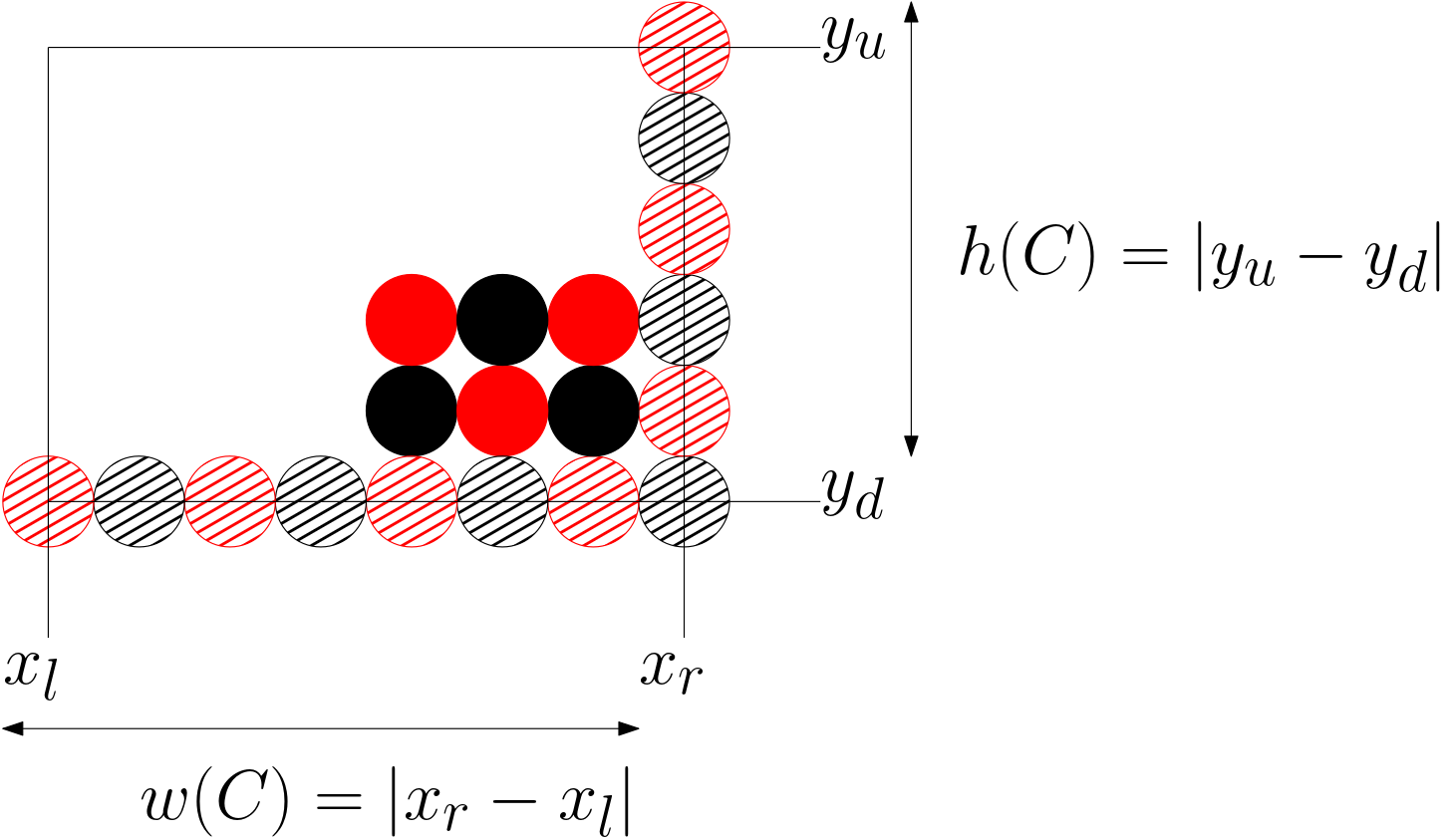}
\caption{A visual representation of the variables we use in our proof.}
\label{fig:variables}
\end{figure}

\begin{figure}
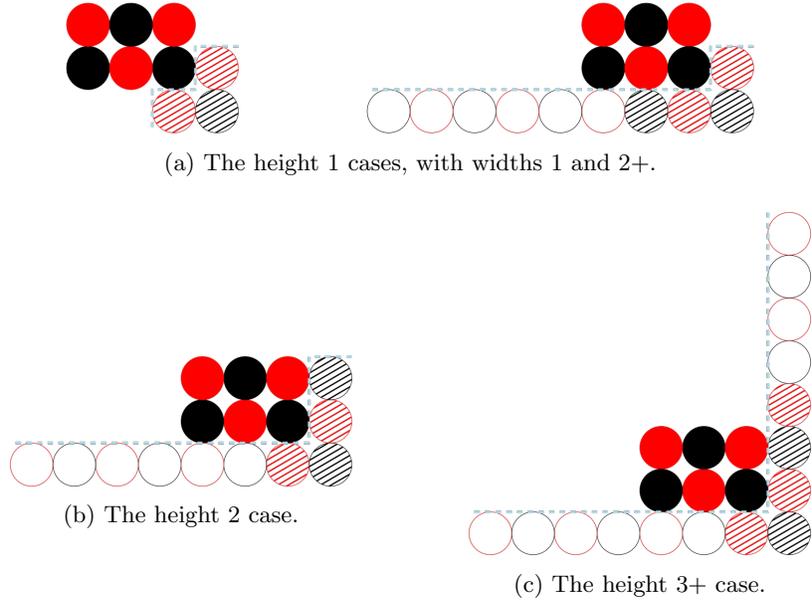

\centering
\begin{subfigure}{1\textwidth}
\centering
\includegraphics[width=0.75\linewidth]{ClimbScenarios1New.png}
\caption{The height 1 cases, with widths 1 and 2+.}
\vspace{0.4cm}
\end{subfigure}
\begin{subfigure}{.5\textwidth}
\vspace{1cm}
\centering
\includegraphics[width=0.75\linewidth]{ClimbScenarios2New.png}
\caption{The height 2 case.}
\end{subfigure}%
\begin{subfigure}{.5\textwidth}
\centering
\includegraphics[width=0.75\linewidth]{ClimbScenarios3New.png}
\caption{The height 3+ case.}
\end{subfigure}
\caption{The four basic corner scenarios of $\mathcal{C}$.}
\label{fig:climb-scenarios}
\end{figure}

\FloatBarrier

\begin{lemma}\label{lem:corner-cases}
For any orthogonal convex shape $S$, the extended external surface defined by the regular expression $d_1(d_1\;|\;d_2)^{*}d_2$ $(d_2\;|\;d_3)^{*}d_3$ of the shape can be divided into a series of shapes $S_0, S_1,\ldots$, where all $S_i \in \mathcal{C}$.
\end{lemma}

\begin{proof}
By Proposition 1, there is the section of the perimeter of orthogonal convex shapes which is defined by the regular expression $d_1(d_1\;|\;d_2)^{*}d_2$ $(d_2\;|\;d_3)^{*}d_3$, where $d_1, d_2$ and $d_3$ are $up$, $right$ and $down$ respectively. This section of the perimeter forms a ``quadrant" where all movement is in the up and right directions, terminated by the first $d_3$, as can be seen in the examples in \ref{fig:prop-one}. By the regular expression, the nodes on the perimeter must necessarily form alternating horizontal and vertical lines. We can therefore divide this section of the perimeter into a series of subsections, where in each subsection we have a horizontal line which connects to a vertical line via the right-most node on the horizontal line. The cases in Figure \ref{fig:climb-scenarios} cover all potential widths and heights where this vertical line positioning constraint holds. This even holds for the edge cases where a vertical ending at $(x_r,y_u)$ is immediately followed by another vertical starting at $(x_r + 1, y_u)$, provided we allow $(x_r,y_u)$ to act both as $(x_r,y_u)$ and $(x_l,y_d)$ for each case respectively. Therefore, they cover all potential cases in the quadrant, and since the quadrant is made up of these cases, it covers the extended external surface of the whole quadrant.
\qed
\end{proof}

Given that the quadrant is made up of cases from $\mathcal{C}$, if the 6-robot is able to move from one vertical to another for all $S_i \in \mathcal{C}$, it is able to do so for any up-right quadrant of the perimeter until it runs into the $d_3$ line. We now show that this movement is possible, first for this quadrant and later for all four quadrants.

\begin{lemma}\label{lem:progress-corner-cases}
For all shapes $C \in \mathcal{C}$, if a $2 \times 3$ shape (the 6-robot) is placed in the cells $(x_l-2,y_d+1),(x_l-1,y_d+1),(x_l,y_d+1),(x_l-2,y_d+2),(x_l-1,y_d+2),(x_l,y_d+2)$, it is capable of translating itself to $(x_r-2,y_u+1),(x_r-1,y_u+1),(x_r,y_u+1),(x_r-2,y_u+2),(x_r-1,y_u+2),(x_r,y_u+2)$.
\end{lemma}

\begin{proof}
For our proof strategy, we present a series of motions which for all shapes $C \in \mathcal{C}$ lead the 6-robot from the leftmost node of the horizontal to the topmost node of the vertical. We group some of these motions into high-level motions (i.e. moving the whole 6-robot by moving individual nodes). We begin by noting that we can perform repetitive motions to traverse a horizontal or vertical line to the end. Our first motion is \emph{sliding}, depicted in Figure \ref{fig:slide-across}, where pairs of nodes rotate around each other to slide across a line. By reorienting the shape, the 6-robot and its movement vertically, it follows that the 6-robot can slide up vertical lines as well. The second action is a special version of the slide depicted in Figure \ref{fig:slide-on}. This slide is slower but allows the object to preserve connectivity in the situation where only one or two nodes are connected to the line. We are therefore already able to claim that moving across and onto lines is possible. What remains is the intersection of horizontal and vertical lines.

\FloatBarrier

\begin{figure}
\begin{minipage}{.5\textwidth}
\hspace{-1cm}
\centering
\includegraphics[width=0.75\linewidth]{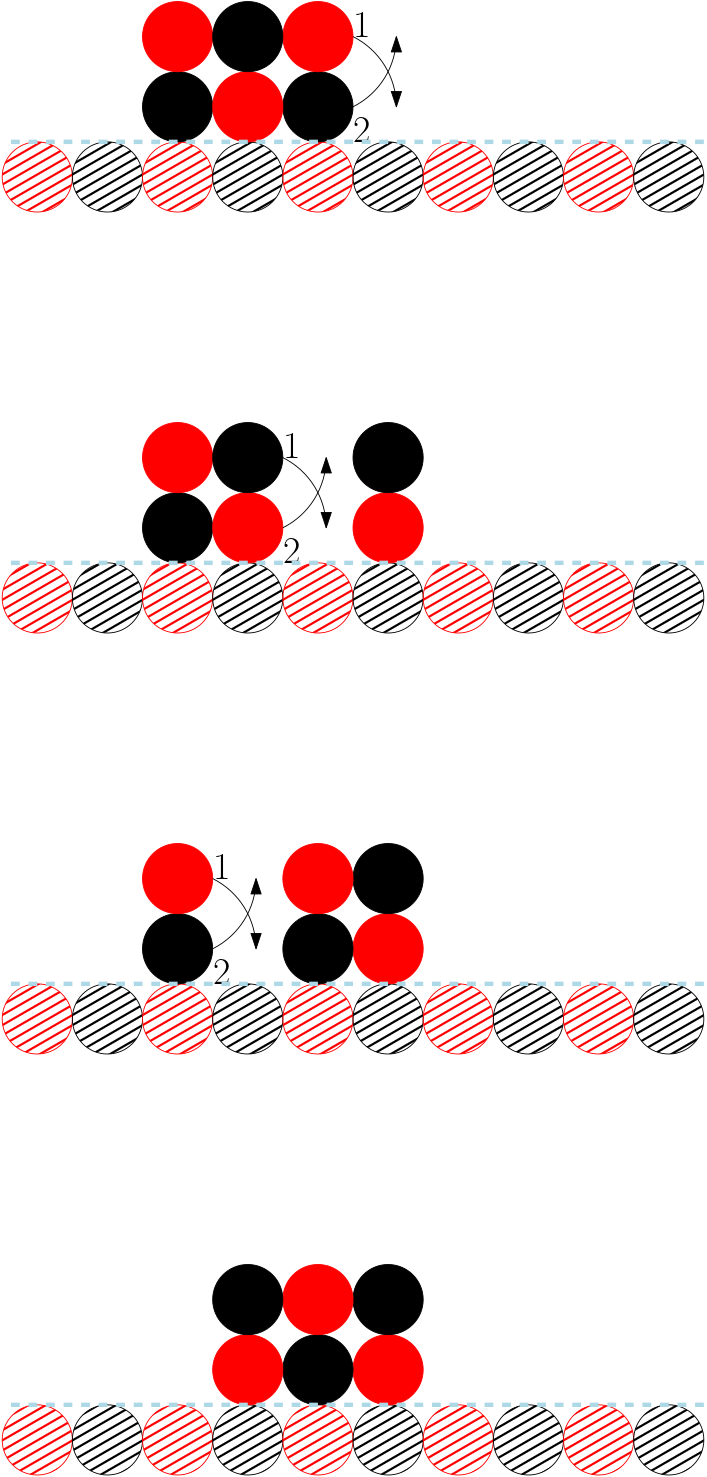}
\captionsetup{width=.9\linewidth}
\caption{Sliding across a horizontal line.}
\label{fig:slide-across}
\end{minipage}%
\begin{minipage}{.5\textwidth}
\hspace{0.5cm}
\centering
\includegraphics[width=0.85\linewidth]{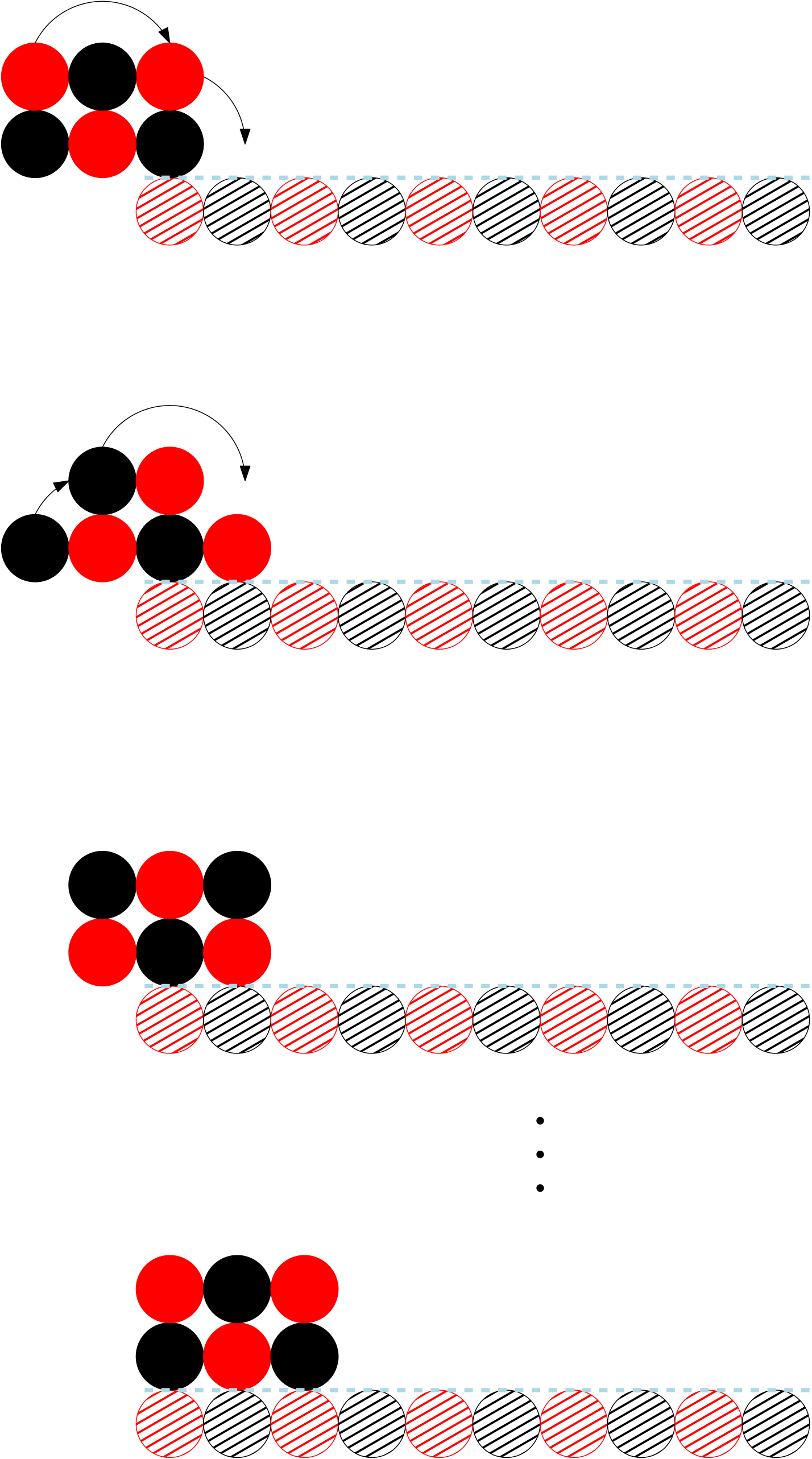}
\captionsetup{width=.9\linewidth}
\caption{Sliding onto a horizontal line. The steps are repeated after the third configuration to reach the fourth configuration.}
\label{fig:slide-on}
\end{minipage}
\end{figure}

\FloatBarrier

Our first is the case where the 6-robot lies on a horizontal of arbitrary width, and attempts to climb a vertical of height 3. It does by following the series of rotations in Figure \ref{fig:climb3}. The result is that the robot lies on top of the vertical line, and is therefore able to use the slide and/or special slide movements (if necessary) to move across the horizontal line at the top (not depicted) to the next vertical. This climbing procedure can be performed no matter how long the first horizontal (i.e. the one the robot lies on initially) is. In addition, there is a section of the movement which puts the seed in the position to slide vertically (see Figure \ref{fig:climb4+}), allowing a modified version the same procedure to climb verticals of arbitrary height.

Finally, there are the solutions for the cases where the height is 2 (Figure \ref{fig:climb2}) and where the height is 1 (Figure \ref{fig:climb1_1} and \ref{fig:climb1_2}). Note that unlike the former movement, the latter movements vary depending on the width (2+ and 1 respectively).
\qed
\end{proof}

\begin{theorem}\label{the:6-robot-traverse}
For any orthogonal convex shape $S$, a 6-robot is capable of traversing the perimeter of $S$.
\end{theorem}

\begin{proof}
By Lemma \ref{lem:progress-corner-cases} we have shown for the up-right quadrant firstly that it is possible to slide across a horizontal of arbitrary width no matter the robot's initial position, that it is possible to climb a height 3 vertical, that part of this movement can be repeated indefinitely to climb verticals of arbitrary height, that special movements exist for smaller verticals and that all of this is possible no matter how long the horizontal line the object lies on is.

By rotating the robot and the quadrant as necessary, we are able to replicate our movements for all other quadrants: $d_4(d_4\;|\;d_1)^{*}$ (the left-up quadrant) $d_2(d_2\;|\;d_3)^{*}$ (the right-down quadrant) and $d_3(d_3\;|\;d_4)^{*}$ (the left-down quadrant). All that remains is the transition between the quadrants.

There are two cases. In the first case, the next quadrant consists of multiple lines. In this case, when the line signifying the end of the current quadrant is met it is sufficient to begin movements appropriate to travelling in the next quadrant. However, there is an edge case where a quadrant consists of a single line. In this case, a unique movement is necessary (see Figure \ref{fig:edge1}) to transfer the 6-robot onto the line. These movements are then followed by special slides to put the object into position for the next quadrant. Naturally, these transformations are reversible and can be mirrored as well. We are therefore able to deal with any quadrant transition, even rotating the 6-robot around a single node.

Therefore, because we can move through any variant of all quadrants and transition between them, a 6-robot can traverse the perimeter of any orthogonal convex shape.
\qed
\end{proof}

\FloatBarrier

\begin{figure}
\centering
\includegraphics[width=0.4\textwidth]{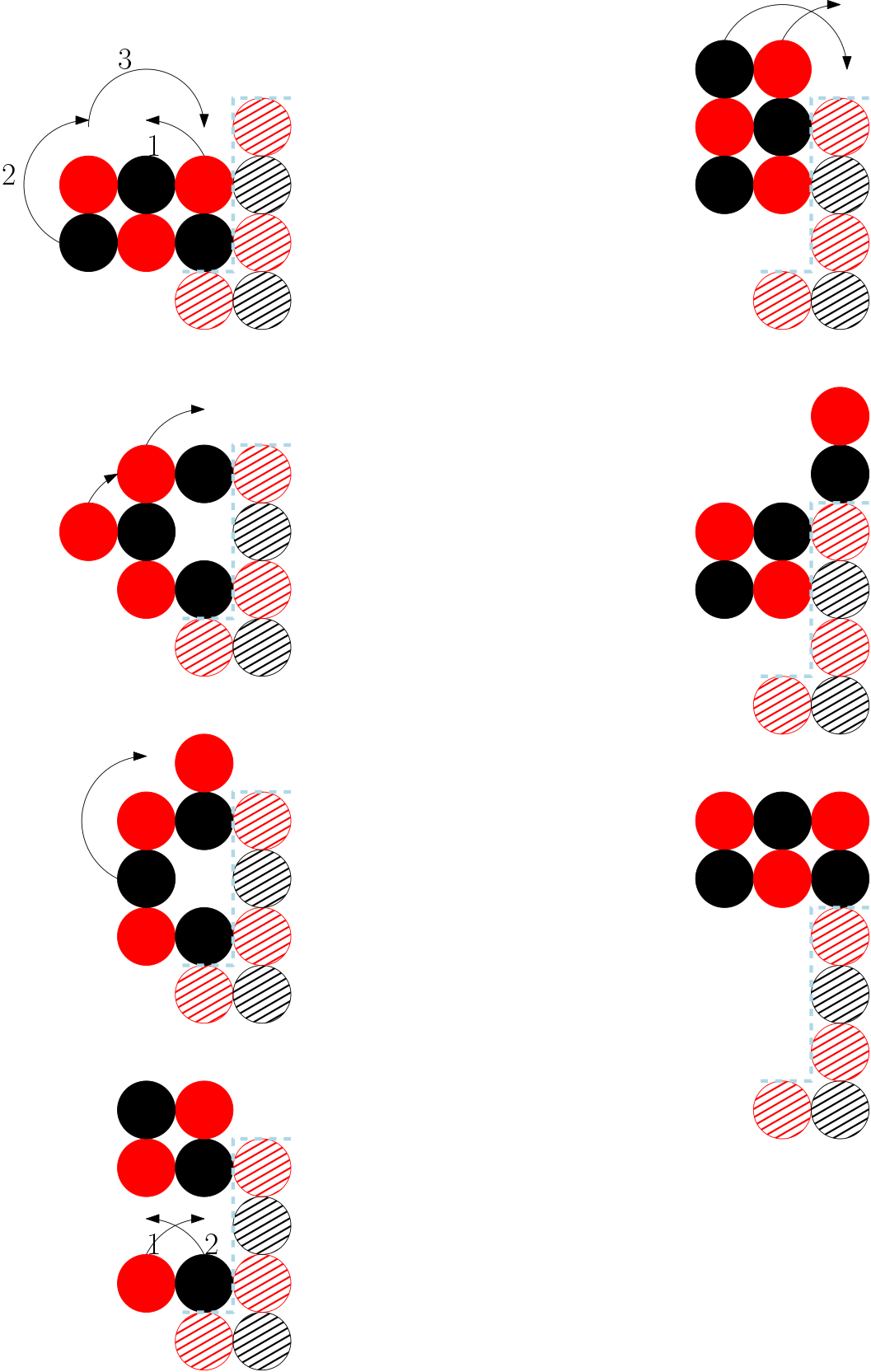}
\caption{Climbing a height 3 vertical with a width 1 horizontal. All figures are read as pairs of columns, top-down.}
\label{fig:climb3}
\end{figure}

\begin{figure}
\centering
\includegraphics[width=0.4\textwidth]{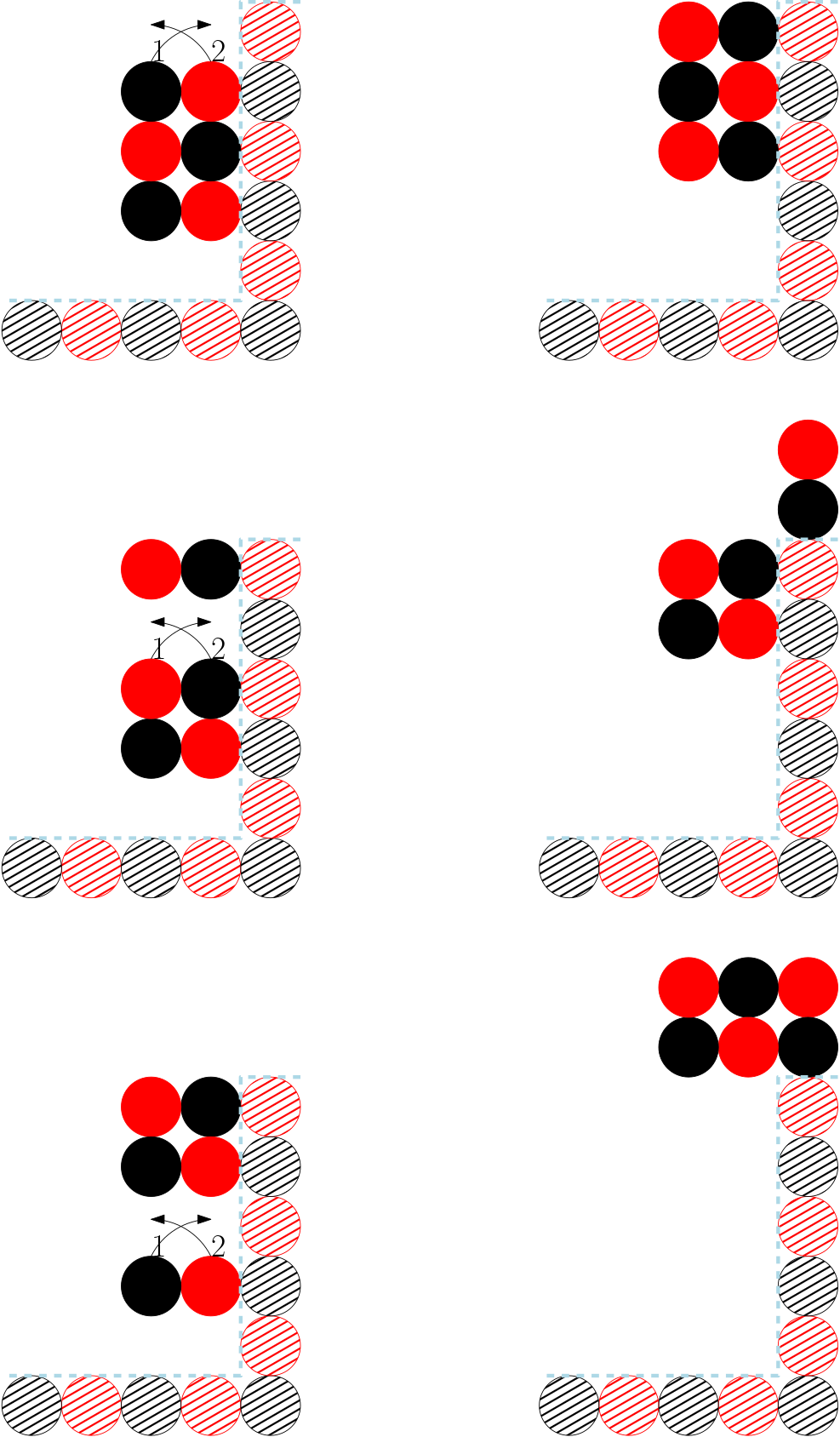}
\caption{Climbing a vertical of arbitrary height with an arbitrary width horizontal. The process starts as in Figure \ref{fig:climb3}, and the upwards slide in the first column of snapshots can be repeated for as long as necessary to climb the wall. This corresponds to case (c) of Figure \ref{fig:climb-scenarios}.}
\label{fig:climb4+}
\end{figure}

\begin{figure}
\centering
\includegraphics[width=0.4\textwidth]{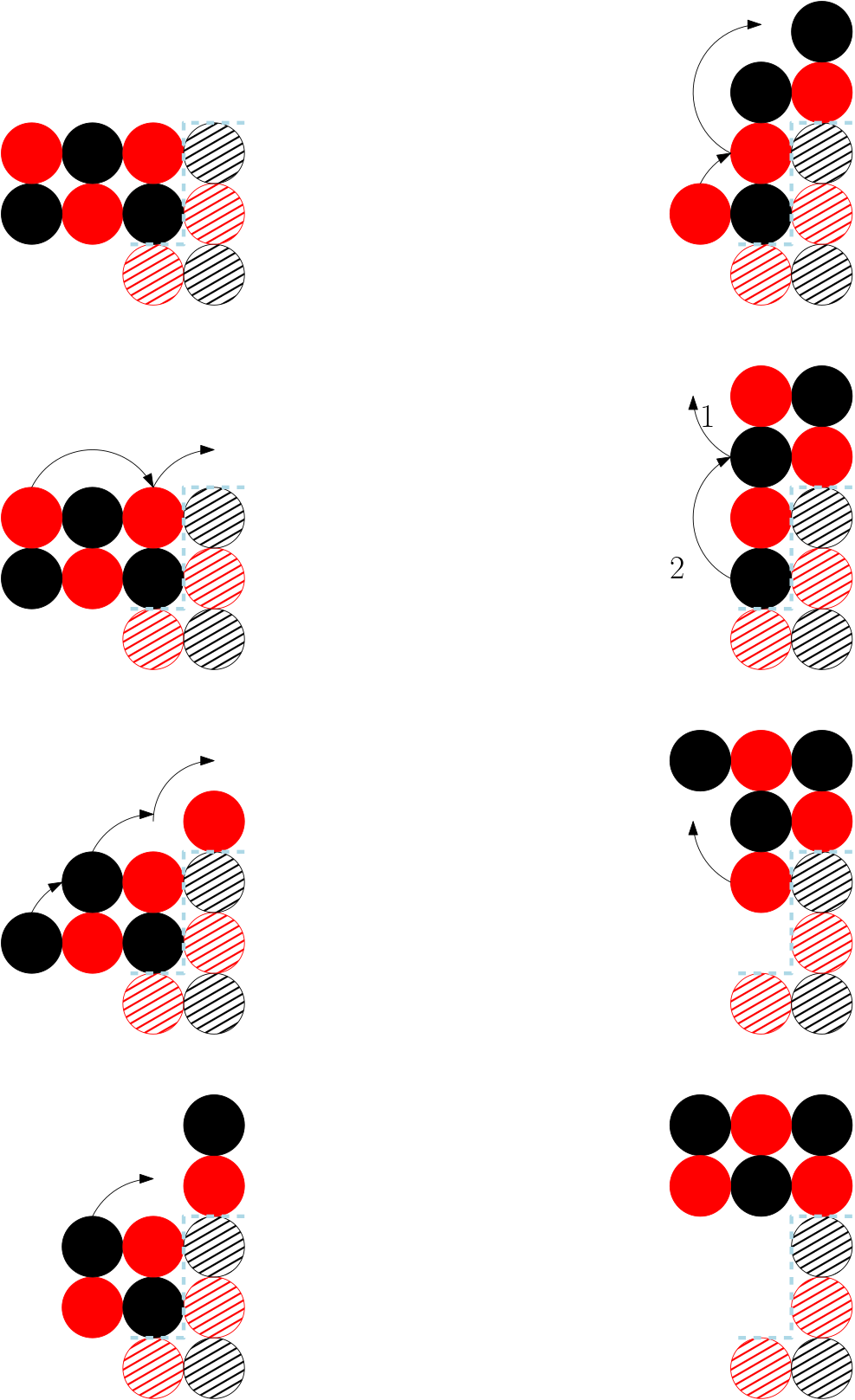}
\caption{Climbing a height 2 vertical with an arbitrary width horizontal. This corresponds to case (b) of Figure \ref{fig:climb-scenarios}.}
\label{fig:climb2}
\end{figure}

\begin{figure}
\centering
\includegraphics[width=0.4\textwidth]{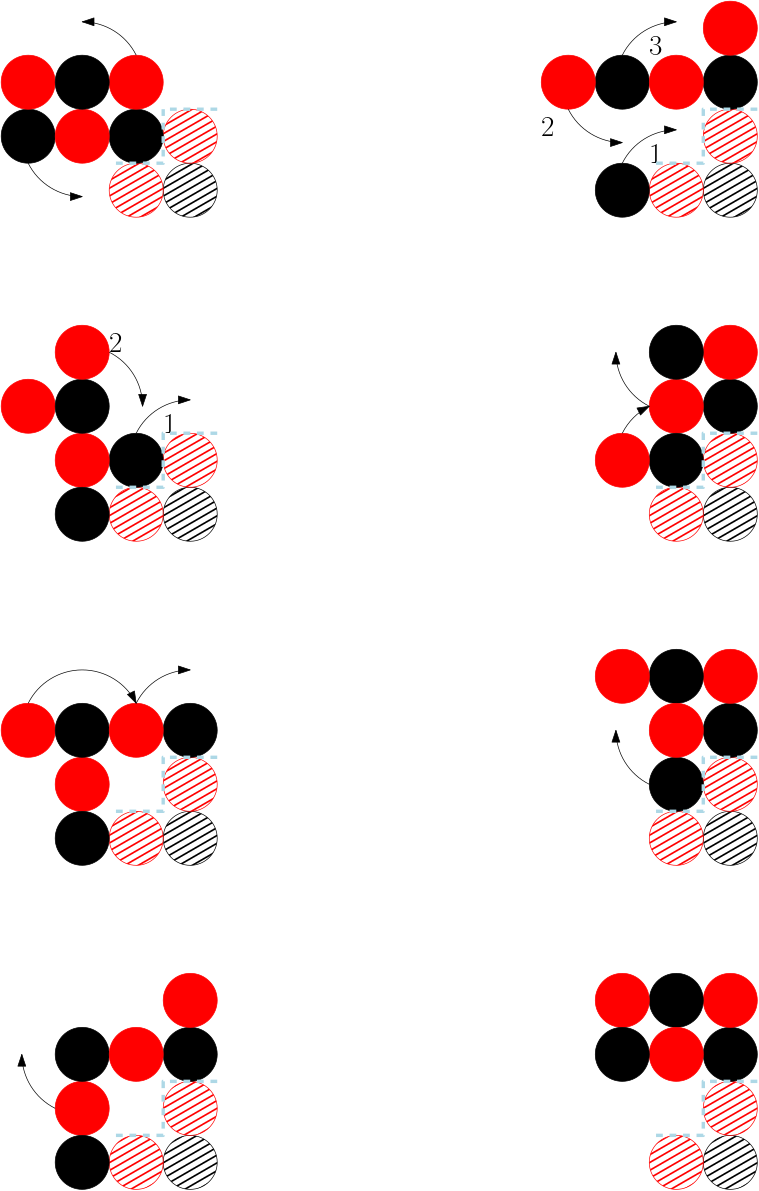}
\caption{Climbing a height 1 vertical with a width 1 horizontal. This corresponds to the first (a) case of Figure \ref{fig:climb-scenarios}.}
\label{fig:climb1_1}
\end{figure}

\begin{figure}
\centering
\includegraphics[width=0.4\textwidth]{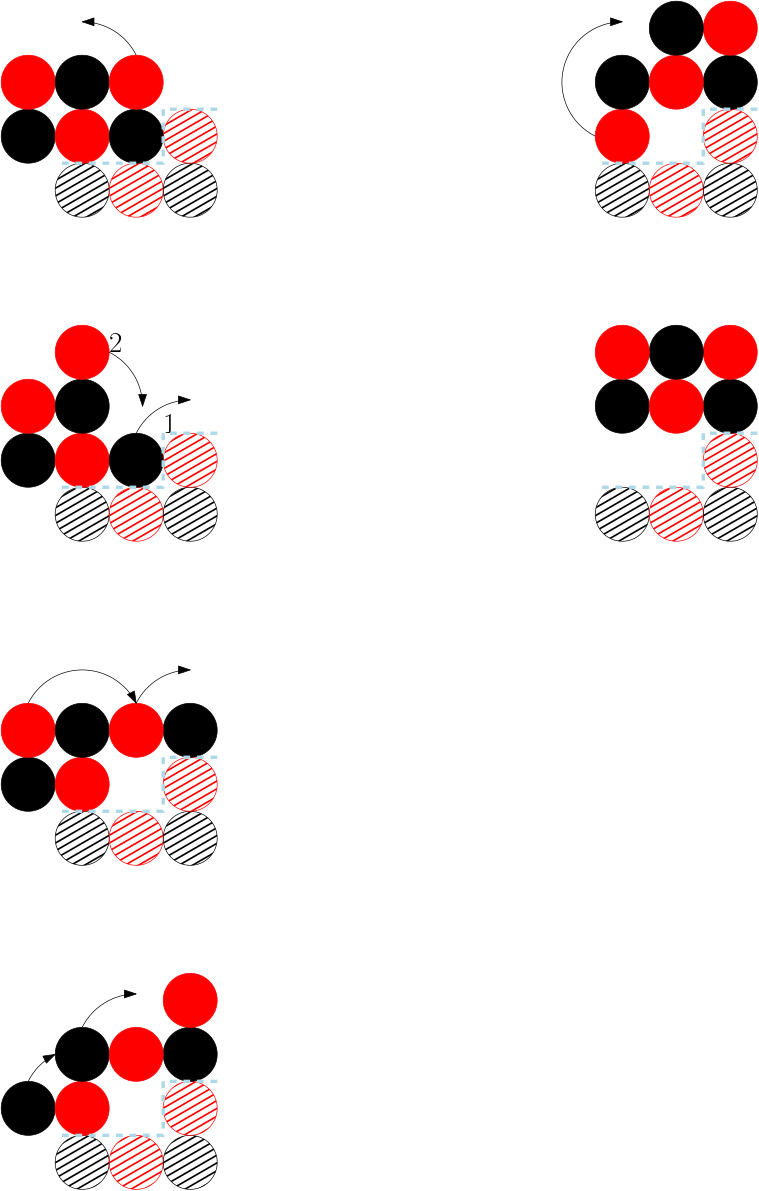}
\caption{Climbing a height 1 vertical with a width 2+ horizontal. This corresponds to the second (a) case of Figure \ref{fig:climb-scenarios}.}
\label{fig:climb1_2}
\end{figure}

\begin{figure}
\centering
\includegraphics[width=0.4\textwidth]{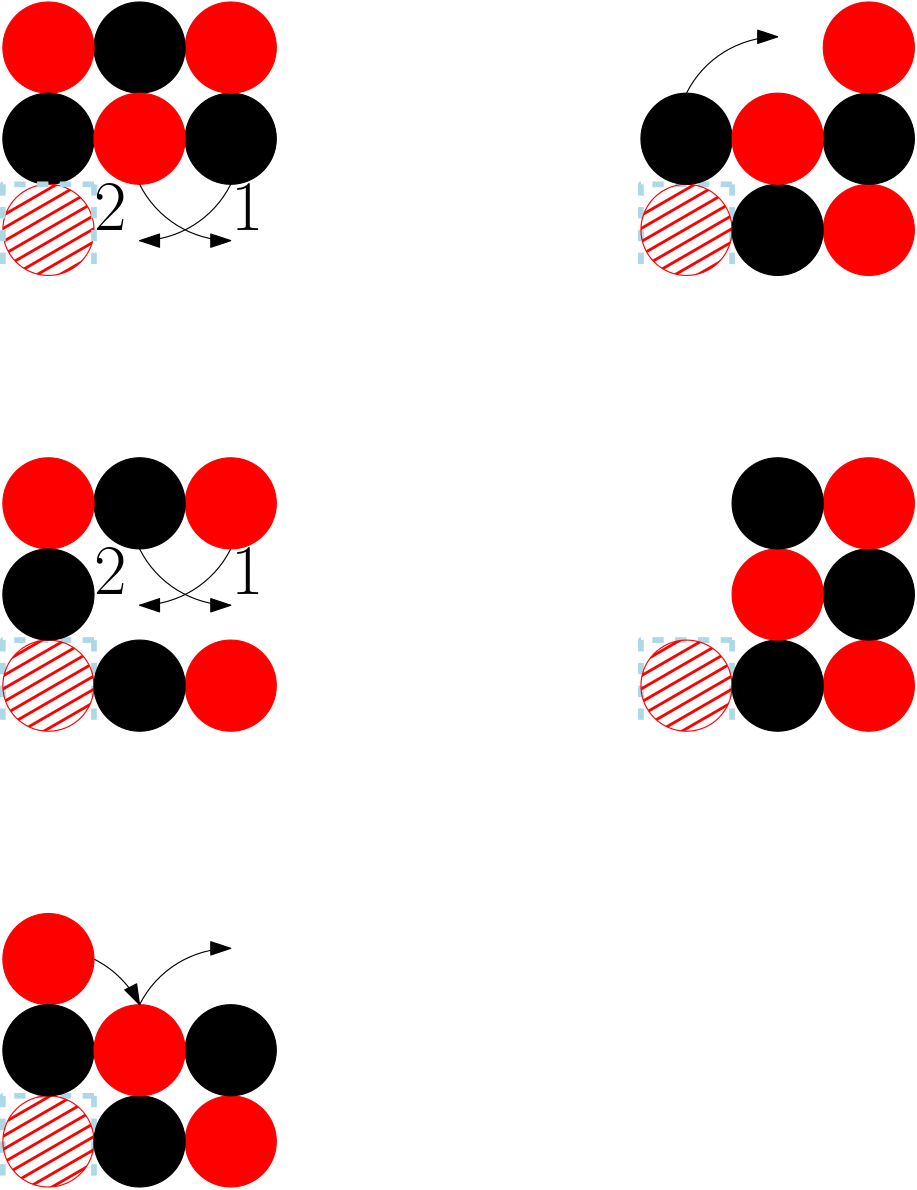}
\caption{Movement into a new quadrant consisting of a line of length 1.}
\label{fig:edge1}
\end{figure}
\FloatBarrier

\subsubsection{7-Robot Movement}
We consider once again the $up$-$right$ quadrant, and generalise to other quadrants later.
We say a cell $c=(x,y)$ is \emph{behind} the robot if $x$ is smaller than the $x$-coordinate of every node in the robot.

The \emph{load} of a 7-robot $S$ is any node $u$ such that $S\setminus \{u\}$ is a $2 \times 3$ shape.
The \emph{position} of the robot is an offset of the $y$ axis (see $y_d$ from Figure \ref{fig:variables}) for the purpose of the initial positioning of the 7-robot.
For our transformations, we maintain the invariant that the 7-robot, after any of its high-level movements, will return to the structure of a $2 \times 3$ shape with a load. For this invariant, we assume that the load is always $behind$ the $2 \times 3$ shape (while remaining connected). We will show in the proof why the situation where the load is positioned differently does not need to be considered. We therefore use $(x, y|y')$ to refer to the co-ordinates of the two cells $(x, y)$ and $(x, y')$ $behind$ the robot which can contain the load, keeping it attached to the robot while the latter is a $2 \times 3$ shape. The \emph{colouring} of a 7-robot is $good$ if the load is in the higher of the two possible positions, and $bad$ if it is in the lower position. Bad colouring usually means the resulting transformations are more difficult.

\begin{lemma}\label{lem:progress-corner-cases-7node}
For all shapes $C \in \mathcal{C}$, if a $2 \times 3$ shape with a load (the 7-robot) is placed in the cells $(x_l-3,y_d+1|y_d+2),(x_l-2,y_d+1),(x_l-1,y_d+1),(x_l,y_d+1),(x_l-2,y_d+2),(x_l-1,y_d+2),(x_l,y_d+2)$, it is capable of translating itself to $(x_r-3,y_u+1|y_u+2),(x_r-2,y_u+1),(x_r-1,y_u+1),(x_r,y_u+1),(x_r-2,y_u+2),(x_r-1,y_u+2),(x_r,y_u+2)$.
\end{lemma}

\begin{proof}
We present a series of motions which for all shapes $C \in \mathcal{C}$ lead the 7-robot from the leftmost node of the horizontal to the topmost node of the vertical. As in Lemma \ref{lem:progress-corner-cases}, we group some of these motions into high-level motions (i.e. moving the whole 7-robot by moving individual nodes).

Our first motion is \emph{sliding}, depicted in Figures \ref{fig:slide-7nodes-high} and \ref{fig:slide-7nodes-low}. The 7-robot can alternate between these two transformations to slide over the horizontal line of $C$. Alternating between the two is possible because the final configuration of each has the form required by the initial configuration of the other. The second motion is also a version of sliding, called \emph{special sliding}, depicted in Figures \ref{fig:slide-7nodes-on-low}-\ref{fig:slide-7nodes-on-high-2}. The purpose of special sliding is to bring the 7-robot ``onto'' the horizontal line, when it starts from an extreme position from which the sliding motion does not apply.
Figures \ref{fig:slide-7nodes-on-low} and \ref{fig:slide-7nodes-on-high-1} cover the cases where only the bottom-right node of the 7-robot is attached to the horizontal line, while Figure \ref{fig:slide-7nodes-on-high-2} the case where there are two points of attachment but the robot colouring is bad.
By special sliding, the 7-robot can move onto the horizontal line and sliding can then be used to move it across the horizontal, until its bottom-right node is at $(x_r-1,y_d+1)$, i.e., in one of the initial configurations of Figures \ref{fig:climb7nodes1_1} and \ref{fig:climb7nodes1_3} (disregarding the height of the vertical). This covers the horizontal part of $C$. It remains to be shown that the 7-seed can then climb up and then onto the vertical part of $C$. We do this with a motion which we refer to as \emph{climbing}, which covers a few sets of cases.

For the first set of cases, we consider the situations where the height of the vertical is 1. If the colouring is good, then we can rotate the load above the vertical. This is depicted in Figure \ref{fig:climb7nodes1_1}. Note that the load necessarily takes the higher of the two possible cells behind the robot at the start of the translation. If the colouring is bad, then the movements depend on the width of the horizontal, which can be 1 (Figure \ref{fig:climb7nodes1_2}), and 2+ (Figure \ref{fig:climb7nodes1_3}). Note that these movements always deposit the 7-robot in the position for a special slide, and the load always remains behind the robot.

For the next set of cases, we consider the situations where the height of the vertical is 2. Our first is the case where the colouring is bad. In this case, we follow the rotations of Figure \ref{fig:climb7nodes2_1}. When the colouring is good, we have two additional cases. We can follow the rotations of Figure \ref{fig:climb7nodes2_2} and Figure \ref{fig:climb7nodes2_3} to climb up and onto the vertical, respectively.

Finally, we consider the cases where the height of the vertical is at least 3. When the height is exactly 3 and the colouring is bad, we can follow Figure \ref{fig:climb7nodesUni1} to reach the top of the vertical and then Figure \ref{fig:climb7nodes2_3} to climb onto it. In the good colouring case, we can repeat the rotations of Figure \ref{fig:climb7nodes2_2} and Figure \ref{fig:climb7nodesUni2} until we reach the top of the vertical, and then perform the rotations of Figure \ref{fig:climb7nodes1_1} or Figure \ref{fig:climb7nodes2_3}, depending on whether the load is in the higher or lower of the two possible positions. We do the same for when the height is over 3 and the colouring is bad, but begin with the Figure \ref{fig:climb7nodesUni1} rotation.

We have thus shown how the 7-robot can climb up and onto verticals of any possible length. Putting everything together, starting from one of the initial positions specified by the lemma statement relative to the horizontal line of $C$, the 7-robot can use special sliding to move onto the horizontal, followed by sliding to move across it, and finally climbing to move up and onto the vertical, for all possible widths and heights of $C$. Moreover, the final position of the 7-robot relative to the vertical is as required by the statement.
\qed
\end{proof}

\begin{theorem}\label{the:7-robot-traverse}
For any orthogonal convex shape $S$, a 7-robot is capable of traversing the perimeter of $S$.
\end{theorem}

\begin{proof}
By Lemma \ref{lem:progress-corner-cases-7node} we have shown for the up-right quadrant firstly that it is possible to slide across a horizontal of arbitrary width no matter the robot's initial position, that it is possible to climb a height 2 vertical, that part of this movement can be repeated indefinitely to climb verticals of arbitrary height, that special movements exist for the height 1 vertical and that all of this is possible no matter how long the horizontal line the object lies on is, nor whether the 7-robot is red or black.

By rotating the robot and the quadrant as necessary, we are able to replicate our movements for all other quadrants: $d_4(d_4\;|\;d_1)^{*}$ (the left-up quadrant) $d_2(d_2\;|\;d_3)^{*}$ (the right-down quadrant) and $d_3(d_3\;|\;d_4)^{*}$ (the left-down quadrant). All that remains is the transition between the quadrants.

There are two cases. In the first case, the next quadrant consists of multiple lines. In this case, when the line signifying the end of the current quadrant is met it is sufficient to begin movements appropriate to travelling in the next quadrant. However, there is an edge case where a quadrant consists of a single line. In this case, a unique movement is necessary (see Figure \ref{fig:7nodes-edge-red} and Figure \ref{fig:7nodes-edge-black}) to transfer the 7-node object onto the line, with the exact movement depending on whether the first node of the line is the same or a different colour from the load. These movements are then followed by special slides to put the object into position for the next quadrant. Naturally, these transformations are reversible and can be mirrored as well. We are therefore able to deal with any quadrant transition, even rotating the 7-node object around a single node.

Therefore, because we can move through any variant of all quadrants and transition between them, a 7-robot can traverse the perimeter of any orthogonal convex shape.
\qed
\end{proof}

\FloatBarrier

\begin{figure}
\centering
\includegraphics[width=0.45\linewidth]{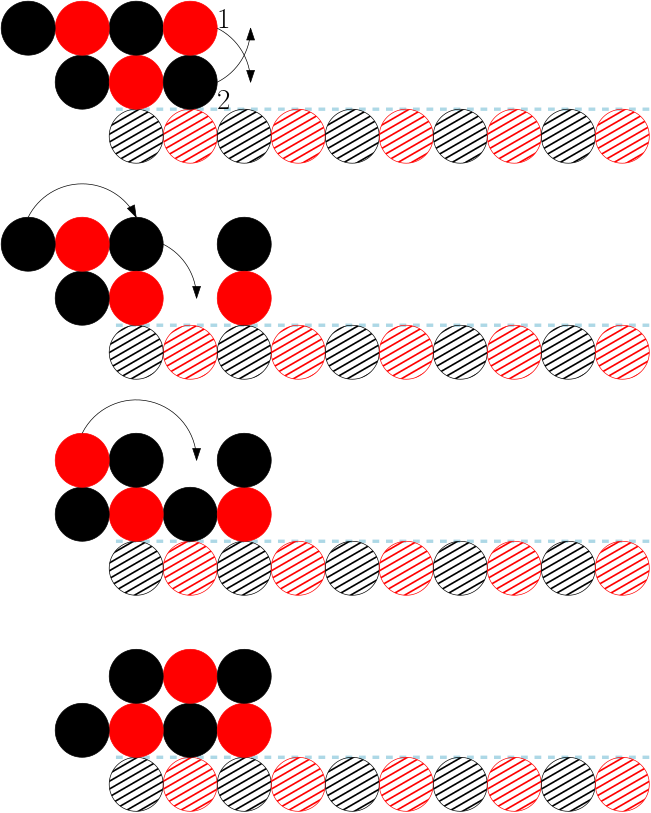}
\caption{Sliding across a line with a 7-node robot - case 1}
\label{fig:slide-7nodes-high}
\end{figure}

\begin{figure}
\includegraphics[width=0.4\textwidth]{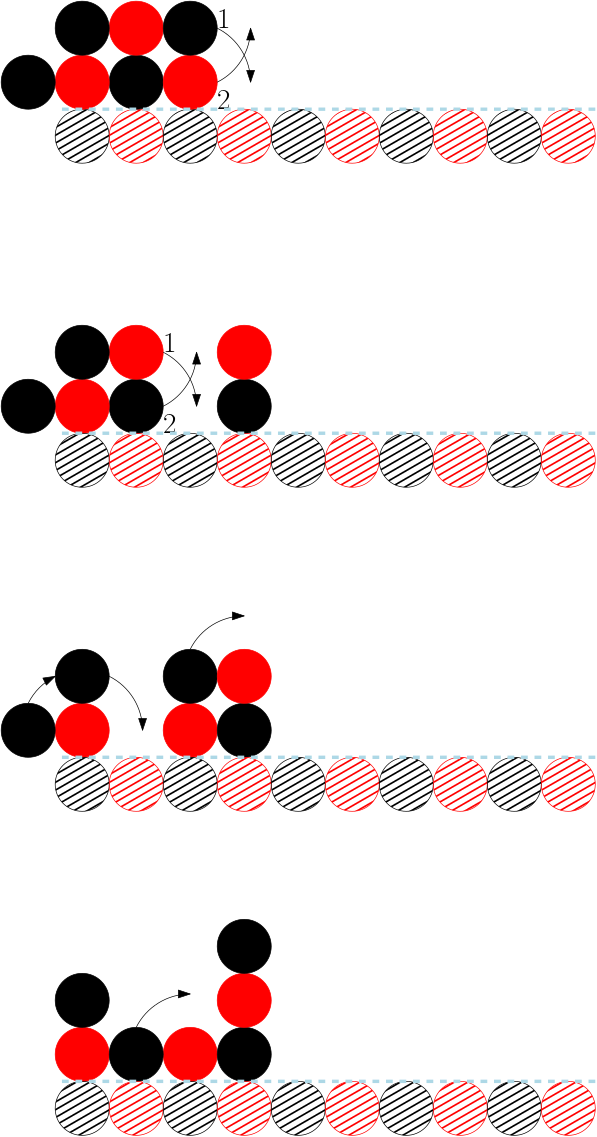}
\hspace{1.5cm}
\includegraphics[width=0.4\textwidth]{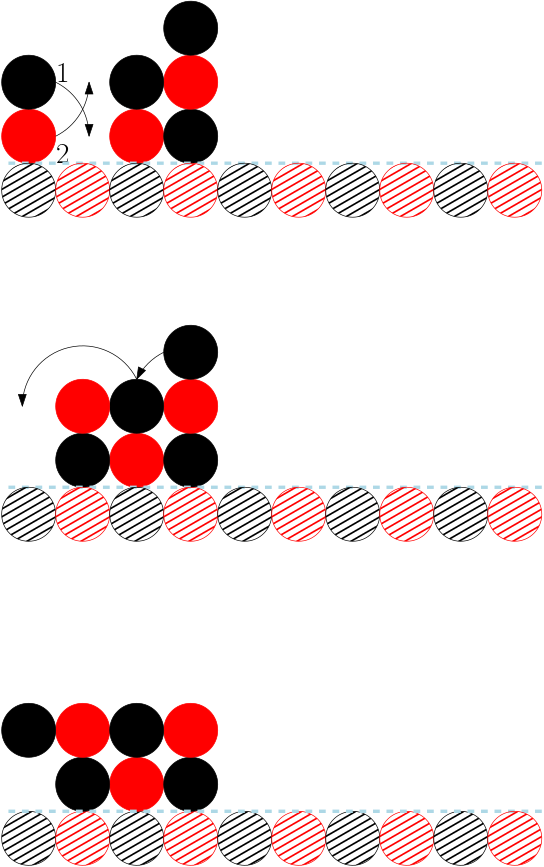}
\caption{Sliding across a line with a 7-node robot - case 2}
\label{fig:slide-7nodes-low}
\end{figure}

\begin{figure}
\includegraphics[width=0.4\linewidth]{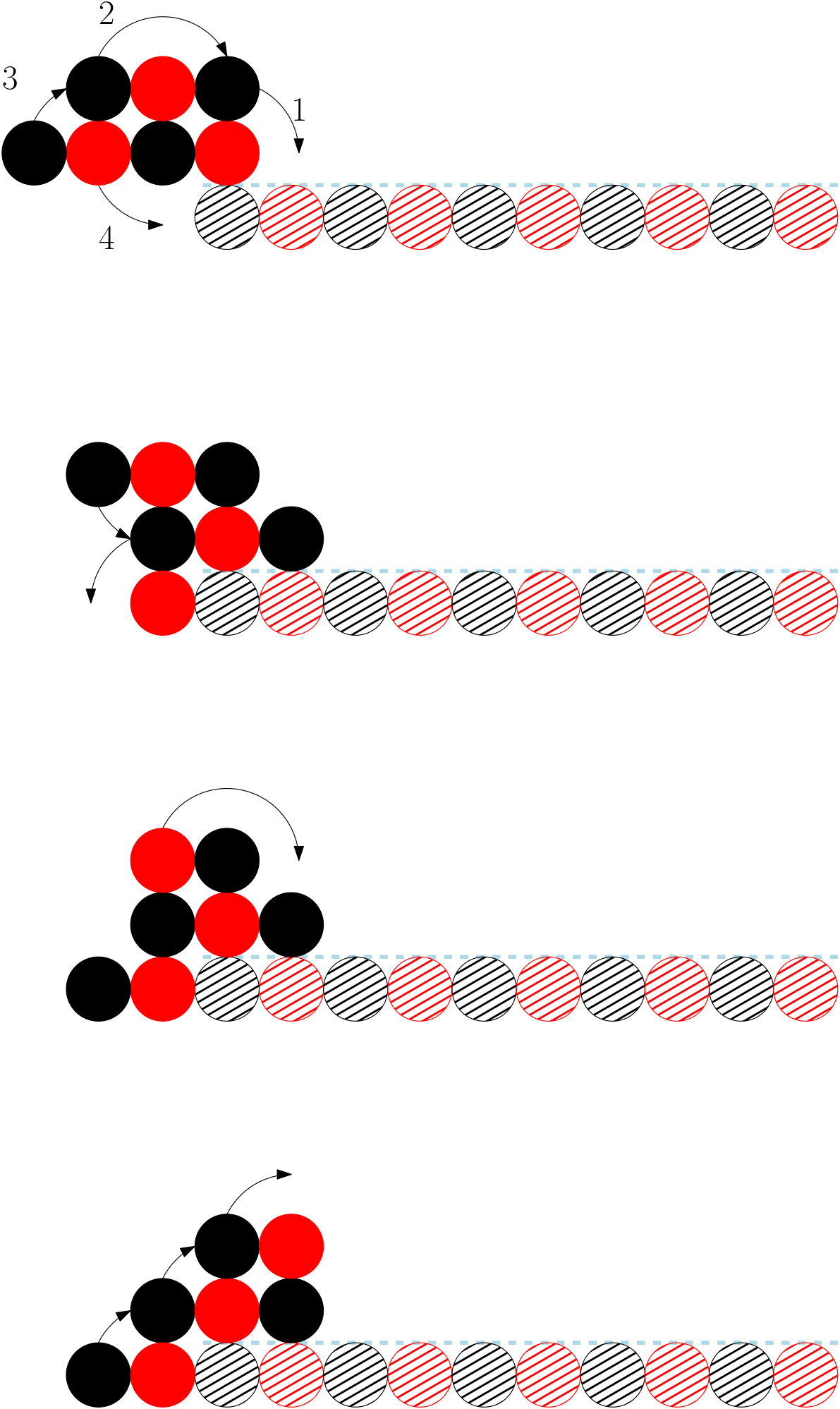}
\hspace{1.5cm}
\includegraphics[width=0.4\linewidth]{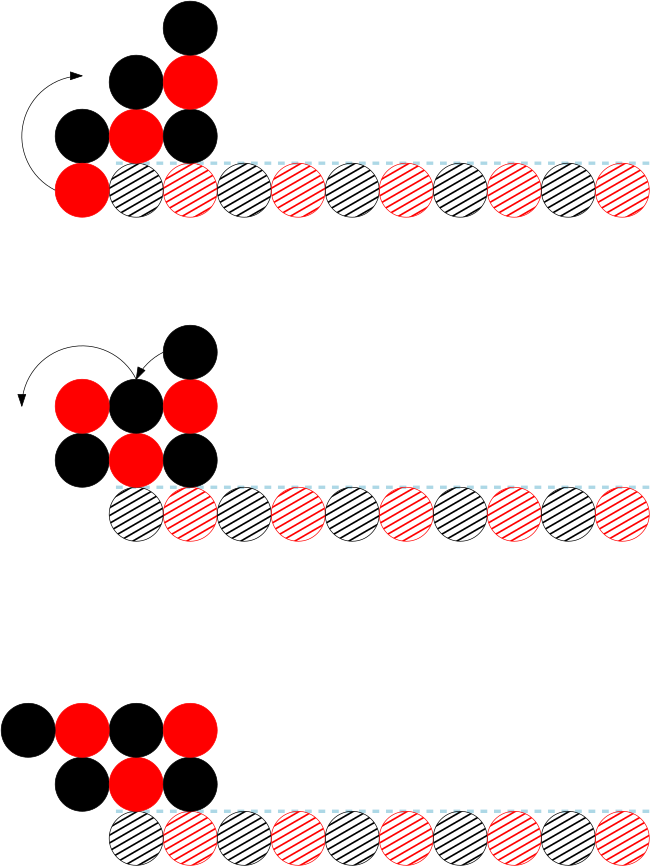}
\caption{Sliding on a line with a 7-node robot - case 1}
\label{fig:slide-7nodes-on-low}
\end{figure}

\begin{figure}
\includegraphics[width=0.4\textwidth]{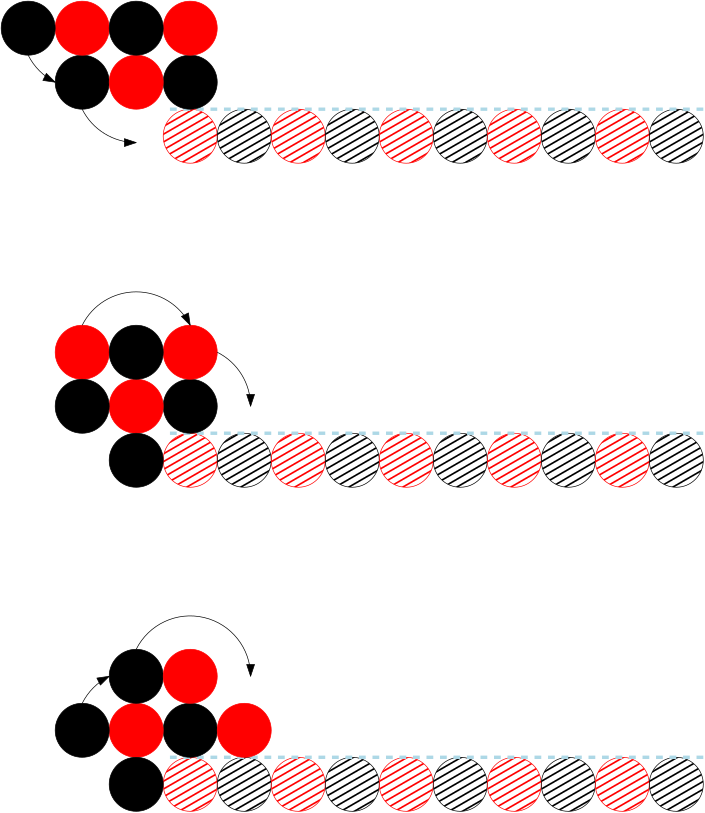}
\hspace{1.5cm}
\includegraphics[width=0.4\textwidth]{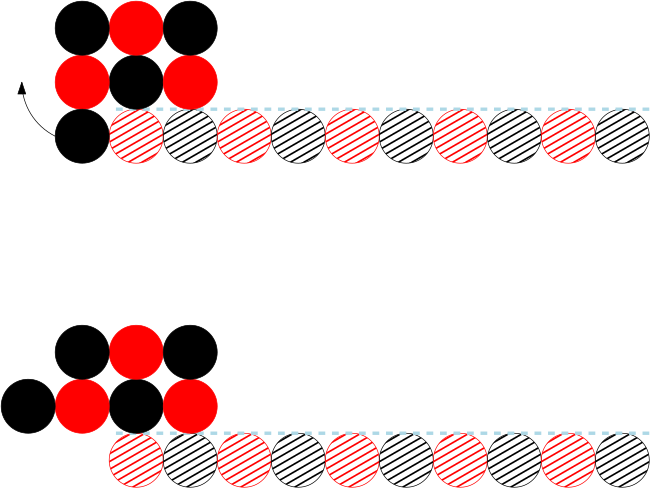}
\vspace{1cm}
\caption{Sliding on a line with a 7-node robot - case 2}
\label{fig:slide-7nodes-on-high-1}
\end{figure}

\begin{figure}
\includegraphics[width=0.4\textwidth]{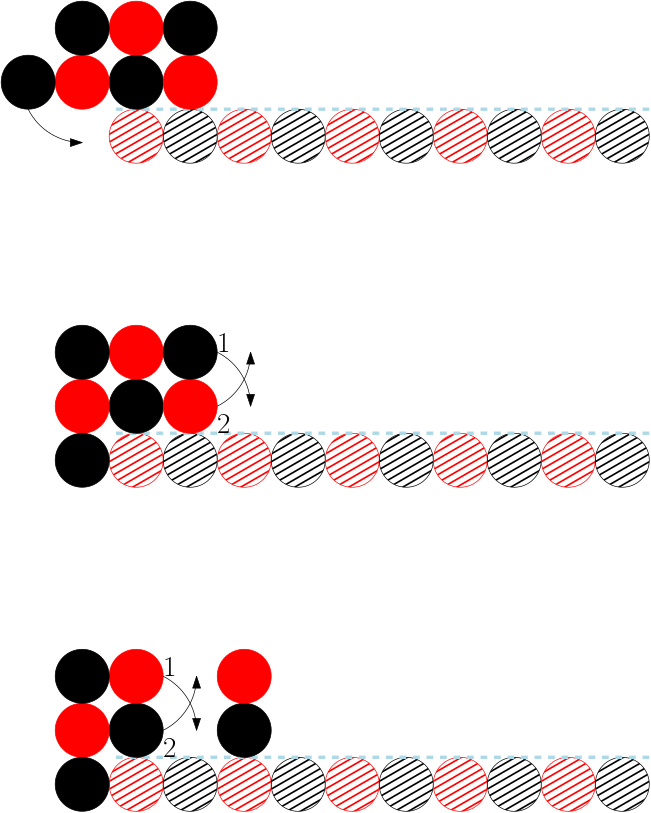}
\hspace{1.5cm}
\includegraphics[width=0.4\textwidth]{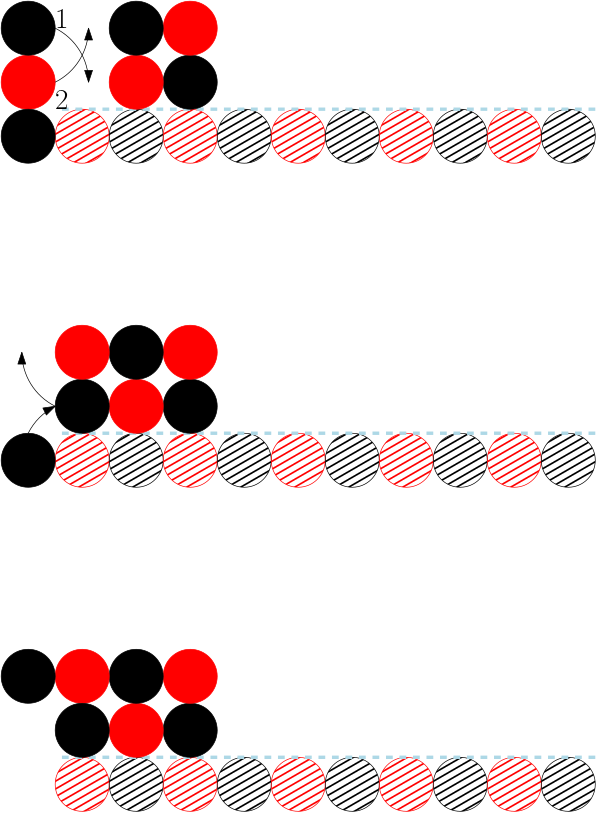}
\vspace{1cm}
\caption{Sliding on a line with a 7-node robot - case 3. The transformation of Figure \ref{fig:slide-7nodes-high} can be applied for further movement.}
\label{fig:slide-7nodes-on-high-2}
\end{figure}

\begin{figure}
\centering
\includegraphics[width=0.45\linewidth]{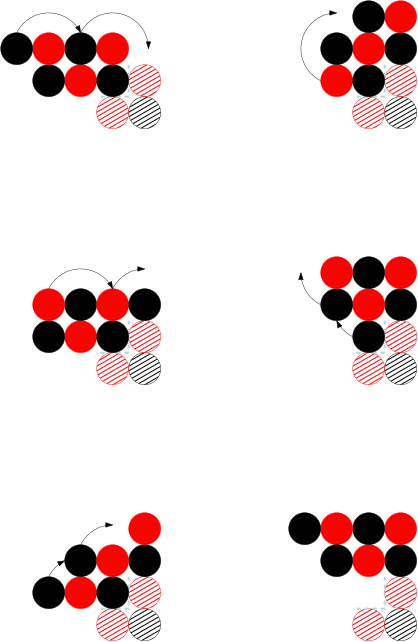}
\caption{Climbing on top of a vertical when the load is in the upper cell.}
\label{fig:climb7nodes1_1}
\end{figure}

\begin{figure}
\centering
\includegraphics[width=0.45\linewidth]{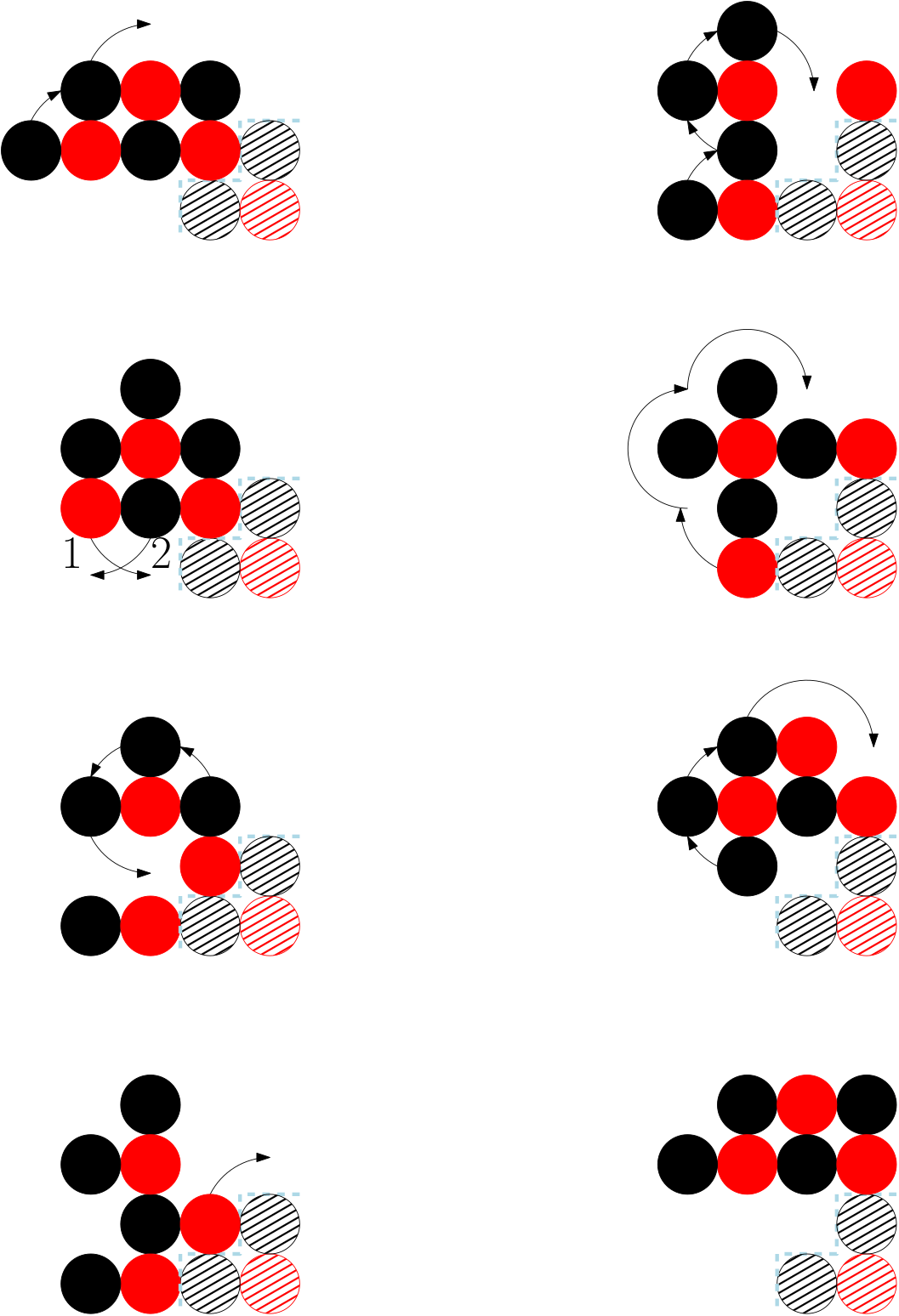}
\caption{Climbing a height 1 vertical with a width 1 horizontal and bad colouring}
\label{fig:climb7nodes1_2}
\end{figure}

\begin{figure}
\centering
\includegraphics[width=0.45\linewidth]{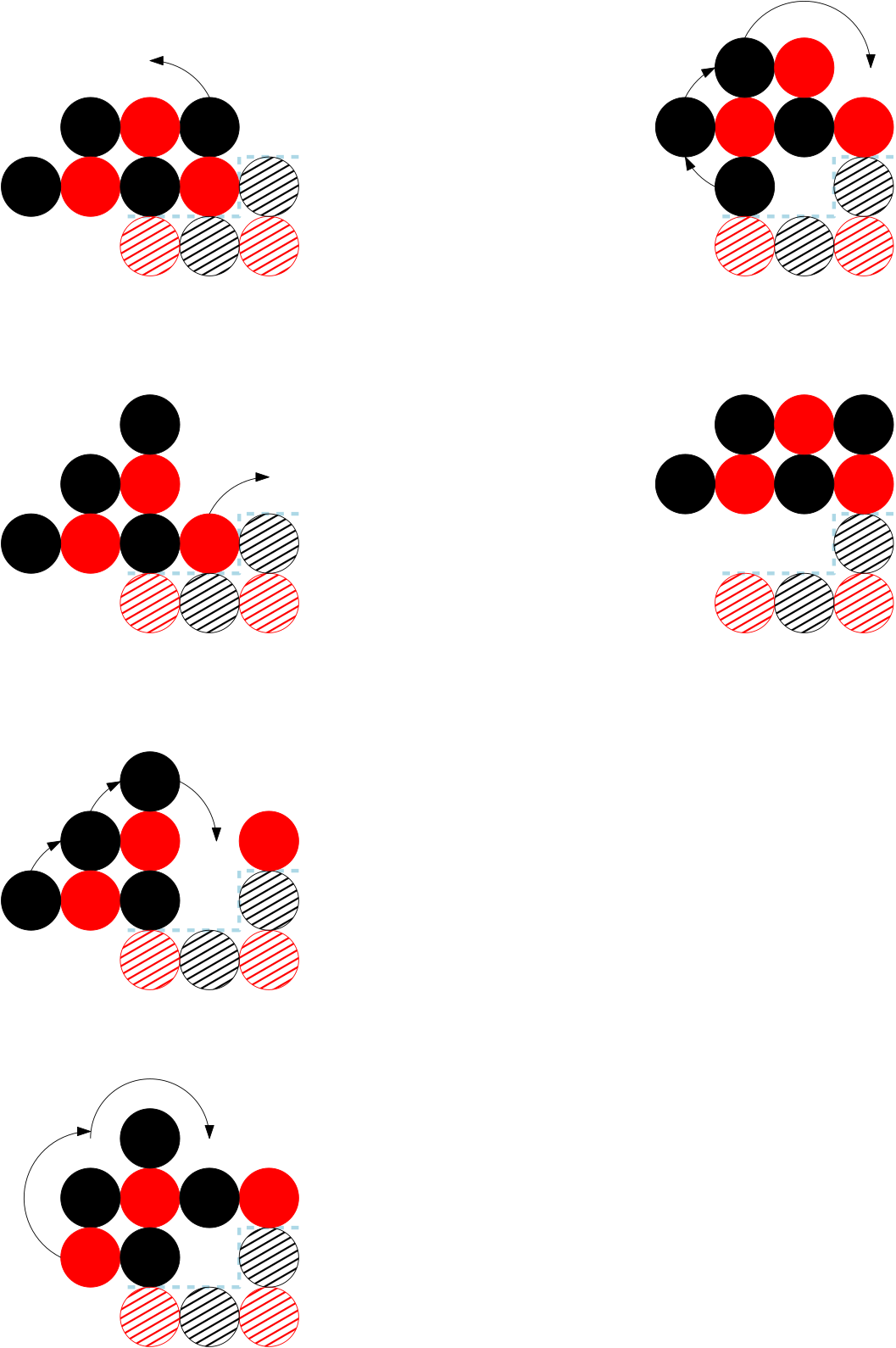}
\caption{Climbing a height 1 vertical with a width 2+ horizontal and bad colouring}
\label{fig:climb7nodes1_3}
\end{figure}

\begin{figure}
\centering
\includegraphics[width=0.45\linewidth]{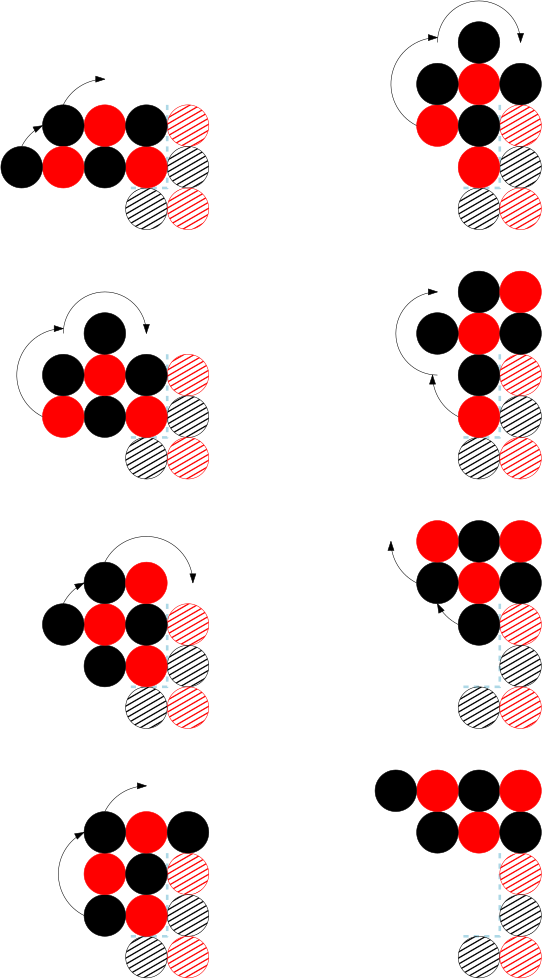}
\caption{Climbing on top of a vertical of height $2$ from position 0 with bad colouring.}
\label{fig:climb7nodes2_1}
\end{figure}

\begin{figure}
\centering
\includegraphics[width=0.45\linewidth]{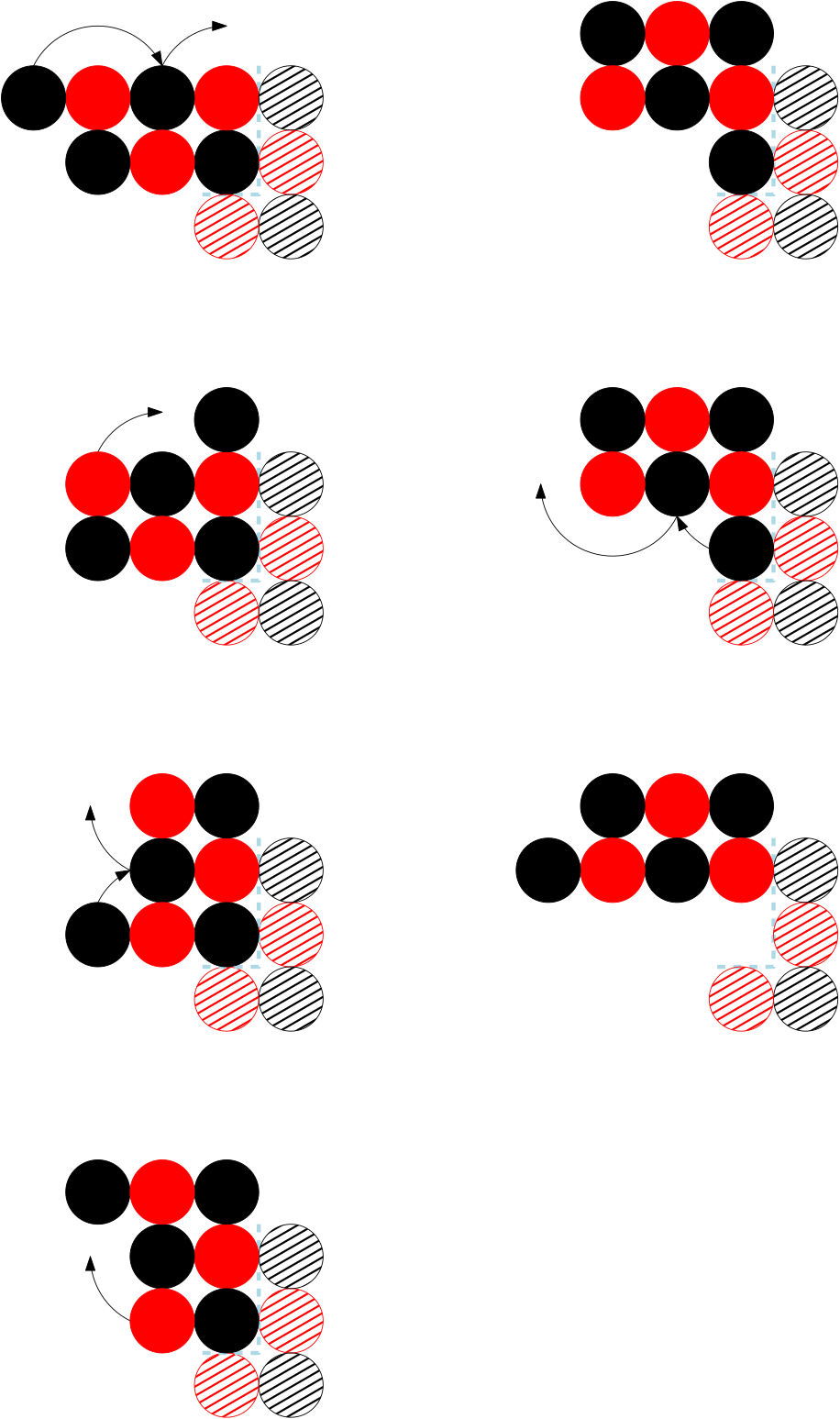}
\caption{Climbing a vertical of height 2+}
\label{fig:climb7nodes2_2}
\end{figure}

\begin{figure}
\centering
\includegraphics[width=0.45\linewidth]{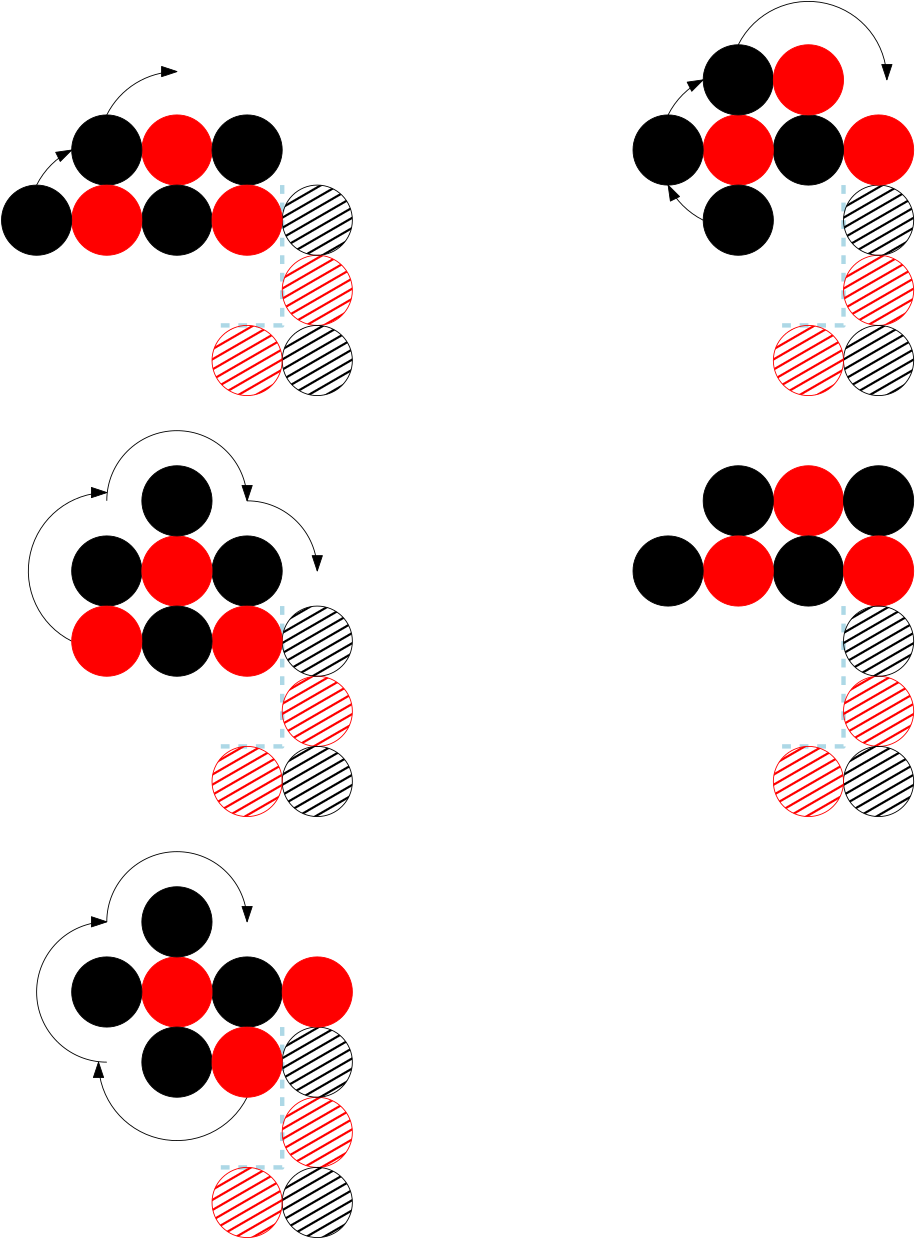}
\caption{Climbing on top of the vertical from position 1 when the load is in the lower cell.}
\label{fig:climb7nodes2_3}
\end{figure}

\begin{figure}
\centering
\includegraphics[width=0.45\linewidth]{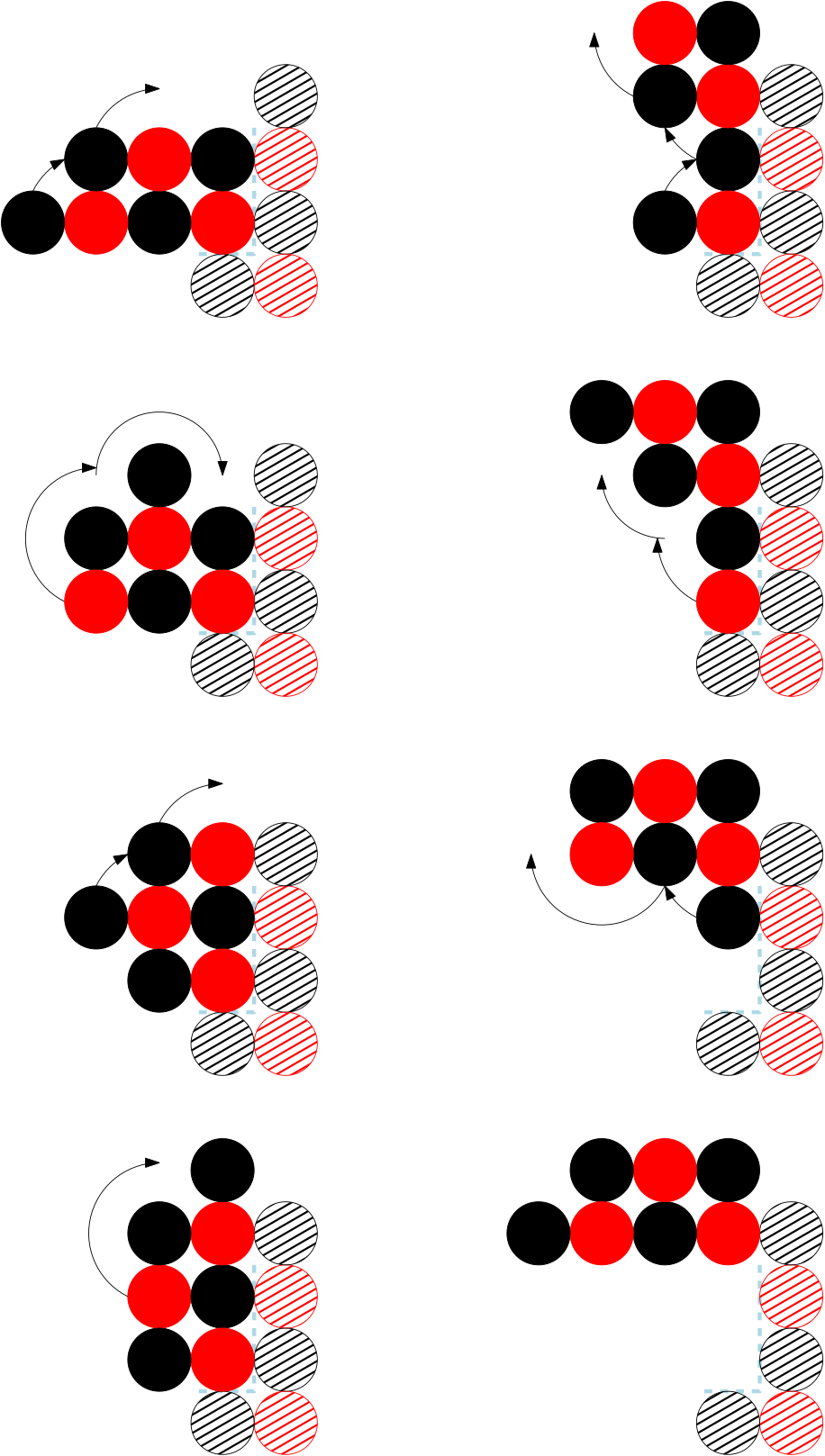}
\caption{Climbing a vertical of height 3 from position 0 with bad colouring.}
\label{fig:climb7nodesUni1}
\end{figure}

\begin{figure}
\centering
\includegraphics[width=0.45\linewidth]{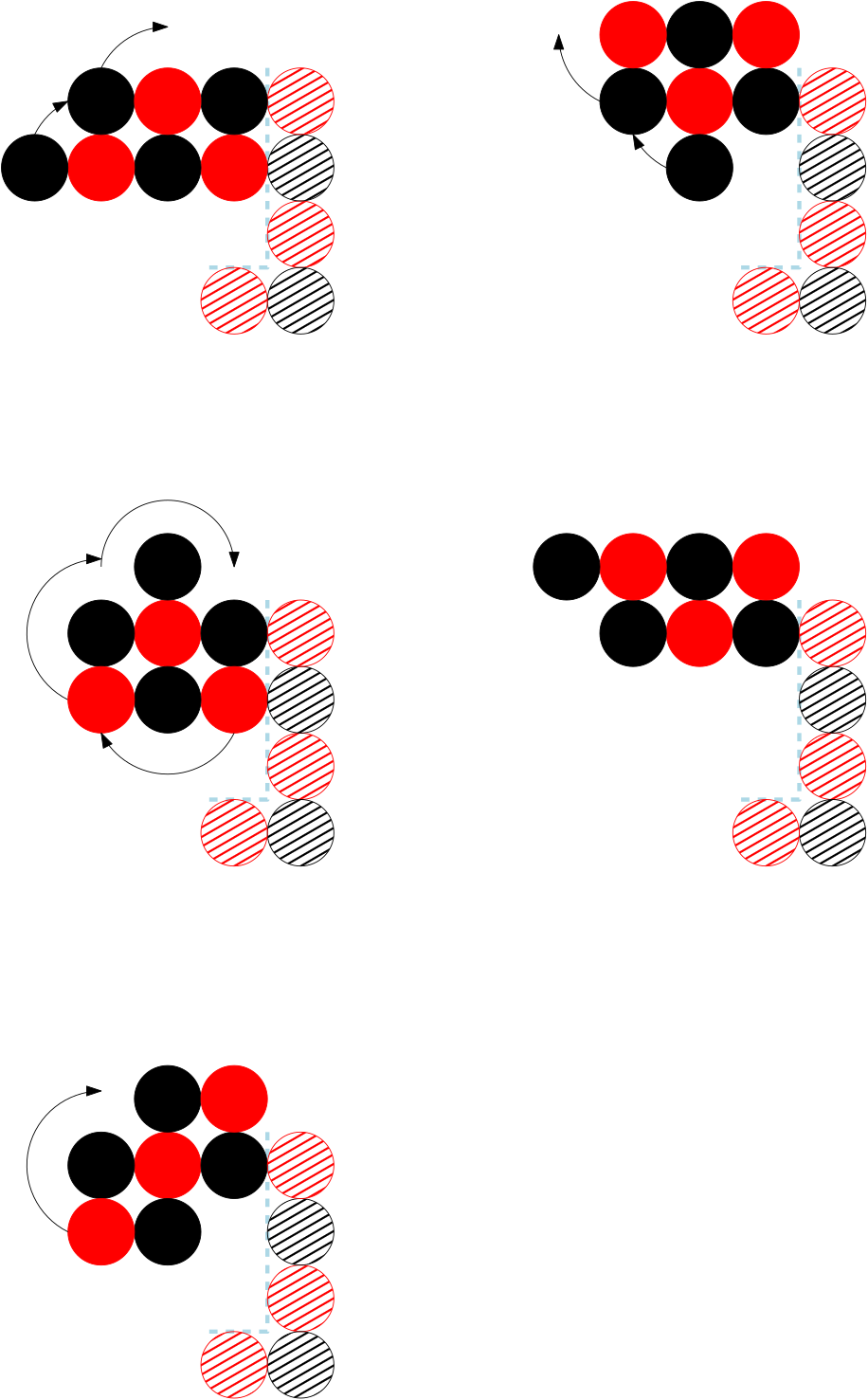}
\caption{Climbing a vertical of height 3+ from position 2+.}
\label{fig:climb7nodesUni2}
\end{figure}

\begin{figure}
\centering
\includegraphics[width=0.45\textwidth]{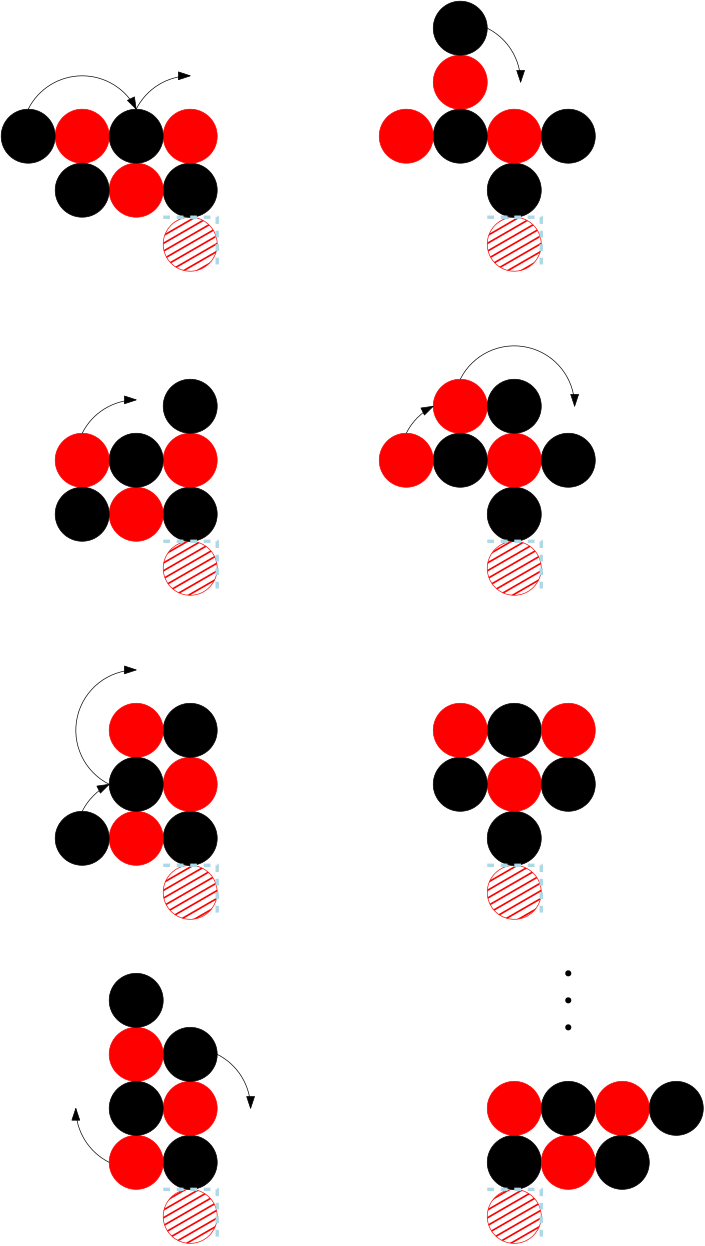}
\centering
\hspace{0.25cm}
\includegraphics[width=0.4\linewidth]{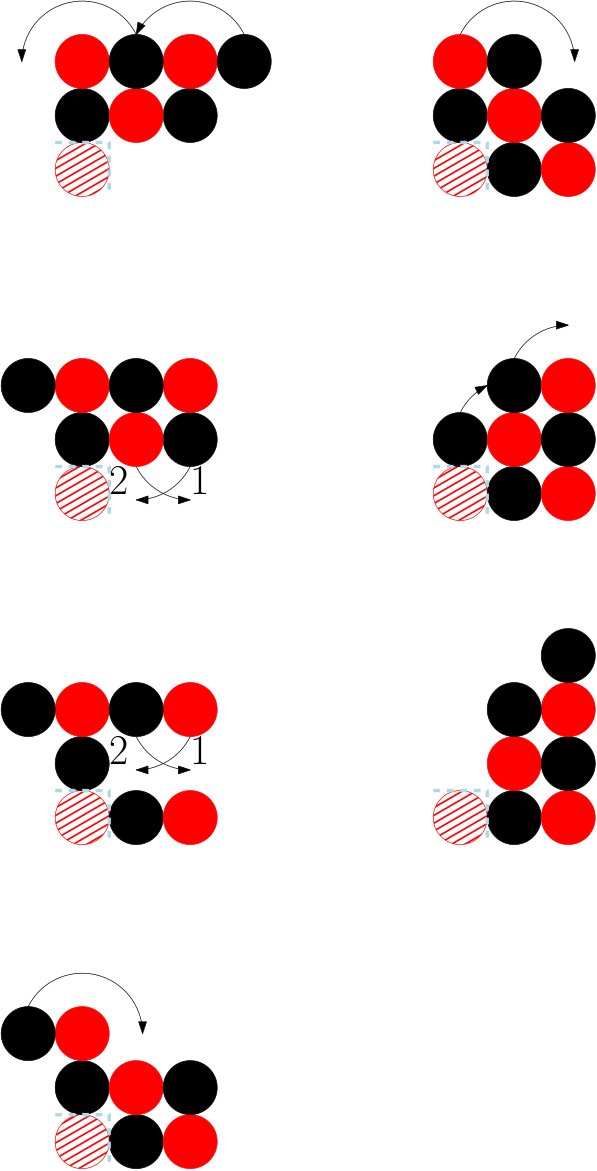}
\caption{Movement into a new quadrant consisting of a line of length 1 when the perimeter node is not the same colour of the load. To reach the configuration after the dots, the operation is repeated in an inverted manner.}
\label{fig:7nodes-edge-red}
\end{figure}

\begin{figure}
\centering
\includegraphics[width=0.45\textwidth]{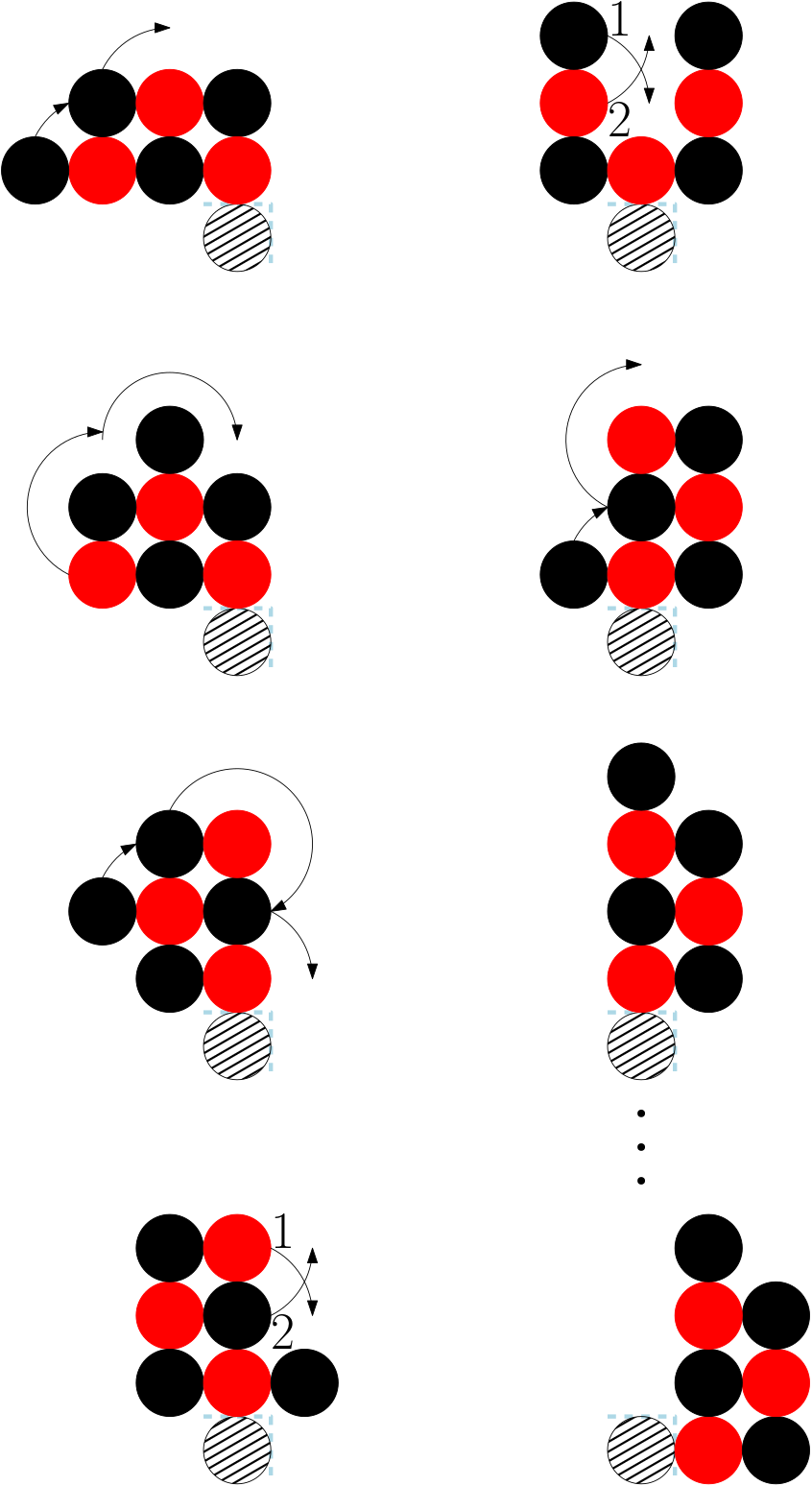}
\caption{Movement into a new quadrant consisting of a line of length 1, with the perimeter node the same colour as the load. To reach the configuration after the dots, we follow a rotated version of the transformation in Figure \ref{fig:climb7nodes2_3}.}
\label{fig:7nodes-edge-black}
\end{figure}
\FloatBarrier

\subsubsection{Repository Traversal}

Whenever the single-black repository is occupied, the robot may need to traverse a non-convex region when moving between $S$ and the extended staircase. The following lemma shows that this is not an issue.

\begin{lemma}\label{lem:black-repository}
If the single-black repository of the extended staircase is occupied, then both the 6-robot and the 7-robot are able to traverse past it.
\end{lemma}

\begin{proof}
We present a series of motions which for all variants of the shapes $C \in \mathcal{C}$ created by the addition of a node to the left of the vertical lead both the 6-robot and the 7-robot to the point where further movement to the topmost node of the vertical is equivalent to movement across an orthogonal convex shape. We refer to the \emph{gap} as the set of empty cells between the single-black repository and the node below it. As before, we group some of these motions into high-level motions (i.e. moving the whole robot by moving individual nodes). When there is no node below the single-black repository, then the shape is orthogonal convex and therefore, by Theorem \ref{the:6-robot-traverse} and Theorem \ref{the:7-robot-traverse}, the robot can traverse past the single-black.

Our first motion is \emph{sliding}. If there is a 1-cell gap, then the robot can still traverse past the cell by using the special slide (Figure \ref{fig:slide-on}) for the 6-robot and the slides (Figure \ref{fig:slide-7nodes-high} and Figure \ref{fig:slide-7nodes-low}) for the 7-robot, because these movements do not depend on the existence of a node in $(x_l - 2, y_d - 1)$ to provide connectivity. If there is a gap of 2 or more cells, then the robot must use special movements (Figure \ref{fig:repo-6node-slide}, Figure \ref{fig:repo-7node-slide-high} and Figure \ref{fig:repo-7node-slide-low}) to traverse past it. The next motion is \emph{climbing}. There are special movements (Figure \ref{fig:repo-6node-climb}) which allow the 6-robot to climb a gap of size 2, 3 and 4+. There are more movements for the 7-robot, when the gap is of size 2 (Figure \ref{fig:repo-7node-climb-1}), 3 (Figure \ref{fig:repo-7node-climb-2}) and 4+ (Figure \ref{fig:repo-7node-climb-3}).

We have thus shown how both types of robot can climb up and onto verticals with the additional node of any possible length. Putting everything together, the robot can use special sliding to move across a 1-cell gap, followed by special motions for larger gaps, and finally another set of special motions to climb a vertical, for all possible widths and heights. Moreover, the final position of the robot allows for the resumption of regular movement i.e. those which allow the robot to traverse an orthogonal convex shape.
\qed
\end{proof}

\FloatBarrier

\begin{figure}
\centering
\begin{subfigure}{.5\textwidth}
\centering
\includegraphics[width=0.75\linewidth]{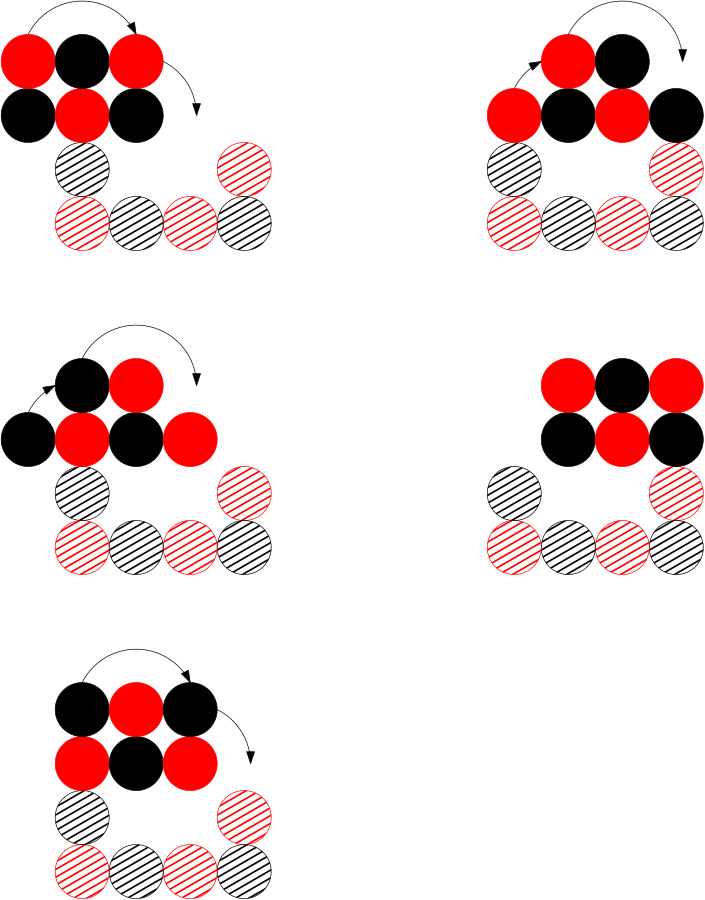}
\end{subfigure}%
\begin{subfigure}{.5\textwidth}
\centering
\includegraphics[width=0.75\linewidth]{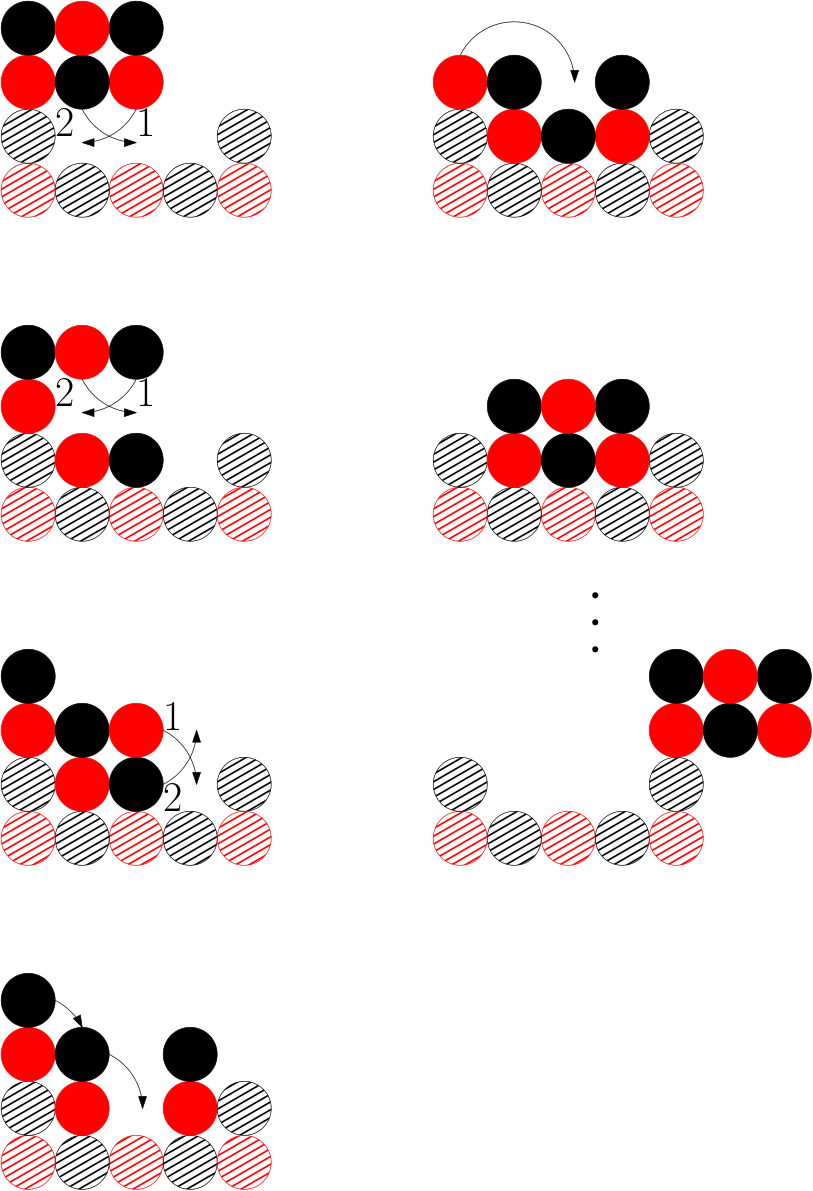}
\end{subfigure}
\caption{Sliding a 6-robot on a line with a black repository with a gap of size 2 and 3+. All movement after the final positions is equivalent to orthogonal convex movement.}
\label{fig:repo-6node-slide}
\end{figure}

\begin{figure}
\centering
\includegraphics[width=0.45\linewidth]{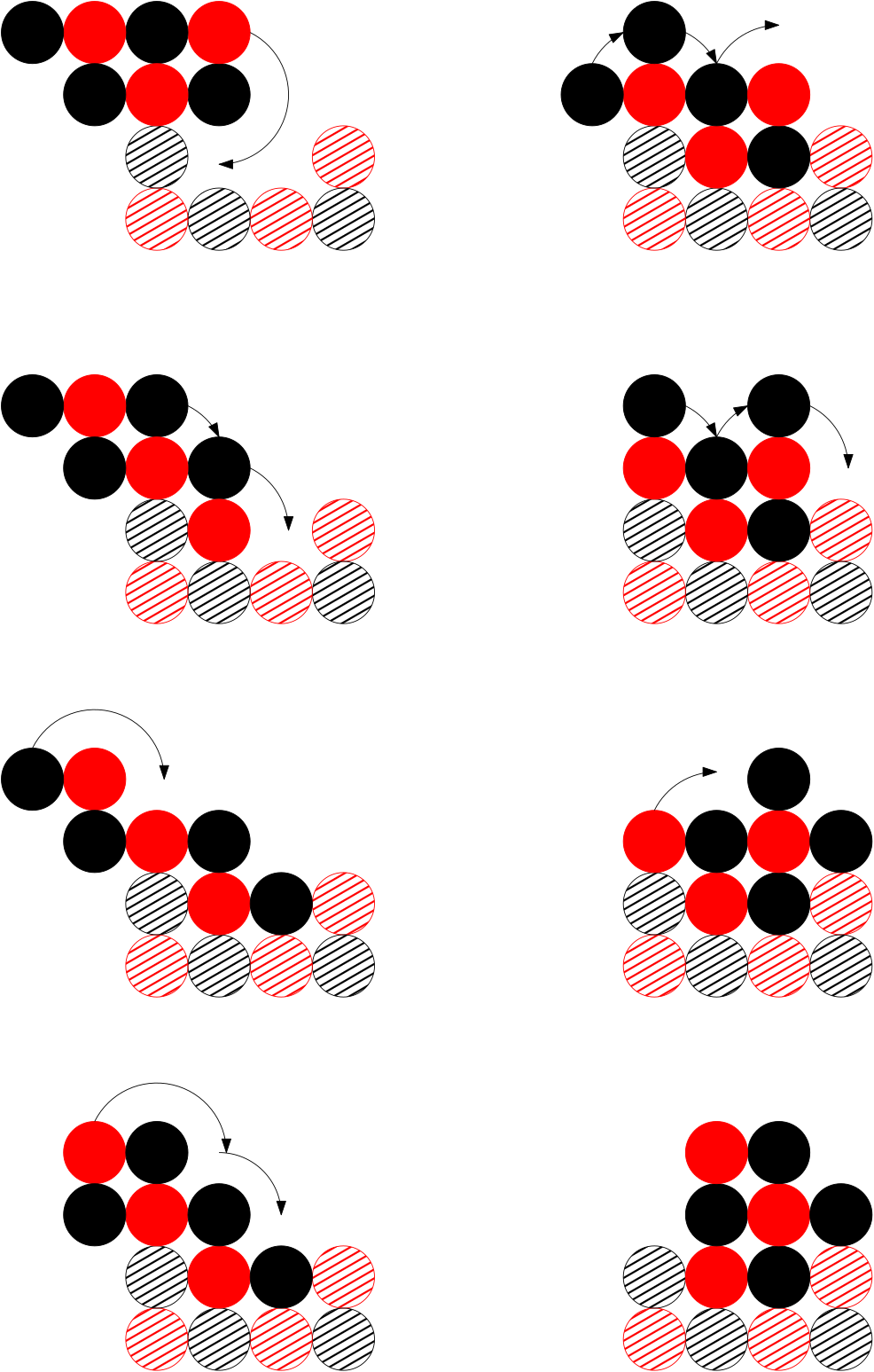}
\centering
\includegraphics[width=0.45\linewidth]{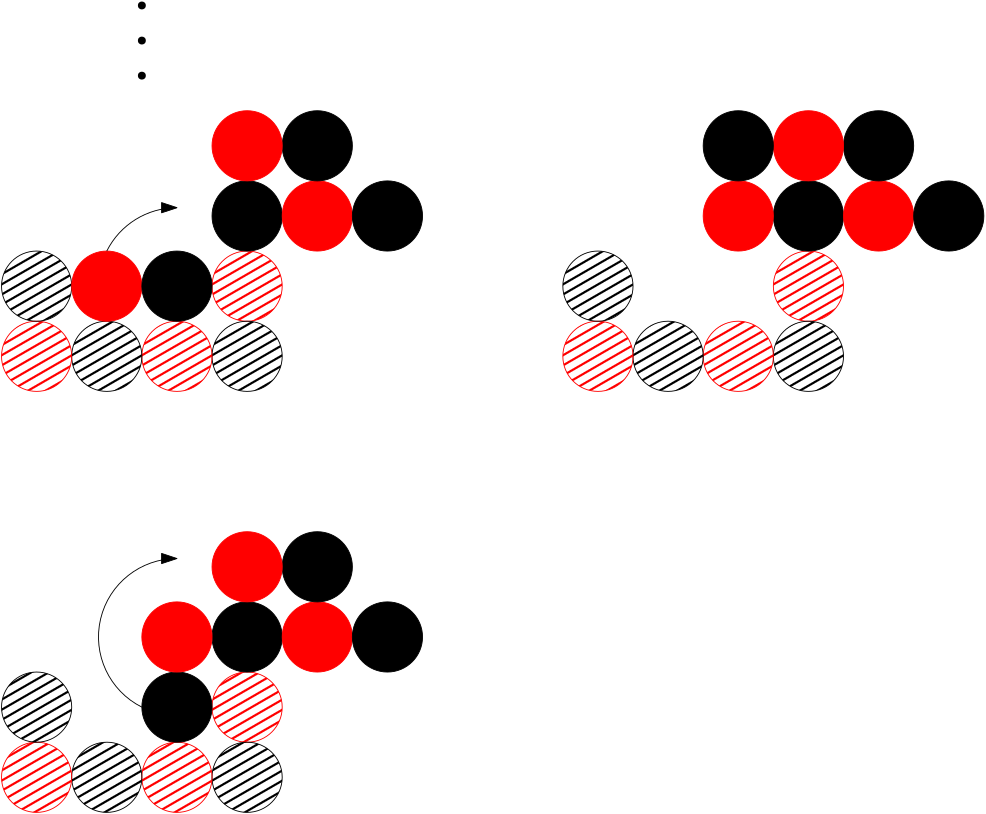}
\centering
\includegraphics[width=0.45\linewidth]{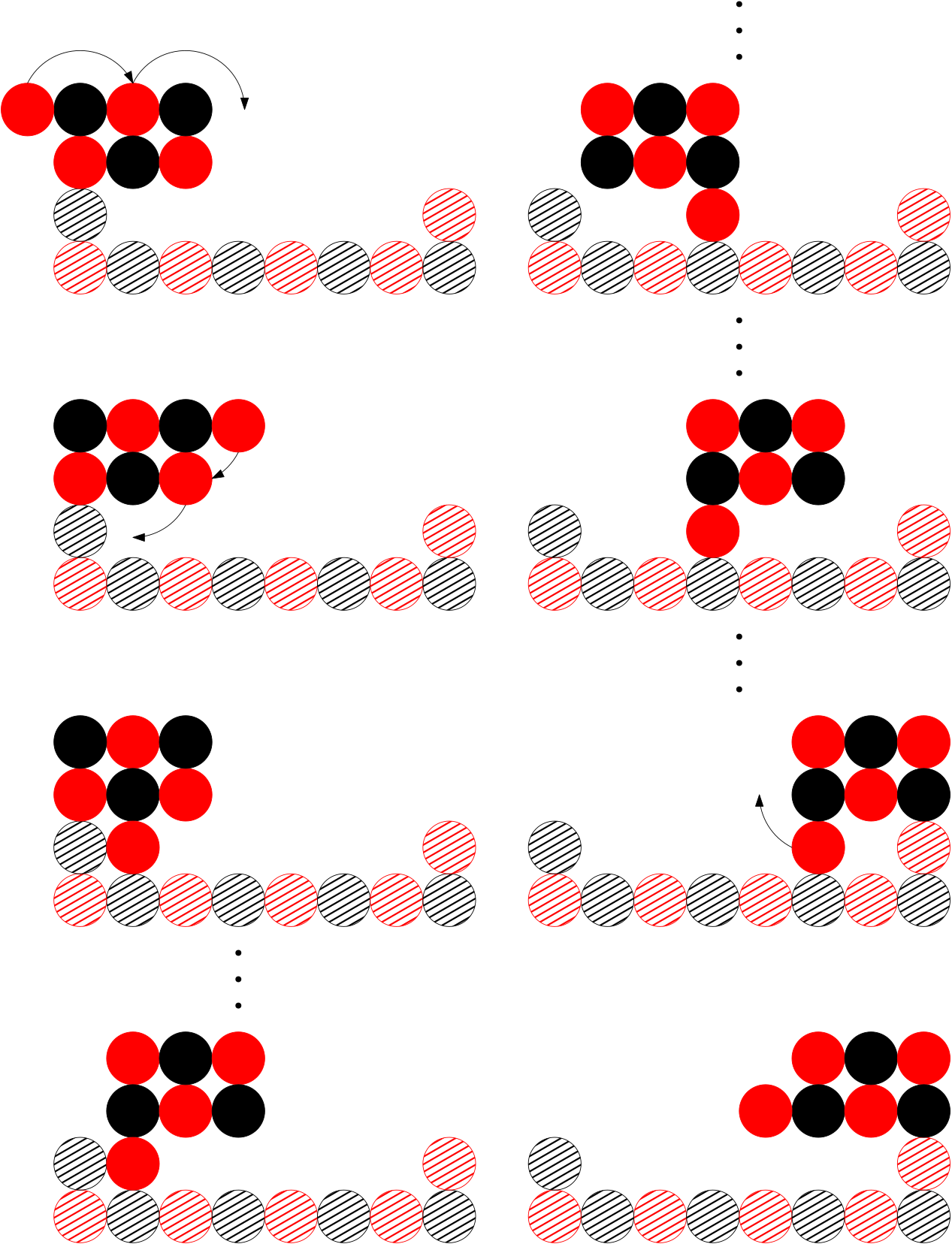}
\caption{Sliding a 7-robot on a line with a black repository with the load in the high position, with a gap of size 2 and 3+. Note the movement after the dots can be repeated for gaps larger than 3.}
\label{fig:repo-7node-slide-high}
\end{figure}

\begin{figure}
\centering
\includegraphics[width=0.45\linewidth]{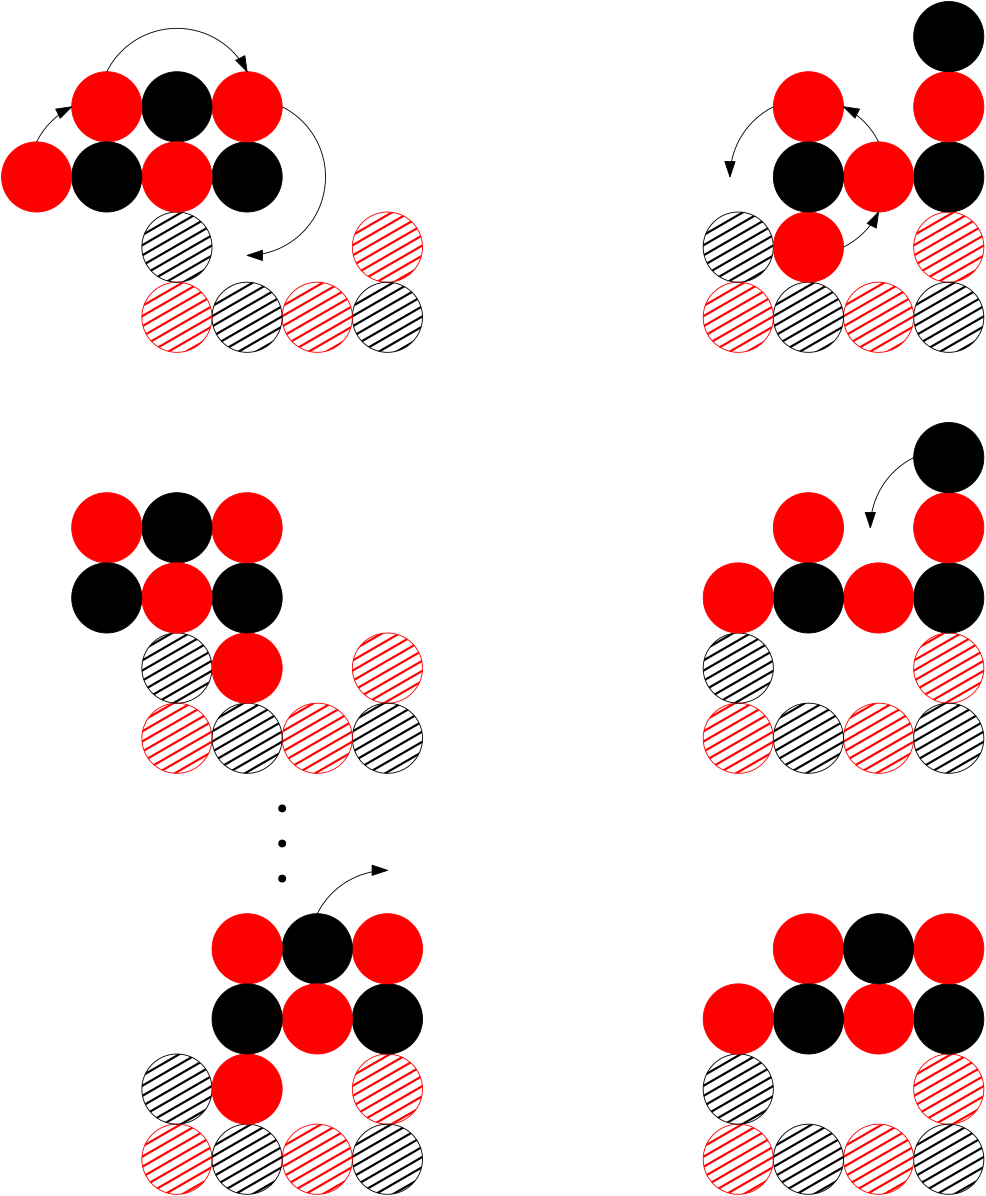}
\centering
\includegraphics[width=0.45\linewidth]{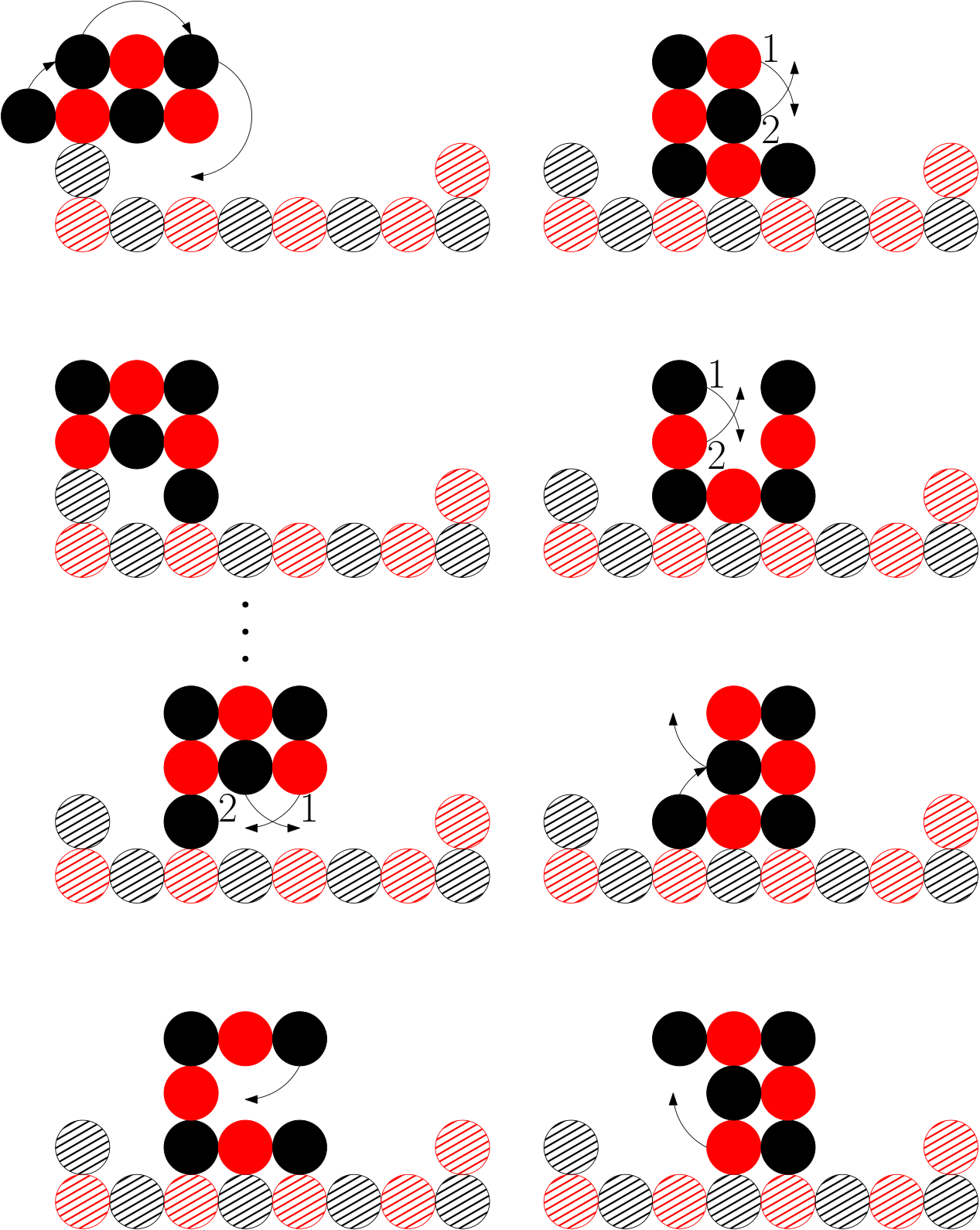}
\centering
\includegraphics[width=0.45\linewidth]{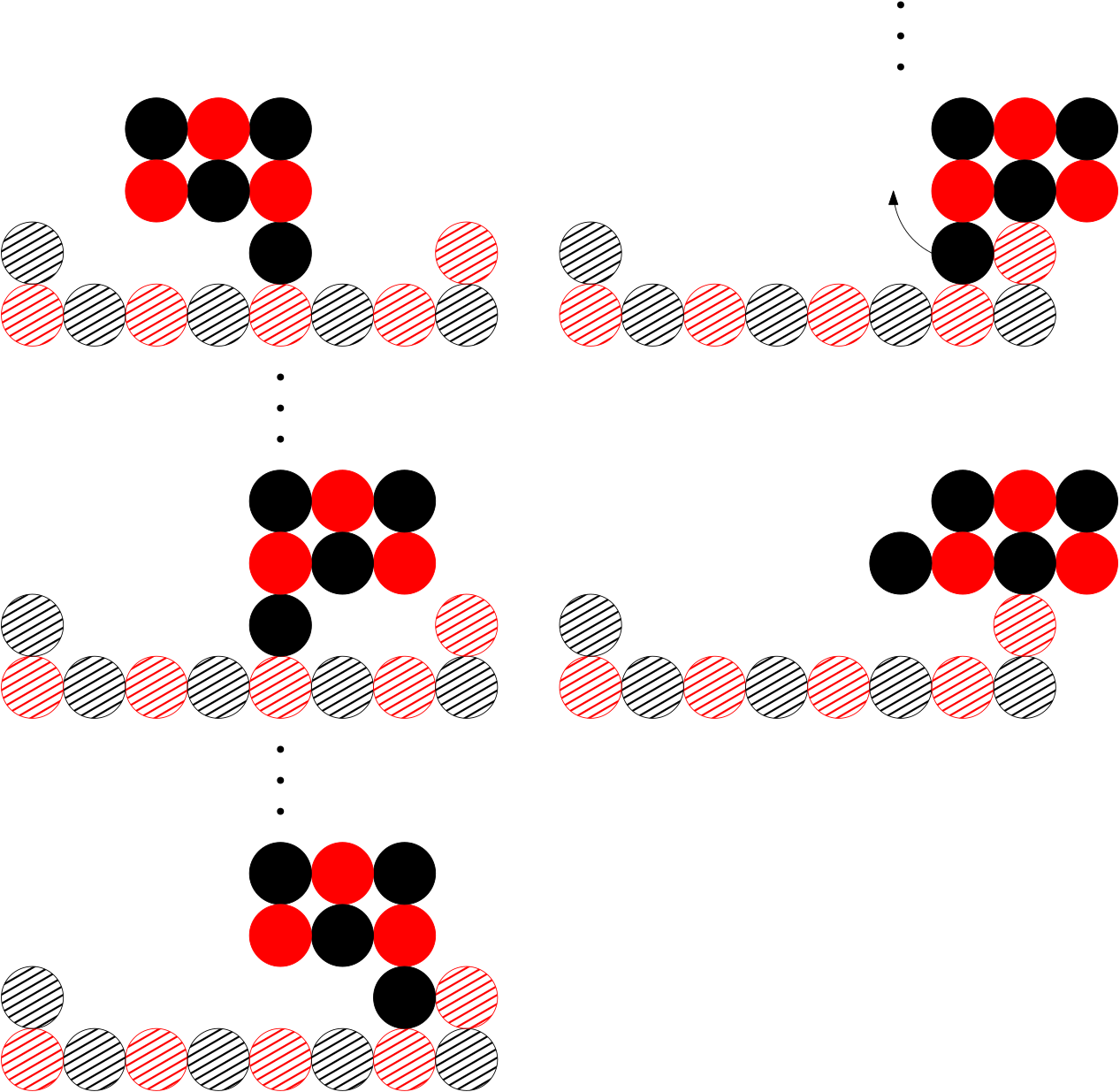}
\caption{Sliding a 7-robot on a line with a black repository with the load in the low position, with a gap of size 2 and 3+. Note the movement after the dots can be repeated for gaps larger than 3.}
\label{fig:repo-7node-slide-low}
\end{figure}

\begin{figure}
\centering
\includegraphics[width=0.45\linewidth]{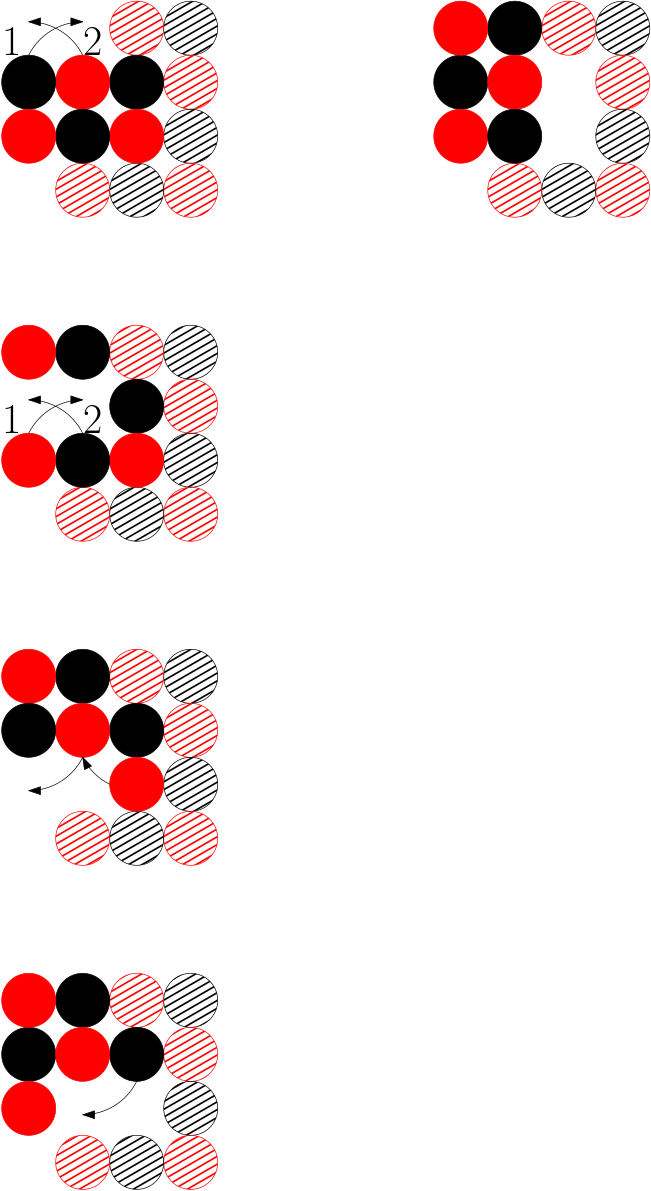}
\centering
\includegraphics[width=0.45\linewidth]{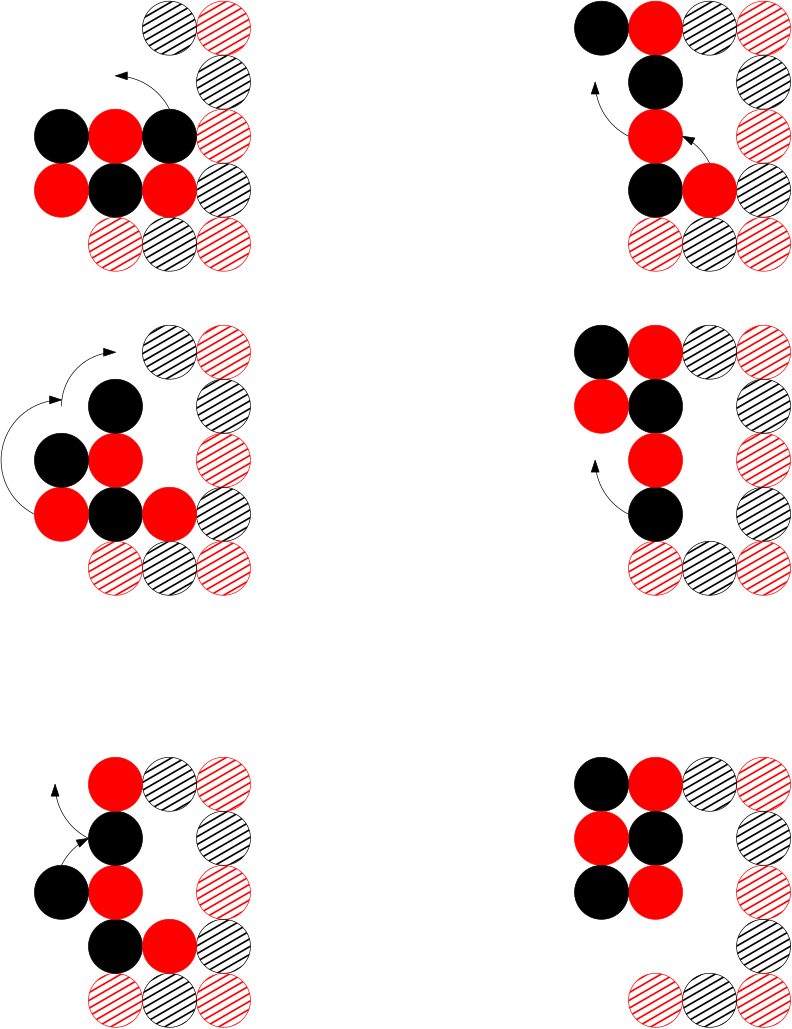}
\centering
\includegraphics[width=0.45\linewidth]{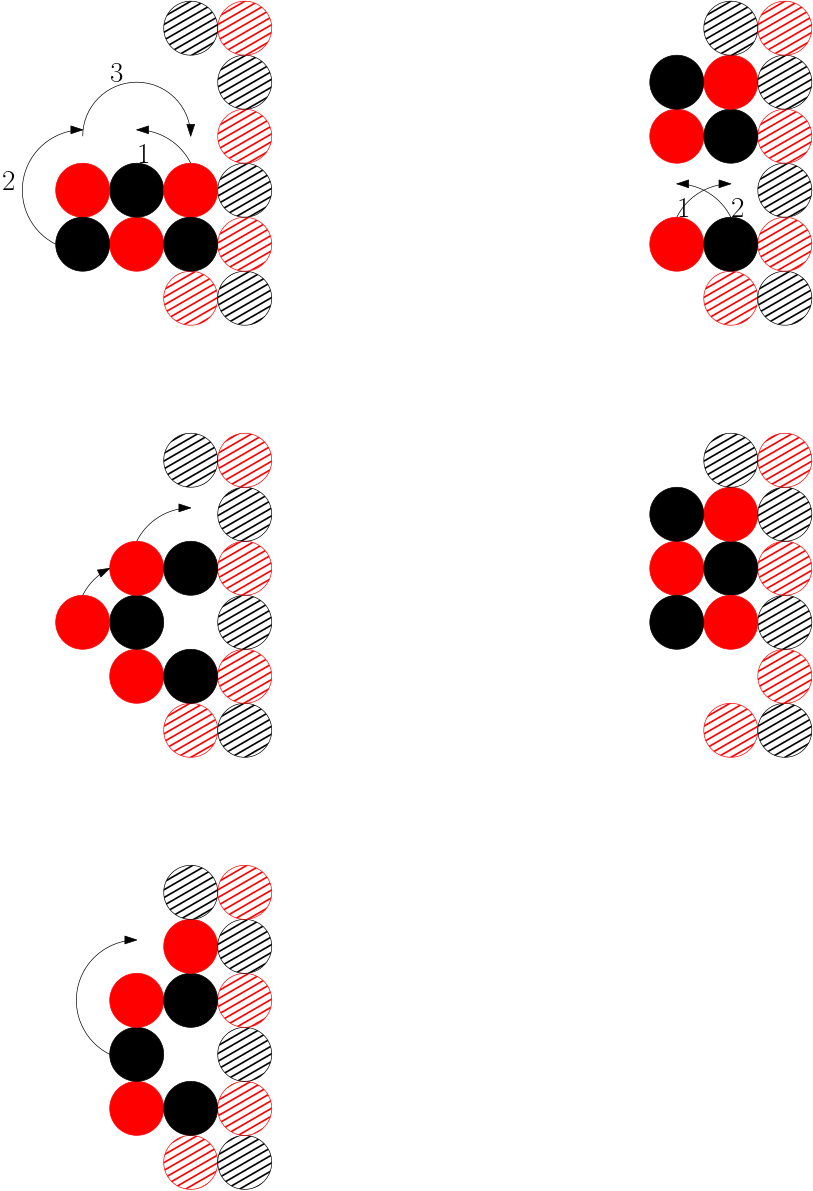}
\caption{Climbing a 6-robot on a line with a black repository, with a gap of size 2, 3 and 4+.}
\label{fig:repo-6node-climb}
\end{figure}

\begin{figure}
\centering
\includegraphics[width=0.45\linewidth]{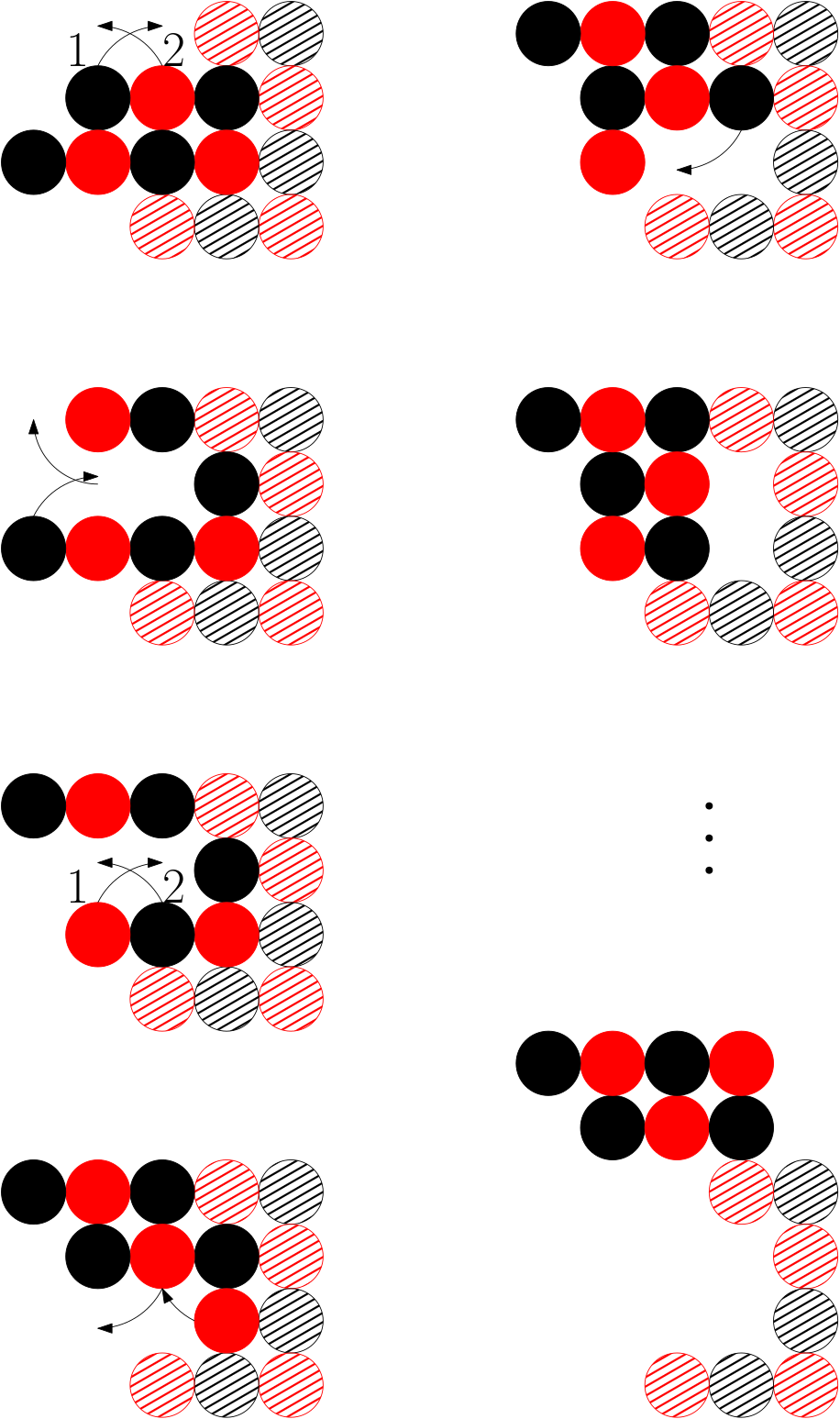}
\centering
\includegraphics[width=0.45\linewidth]{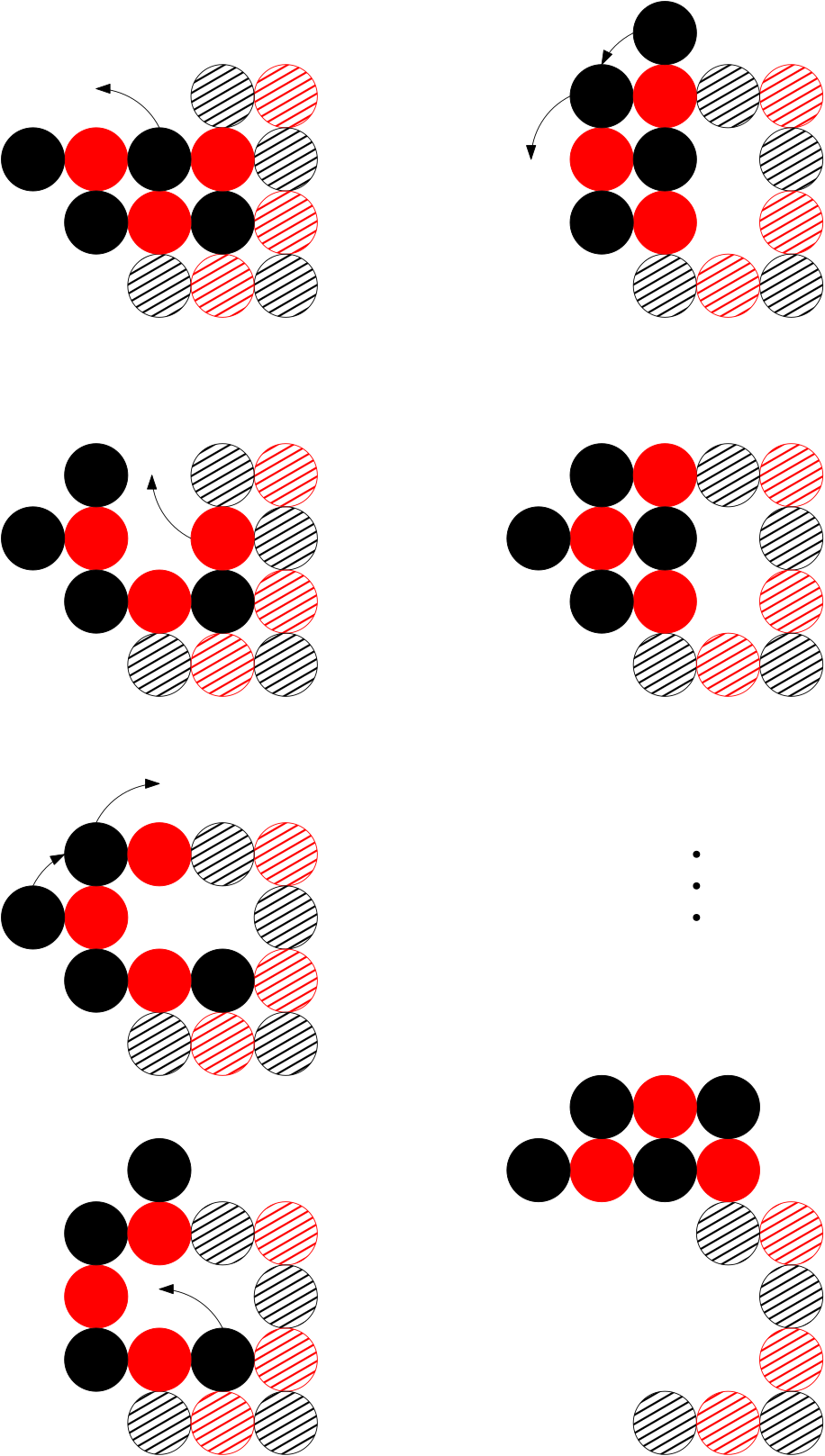}
\caption{Climbing a 7-robot on a line with a black repository with a gap of size 2.}
\label{fig:repo-7node-climb-1}
\end{figure}

\begin{figure}
\centering
\includegraphics[width=0.45\linewidth]{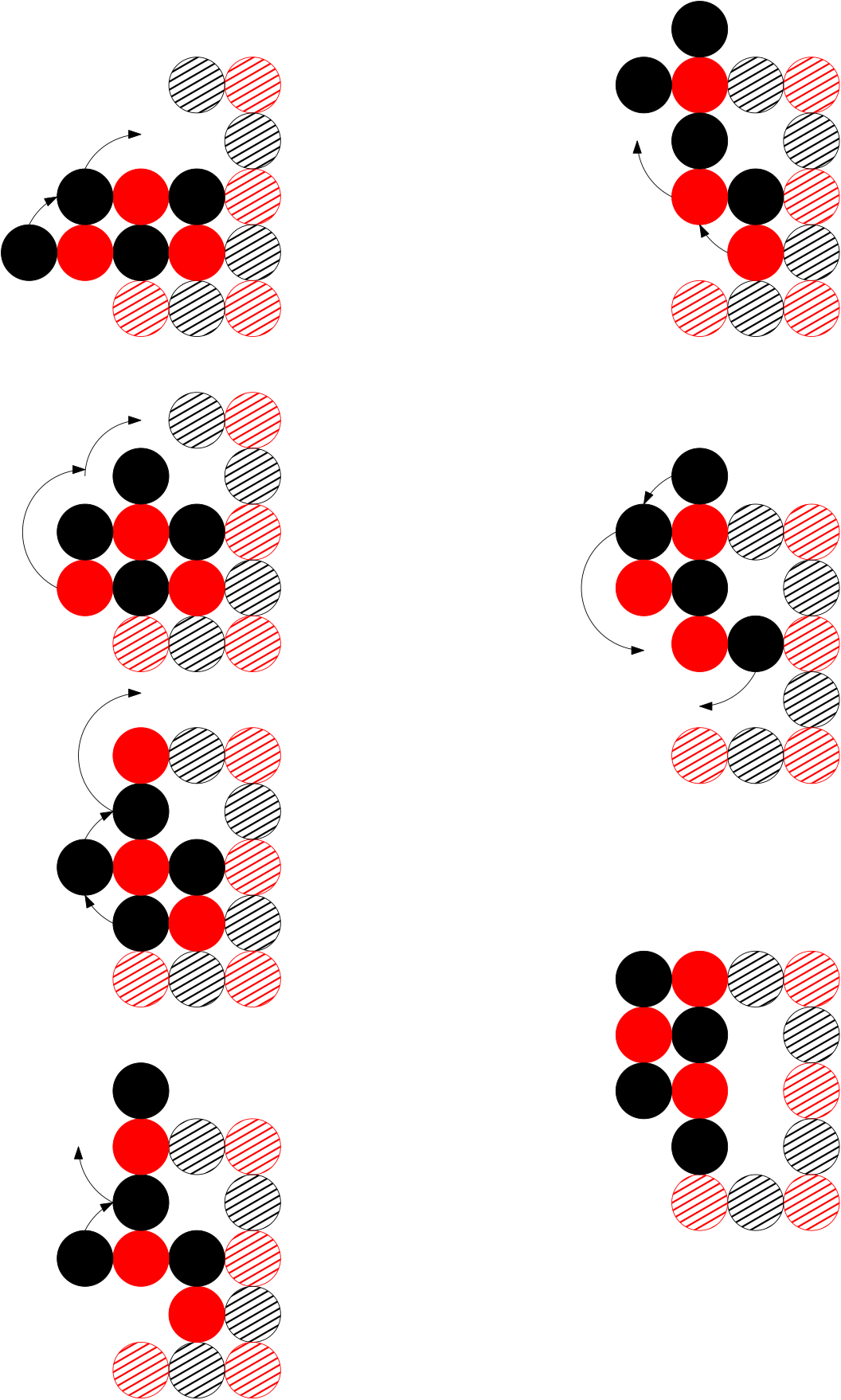}
\centering
\includegraphics[width=0.45\linewidth]{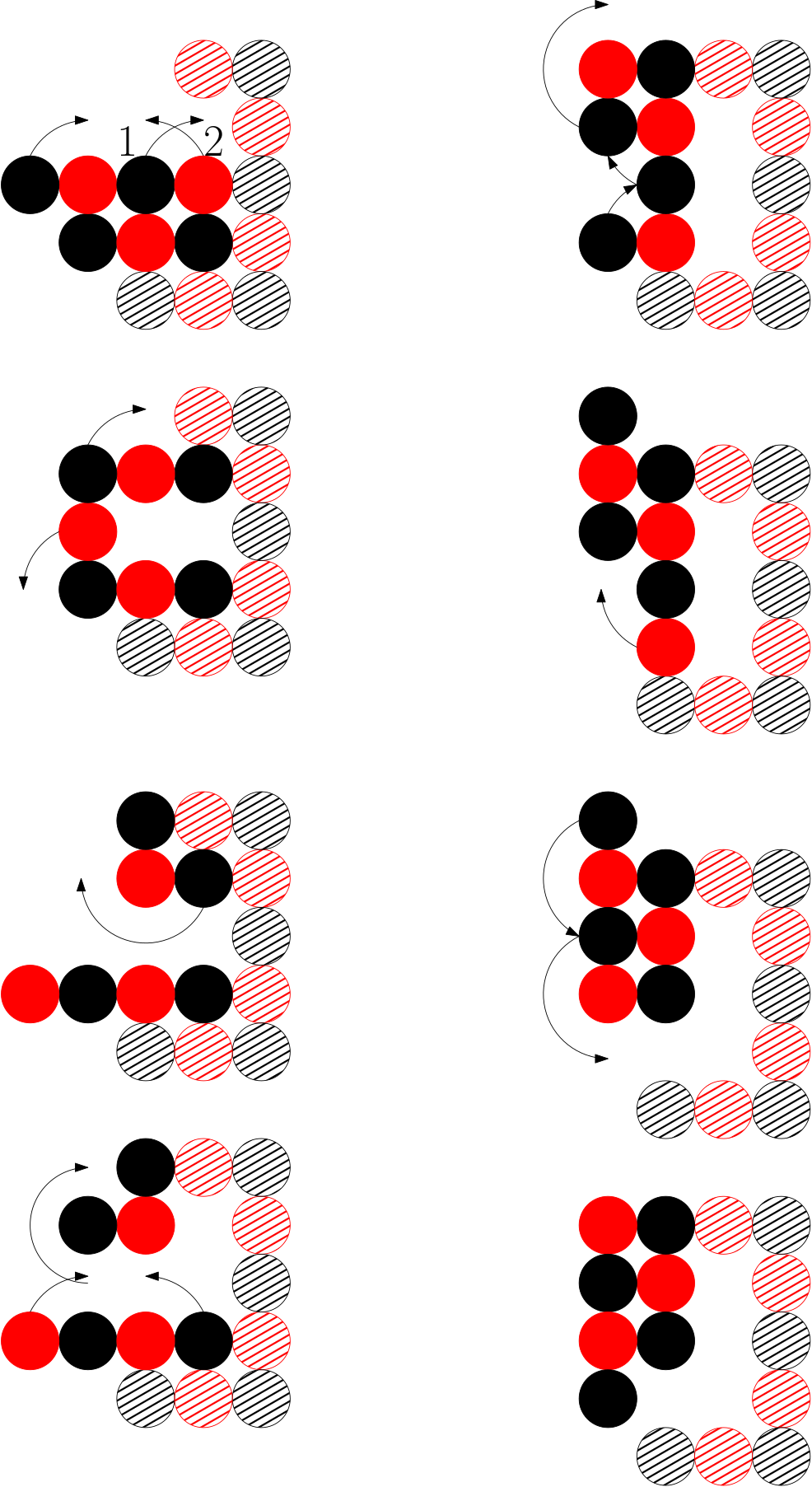}
\caption{Climbing a 7-robot on a line with a black repository with a gap of size 3.}
\label{fig:repo-7node-climb-2}
\end{figure}

\begin{figure}
\centering
\includegraphics[width=0.45\linewidth]{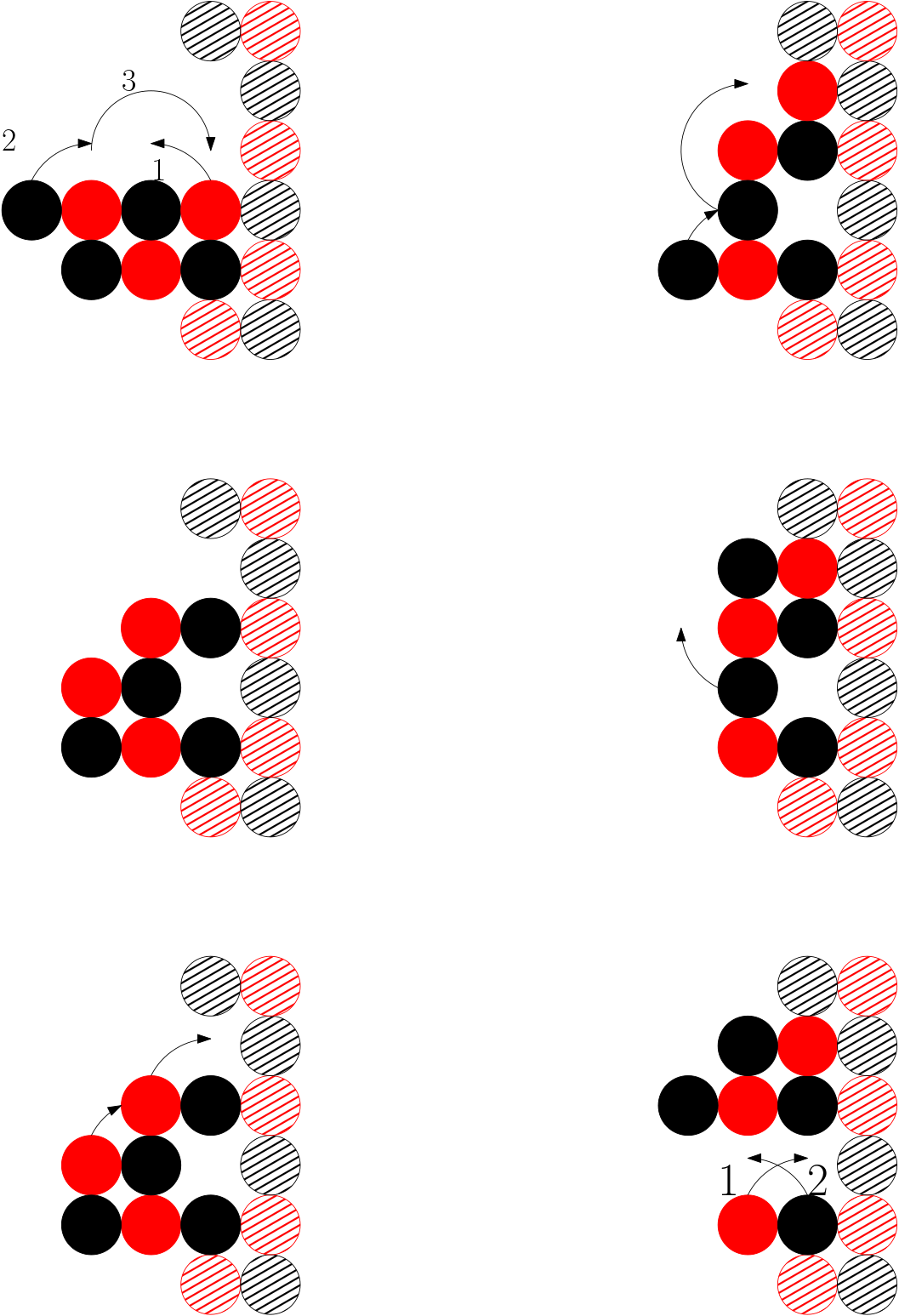}
\centering
\includegraphics[width=0.45\linewidth]{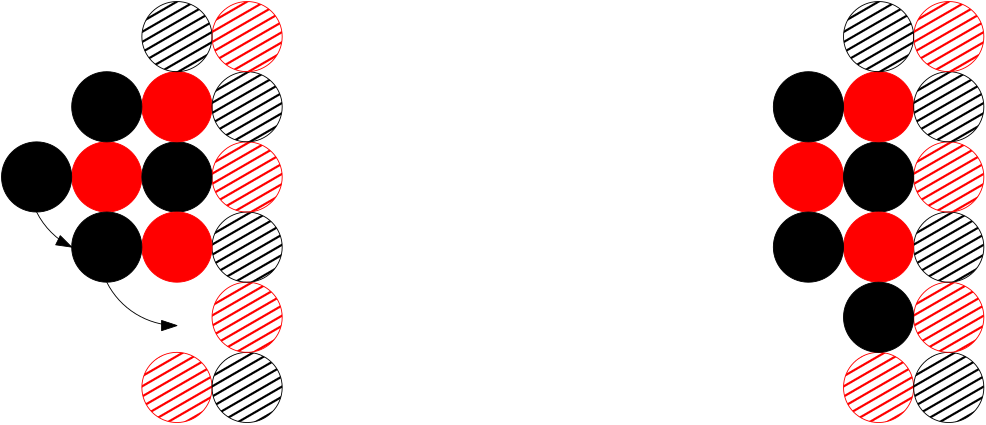}
\centering
\includegraphics[width=0.45\linewidth]{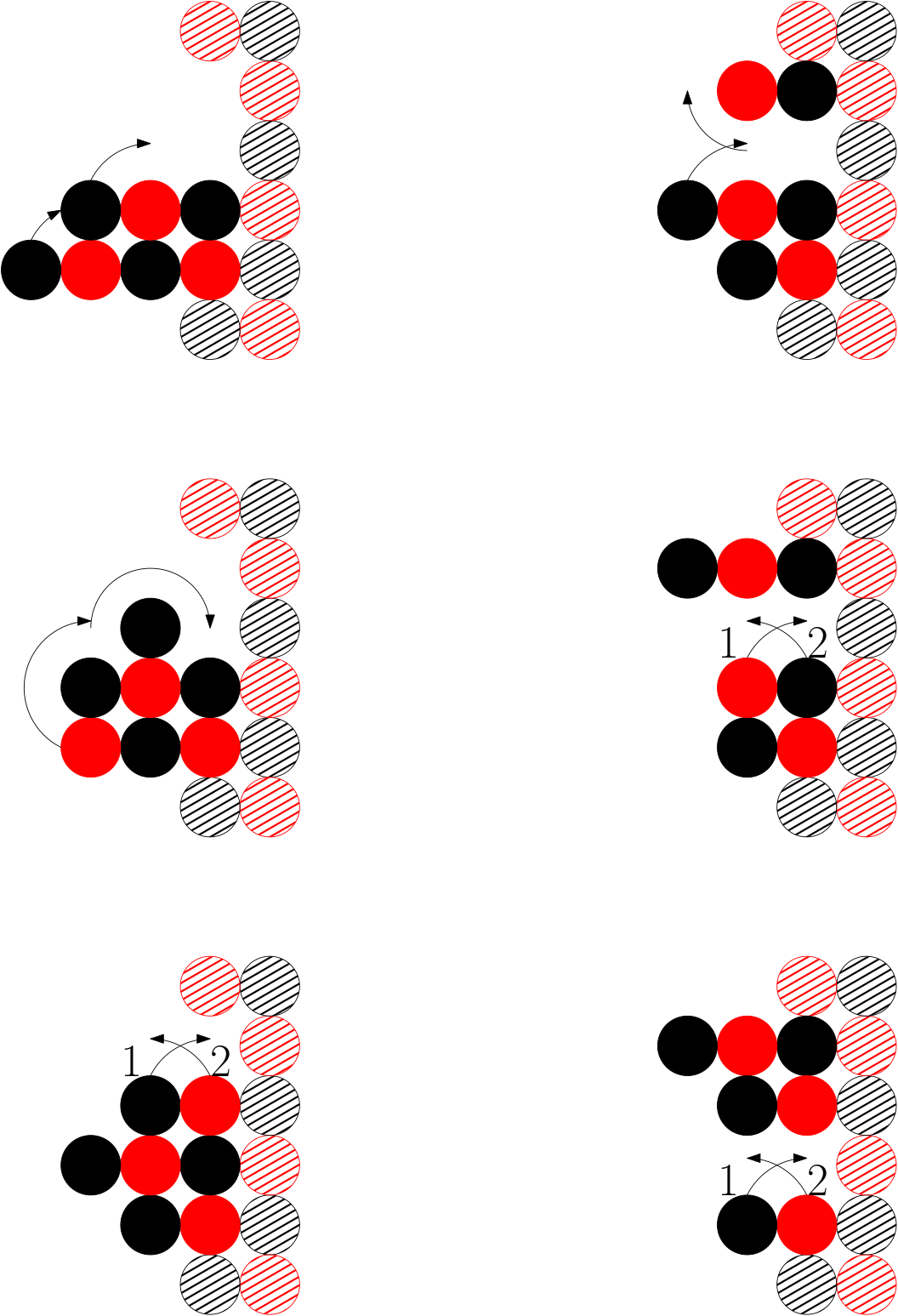}
\centering
\includegraphics[width=0.45\linewidth]{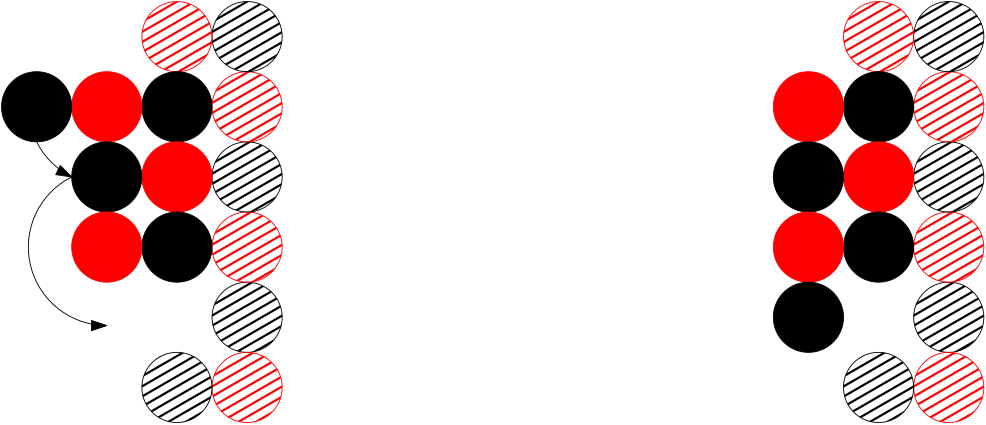}
\caption{Climbing a 7-robot on a line with a black repository with a gap of size 4+.}
\label{fig:repo-7node-climb-3}
\end{figure}

\FloatBarrier

\subsection{Initialisation}

\subsubsection{Robot Generation}

We now prove that we can generate a 6-robot from the orthogonal convex shape $S$ with the help of the 3 musketeers.

\begin{lemma}\label{lem:seed-to-robot}
Let $S$ be a connected orthogonal convex shape. Then there is a connected shape $M$ of 3 nodes (the 3 musketeers) and an attachment of $M$ to the bottom-most row of $S$, such that $S\cup M$ can reach a configuration $S^\prime\cup M^\prime$ satisfying the following properties. $S^\prime=S\setminus\{u_1,u_2,u_{3}\}$, where $\{u_1,u_2,u_{3}\}$ is the $3$-prefix of a row elimination sequence $\sigma$ of $S$ starting from the bottom-most row of $S$. $M^\prime$ is a $6$-robot on the perimeter of $S^\prime$.
\end{lemma}

\begin{proof}
Let $R_i$, $i\geq 1$, be the $i$th row of $S$ counting bottom up. Assume first that $|R_1|\geq 5$ and that the elimination sequence $\sigma$ can start from the rightmost node $(x,y)$ of $R_1$ (the leftmost case is symmetric). If $\sigma$ can continue without switching direction for at least 3 steps, then placing $M$ as a horizontal line at $(x,y-1),(x-1,y-1),(x-2,y-1)$ gives the required 6 robot. If not, then for at least one of the two endpoints, $\sigma$ can make 2 steps before switching, let that endpoint w.l.o.g. be again the rightmost node. Placing $M$ at $(x-2,y-1),(x-3,y-1),(x-4,y-1)$ allows it to lift the two rightmost nodes, become a 5-seed, travel to the other endpoint and lift it, thus becoming a 6-robot.

Next, let $|R_1|\in\{3,4\}$. Then, aligning the 3 nodes of $M$ below the rightmost 3 nodes of $R_1$ immediately gives the required 6-robot.

If on the other hand $|R_1|=2$ or $|R_1|=1$, then $M$ can be placed so that a 5-seed or a $4$-seed, respectively, is attached to the bottom of row $R_2$. If it is a 5-seed then it can reach the rightmost/leftmost endpoint of $R_2$ and lift that node, thus becoming a 6-robot. If it is a 4-seed then if $|R_2|\geq 2$ it can reach the rightmost/leftmost endpoint of $R_2$ and lift two nodes, possibly one node from each endpoint, thus becoming a 6-robot. The only remaining case is when a 4-seed is attached to $R_2$ and $|R_2|=1$. In that case, the configuration of the 4-seed and the single node of $R_2$ can be transformed into a $5$-seed attached to the bottom of row $R_3$, from which the previous case can again be applied.
\qed
\end{proof}

\begin{figure}
\centering
\includegraphics[width=0.45\textwidth]{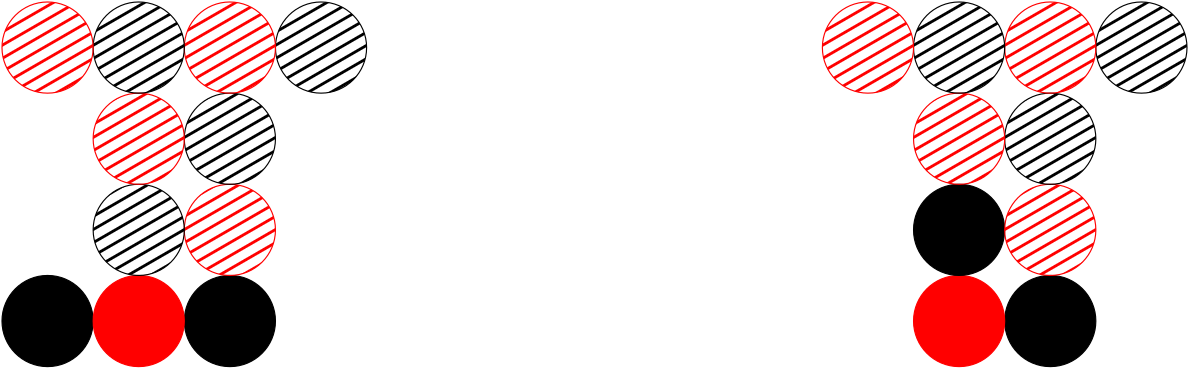}
\caption{Some seed placements. The striped circles represent the orthogonal convex shape $S$.}
\label{fig:seed-placements}
\end{figure}
\FloatBarrier

\subsubsection{Prefix Construction}

The \emph{x-gradient} and \emph{y-gradient} of two neighbouring nodes is the difference in the $x$ and $y$ co-ordinates of the two nodes, respectively. A \emph{parity rhombus} is a shape where every line is of odd length and the same colour, and the x-gradient and y-gradient of every node on the end of every line is at most $1$.

To construct the extended staircase from an orthogonal convex shape $S$, we must first retrieve a sequence of $3$ nodes $u_1, u_2, u_3$ from $S$, where $u_3$ is black. We assume w.l.o.g. that $S$ is a black parity shape. We now show with the following 4 lemmas that this is possible, even in the edge case where $S$ is a parity rhombus.

\begin{lemma}\label{lem:extract-rhombus}
For any shape $S \cup M$, where $S$ is a black parity rhombus of $n$ nodes divided into $p$ rows, $R_1, R_2, \ldots, R_p$ and $M$ is a 6-robot, it is possible for $M$ to extract two black nodes and a red node $u_1, u_2, u_3$ from $S$.
\end{lemma}

\begin{proof}
We extract these nodes by following the procedure of Figure \ref{fig:single-black}. This example is for a rhombus with $5$ rows but, by Theorem \ref{the:6-robot-traverse} and Theorem \ref{the:7-robot-traverse}, additional rows can be navigated and so do not fundamentally change the procedure.
\end{proof}

\begin{figure}
\centering
\begin{subfigure}{.5\textwidth}
\centering
\includegraphics[width=0.75\linewidth]{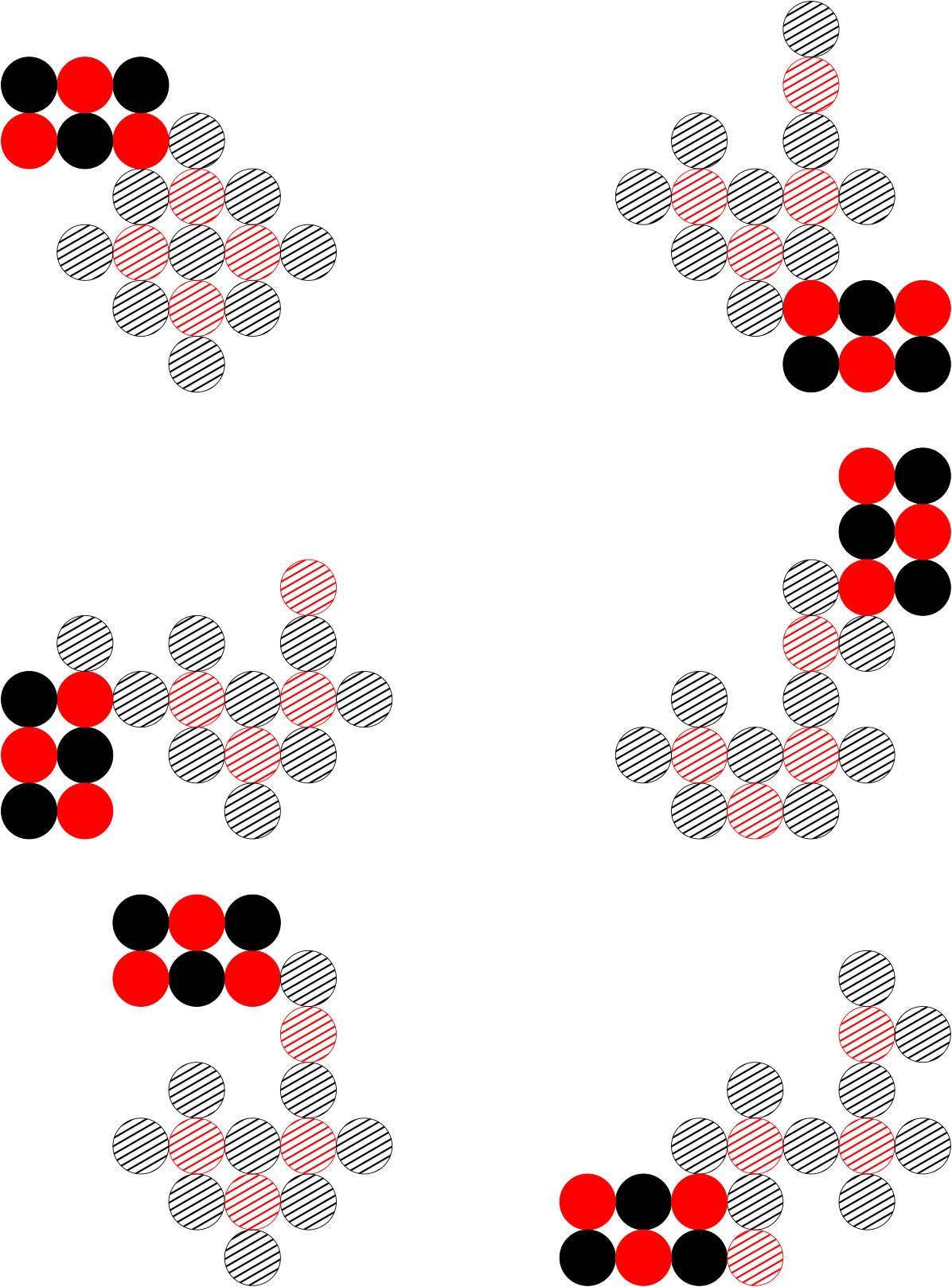}
\end{subfigure}%
\begin{subfigure}{.5\textwidth}
\centering
\includegraphics[width=0.75\linewidth]{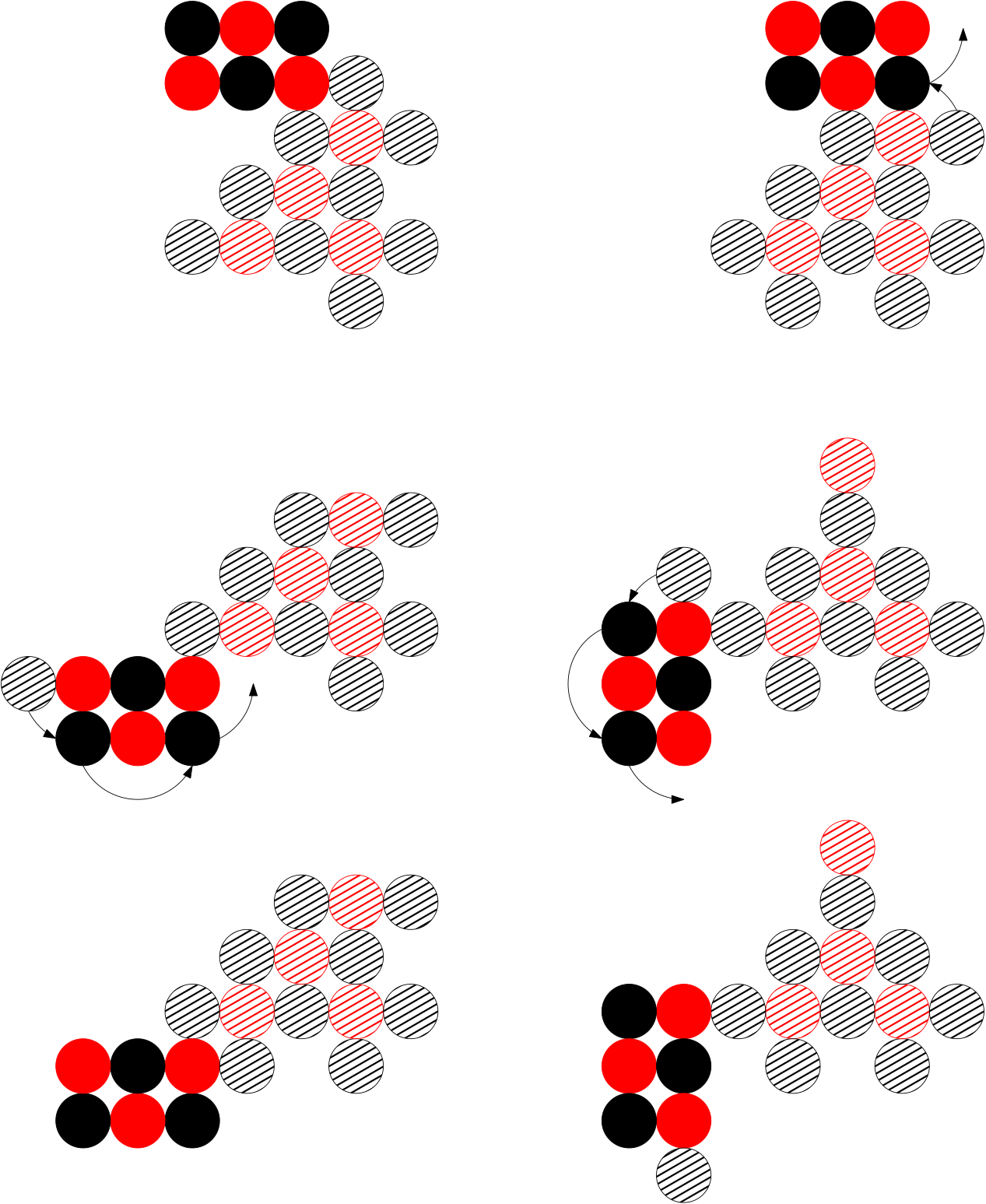}
\end{subfigure}
\begin{subfigure}{.5\textwidth}
\vspace{1cm}
\centering
\includegraphics[width=0.75\linewidth]{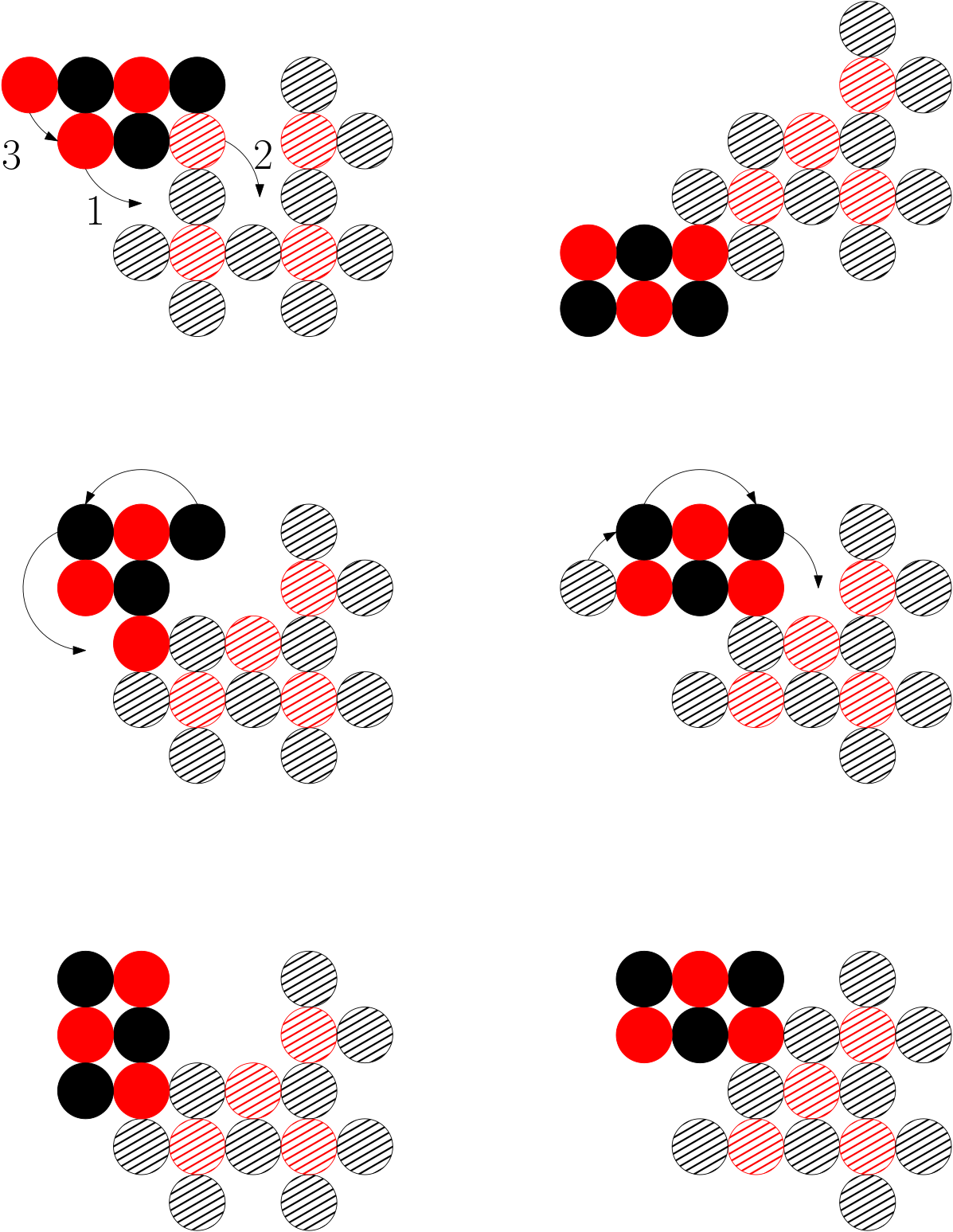}
\end{subfigure}%
\begin{subfigure}{.5\textwidth}
\centering
\includegraphics[width=0.75\linewidth]{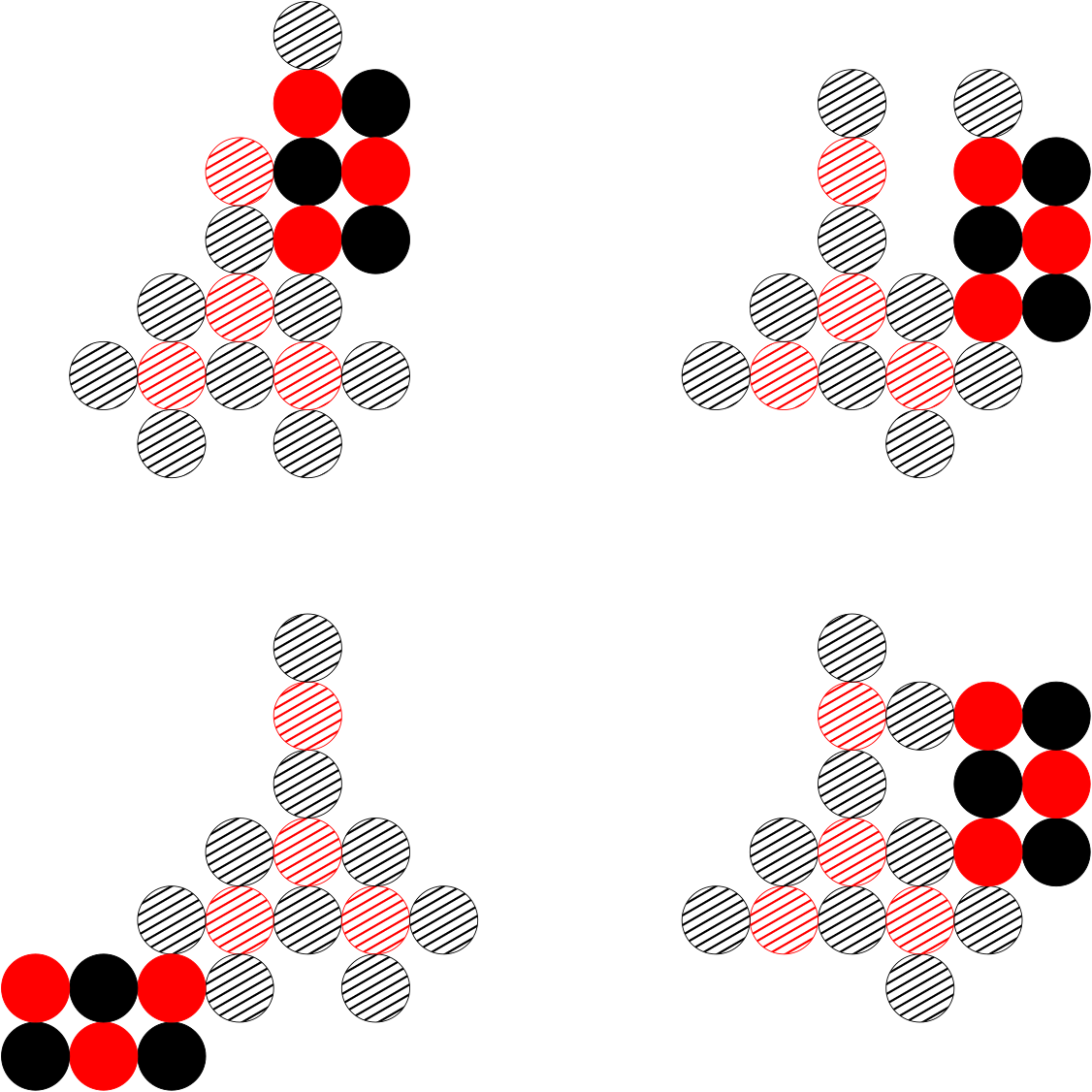}
\end{subfigure}
\caption{Converting a black parity rhombus.}
\label{fig:single-black}
\end{figure}

An orthogonal convex shape $S$ divided into $p$ rows, $R_1, R_2, \ldots, R_p$ is \emph{line-like} if the first node in $R_i$ is above the last node in $R_{i - 1}$ for all $0 < i \leq p$.

A line $l$ \emph{blocks} an empty cell $c$ in an orthogonal convex shape $S$ if there is a node in $l$ such that adding a node to $c$ would cause $S$ to lose orthogonal convexity.

\begin{lemma}\label{lem:extract-2-nodes}
For any shape $S \cup R$ where $S$ is a non-red parity connected orthogonal convex shape of $n$ nodes divided into $p$ rows, $R_1, R_2, \ldots, R_p$ and $R$ is a 6-robot, it is possible for $R$ to extract a bicolour pair of nodes $u, v$ from $S$, where the resulting shape $S' = S \setminus \{ u, v\}$ is a connected orthogonal convex shape.
\end{lemma}

\begin{proof}
We divide our proof into cases. In the first case, $S$ is line-like and the rows $R_1$ and $R_p$ have black nodes which can be extracted. In this case, we can extract the node and then the following node, which is necessarily of the opposite colour. If only one of $R_1$ and $R_p$ has a black node which can be extracted, we can rotate $S$ by $180\degree$ to ensure that the row containing the black does not contain the anchor node and then extract from it. If neither row has a black node which can be extracted, then the line-like shape has a red parity, which violates our assumption of a black parity shape.

If $S$ is not line-like, then we consider the row $R_p$. If $R_p$ is of length $\geq 4$, then we can extract two nodes from $R_p$ without breaking connectivity by extracting from the side furthest from the point where $R_{p - 1}$ connects to $R_p$. If $R_p$ is of length 3 and $R_{p - 1}$ is of length $\geq 2$, then we can extract two nodes from $R_p$, leaving one connected to $R_{p - 1}$. If $R_p$ is of length 3 and $R_{p - 1}$ is of length 1, then we can extract 2 nodes unless $R_{p - 1}$ is connected to the middle node of $R_p$. In this case, $R_{p - 2}$ to $R_1$ must be a single node and we can extract from the other end of the line starting with $R_1$ as the existence of a column after $R_{p - 1}$ would violate convexity. If $R_p$ is of length 2, then we can extract $R_p$. If $R_p$ and $R_{p - 1}$ are of length 1 each, then we can extract them. This leaves the cases where $R_p$ is of length 1 and $R_{p - 1}$ is of length $\geq 1$. If $R_{p - 1}$ is of length $\geq 4$ then we can move $R_p$ if necessary and extract two nodes from $R_{p - 1}$. If $R_{p - 1}$ is of length $2$ then we can rotate $R_p$ into $R_{p - 1}$ and the resulting situation is equivalent to $R_p$ of length 3.

If $S$ has a single red node in $R_p$ and $R_{p - 1}$ is of length $3$, then for $S$ to have a black parity $\geq 0$, there must be some row $R_i$ which is an odd length black line. We can move the node in $R_p$ to either end of this line, unless there is a row $R_j$ which blocks this, by extending further than the ends of $R_i$. If $R_j = R_{i + 1}$ or $R_j = R_{i - 1}$, then $R_j$ cannot block $R_i$. If $R_j$ is an odd length black line or even length line, we can place $R_p$ on it. If $R_j$ is an odd length red line, then to maintain parity, there must be another odd length black line. If there is a line $R_k$ which neighbours $R_i$ and has a greater length than it, then $R_j$ cannot block $R_i$. Therefore, given $b$ odd length black lines, for every line to be blocked there must be at least $r = b + 1$ odd length red lines, one long line to block the nodes and $b$ lines to connect them together into a shape. Including $R_p$, such a shape would have a red parity of at least $2$, and is therefore impossible via assumption. We can therefore move $R_p$, creating a new situation where extraction of two nodes from $R_{p - 1}$ is possible.

If $S$ has a single black node in $R_p$ and $R_{p - 1}$ is of length $3$, then if $R_{p - 2}$ is of length $2$ we can move the node in $R_p$ to it and extract from $R_{p - 1}$. If $R_{p - 2}$ is also of a length $\geq 3$ then unless $S$ is a black parity rhombus it is possible for the 6-robot to extract the node in $R_p$ and move away from $R_{p - 1}$. After that, the red node from $R_{p - 1}$ can be rotated. The robot can then store the black node it is carrying on the red node by moving around the perimeter of the shape, which is still orthogonal convex. Finally, the 6-robot can extract both the black node and the red node. If $S$ is a black parity rhombus then by Lemma \ref{lem:extract-rhombus} we can extract two nodes from it using special movements.
\qed
\end{proof}

\begin{lemma}\label{lem:extract-black-node}
For any shape $S \cup M$, where $S$ is a non-red parity connected orthogonal convex shape of $n$ nodes divided into $p$ rows, $R_1, R_2, \ldots, R_p$ and $M$ is a 6-robot, it is possible for $M$ to extract a black node $u$ from $S$, where the resulting shape $S' = S \setminus \{ u\}$ is a connected orthogonal convex shape.
\end{lemma}

\begin{proof}
We consider two cases for when the black nodes are the majority, and when they are exactly $n/2$.

In the first case, since $S$ is majority black by assumption, there must be at least one row $R_i$ which is a single black node or a line of odd length which begins and ends with a black node. We can extract a black node from this line, unless such an extraction will violate orthogonal convexity. This occurs if the neighbouring lines $R_{i-1}$ and $R_{i+1}$ are of the same length. For $S$ to have black parity, this implies the existence of another odd length line from which a black node can be extracted. More generally, given $x$ odd length lines which end in reds, by the pigeonhole principle there must be $x + 1$ odd length lines ending in blacks, implying at least one can be extracted from without violating orthogonal convexity. Moreover, the removal cannot break connectivity, because this would imply $R_{i-1}$ and $R_{i+1}$ are single red nodes on either end of $R_i$, which also implies the existence of another odd length line to maintain the ratio of black to red nodes.

In the second case, we try to extract a black node from the bottom row $R_p$. If this is not possible, it implies either that $R_p$ is an odd length line with red parity or $R_p$ is an even length line and $R_{p-1}$ is a single red node connected to the black node at the end of $R_p$. In either case, the existence of a row which is of red parity implies the existence of an odd length line of black parity to maintain the overall parity of the shape. A similar pigeonhole principle argument to the first case follows, both for orthogonal convexity and connectivity.

There is a special case where the black node to be extracted happens to be the anchor node. In this scenario, we simply rotate the shape $180\degree$, giving us an equivalent scenario where the black node is in $R_p$ and thus guaranteed to be accessible.
\qed
\end{proof}

\begin{lemma}
For any shape $S \cup M$, where $S$ is a non-red parity connected orthogonal convex shape of $n$ nodes divided into $p$ rows, $R_1, R_2, \ldots, R_p$ and $M$ is a 6-robot, it is possible for $M$ to extract a sequence of nodes $(u_1, u_2, u_3)$ from $S$, where $u_1, u_2$ is a bicolour pair, $u_3$ is black, and $S\setminus\{u_1,u_2,u_3\}$ is a connected orthogonal convex shape.
\end{lemma}

\begin{proof}
By Lemma \ref{lem:extract-2-nodes}, we can extract two nodes from a $S$. The 6-robot places these nodes on the anchor node. By Lemma \ref{lem:extract-black-node} we can then extract a black node from $S$. The 6-robot places this node as well.
\end{proof}

\subsection{Transformations between Shapes}

In this section, we show that, given our previous results, we are now in the position to convert an orthogonal convex shape into another such shape. We begin with the conversion of an extended staircase into a diagonal line-with-leaves, then the orthogonal convex shape to the diagonal line-with-leaves, and then our main result follows by reversibility.

\subsubsection{Transforming $S$ to Extended Staircase}

\begin{lemma}\label{lem:build-ext-staircase}
Let $S$ be a connected orthogonal convex shape with $n$ nodes divided into $p$ rows $R_1, R_2, \ldots, R_p$. Given a row elimination sequence $\sigma=(u_1, u_2,$ $\ldots, u_{n})$ of $S$, an extended staircase generation sequence $\sigma'=(u'_1, u'_2,$ $\ldots, u'_{n})$ which is colour-order preserving w.r.t. $\sigma$, and a 6-robot placed on the external surface of $S$, for all $1 \leq i < n$ the 6-robot is capable of picking up the node $u_i$, moving as a 7-robot to the empty cell $u'_i$ and placing it, and then returning as a 6-robot to $u_{i + 1}$.
\end{lemma}

\begin{proof}
We follow the procedure of Algorithm \ref{alg:extended-staircase-build-alg}. By Theorem \ref{the:6-robot-traverse} and Theorem \ref{the:7-robot-traverse}, the 6-robot $R$ and 7-robot $R \cup u_i$ can climb and slide around the external surface of $S$. We use this to move to each $u_i$, extract it, move to the cell for $u'_i$ and then place a node of the same colour as $u_i$ in it, substituting $u_i$ for a node in $R$ as necessary to create new 6-robot. By Lemma \ref{lem:shape-staircase-properties}, so long as we approach $T_i$ from $R_p$, we can climb onto and off $T_i$ to place the nodes using the same movements as the previous theorems. By Lemma \ref{lem:black-repository}, placing a black node in the repository cell does not inhibit movement.
\qed
\end{proof}

\subsubsection{Transforming Extended Staircase to Diagonal-Line-with-Leaves}

\begin{lemma}\label{lem:build-dll}
Let $W \cup T \cup R$ be the union of the $Stairs$ of an extended staircase $W$, $T \subseteq \{BRep \cup RRep\}$ from the extended staircase and a 6-node robot $R$ on the cell perimeter of $S \cup T$. Given a shape elimination sequence $\sigma=(u_1, u_2, \ldots, u_n)$ of $T$, a diagonal line-with-leaves generation sequence $\sigma'$ which is colour-order preserving w.r.t. $\sigma$ and a 6-robot placed on the external surface of $S$, for all $1 \leq i \leq n$ the 6-robot is capable of picking up the node $u_i$, moving as a 7-robot to $u'_i$ and placing it, and then returning as a 6-robot to $u_{i + 1}$.
\end{lemma}

\begin{proof}
We follow the procedure of Algorithm \ref{alg:dll-build-alg}. By Theorem \ref{the:6-robot-traverse} and Theorem \ref{the:7-robot-traverse}, the 6-robot $R$ and 7-robot $R \cup u_i$ can climb and slide around the external surface of $S \cup T$. We use this to move to each $u_i$, extract it, move to the cell for $u'_i$ and then place a node of the same colour as $u_i$ in it, substituting $u_i$ for a node in $R$ as necessary to create new 6-robot. Since the placement of $u'_i$ is extending $Stairs$, the resulting shape is always orthogonal convex for all $1 \leq i \leq n$.
\qed
\end{proof}

\subsubsection{Transforming $S$ to Diagonal-Line-with-Leaves}

\begin{lemma}\label{lem:convex-to-dll}
Let $S$ be a connected orthogonal convex shape. Then there is a connected shape $M$ of $3$ nodes (the $3$ musketeers) and an attachment of $M$ to the bottom-most row of $S$, such that $S \cup M$ can reach the configuration $D$, where $D$ is a diagonal line-with-leaves which is colour-consistent with $S$.
\end{lemma}

\begin{proof}
We follow the procedure of Algorithm \ref{alg:high-level-alg}. By Lemma \ref{lem:seed-to-robot} we can form a 6-robot from $S \cup M$. By Lemma \ref{lem:build-ext-staircase}, we can build an extended staircase from the resulting shape. By Lemma \ref{lem:build-dll}, we can then build a diagonal line-with-leaves. Finally, by reversibility, we can place $R$ such that the removal of 3 nodes leaves a larger diagonal line-with-leaves $D$ which is colour-consistent with $S$.
\qed
\end{proof}

\subsection{Time Analysis and Wrapping Up}

\begin{lemma}\label{lem:time-bound}
There exists a connected orthogonal convex shape of $n$ nodes $S$ and a diagonal line-with-leaves $T$ and such that any strategy which transforms $S$ into $T$ requires $O(n^2)$ time steps in the worst case.
\end{lemma}

\begin{proof}
To construct $T$, we must transfer nodes using the robot to the anchor node. In the worst case, $S$ is a staircase, and the robot must move nodes from one end to the other. It must therefore make $O(cn^2)$ moves, where $c$ is the maximum number of rotations needed for the robot to move one step. When the extended staircase has been constructed, it must be converted into a diagonal line-with-leaves. In the worst case every column in the staircase has $4$ nodes, and the robot must extend $Stairs$ until one repository has a single node. Therefore, the robot must make $O(2cn^2)$ moves to travel on both sides of $Stairs$. Combining the worst cases of both procedures therefore takes $O(3cn^2) = O(n^2)$ time steps.
\end{proof}

\begin{proposition}\label{prop:unique-dll}
For any two connected orthogonal convex shapes $S$ and $T$ which are colour-consistent, Algorithm \ref{alg:high-level-alg} generates the diagonal line-with-leaves $D$ and $G$ such that $D = G$.
\end{proposition}

\begin{theorem}\label{the:dll-to-dll}
Let $S$ and $S'$ be connected colour-consistent orthogonal convex shapes. Then there is a connected shape $M$ of $3$ nodes (the $3$ musketeers) and an attachment of $M$ to the bottom-most row of $S$, such that $S \cup M$ can reach the configuration $S'$ in $O(n^2)$ time steps.
\end{theorem}

\begin{proof}
By Lemma \ref{lem:convex-to-dll}, we can convert $S$ into a diagonal line-with-leaves $T$. By reversibility, we can convert $T$ into $S'$. By Lemma \ref{lem:time-bound}, this procedure takes $O(n^2)$ time steps.
\qed
\end{proof}

\section{Conclusions}\label{sec6}

There are some open problems which follow from the findings of our work. The most  obvious is expanding the class of shapes which can be constructed to achieve universal transformation. An example of a bad case is the ``double spiral'', which is a line forming two connected spirals. In this case, preserving connectivity after the removal of a node requires the robot to get to the centre of a spiral, which may not be possible without a special procedure to ``dig'' into it without breaking connectivity. Finally, successfully switching to a decentralised model of transformations will greatly expand the utility of the results, especially because most programmable matter systems which model real-world applications implement programs in this way. This in turn could lead to real-world applications for the efficient transformation of programmable matter systems.

\begin{figure}
\centering
\includegraphics[width=0.45\textwidth]{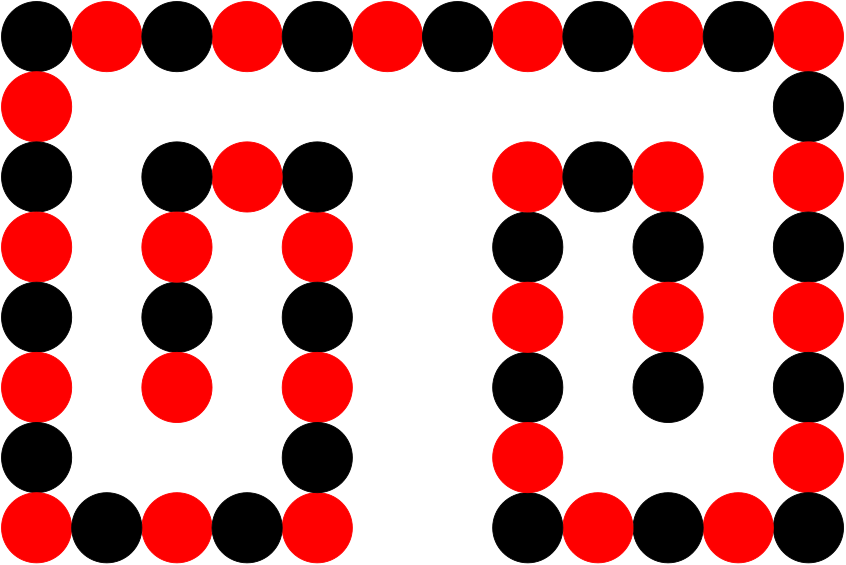}
\caption{An example of the double spiral shape.}
\label{fig:seed-placements}
\end{figure}
\FloatBarrier

\bibliographystyle{splncs04}
\bibliography{samplepaper}

\end{document}